\documentclass[aps,superscriptaddress,twocolumn,showpacs]{revtex4-2}

\usepackage{amsmath}
\usepackage{latexsym}
\usepackage{amssymb}
\usepackage{graphicx}
\usepackage[colorlinks=true, citecolor=blue, urlcolor=blue]{hyperref}
\usepackage{float}
\usepackage{amsfonts}
\usepackage{textcomp}
\usepackage{mathpazo}
\usepackage{comment}

\usepackage{bbm}

\usepackage{xcolor}
\definecolor{myurlcolor}{rgb}{0,0,0.4}
\definecolor{mycitecolor}{rgb}{0,0.5,0}
\definecolor{myrefcolor}{rgb}{0.5,0,0}
\usepackage{hyperref}
\hypersetup{colorlinks,
linkcolor=myrefcolor,
citecolor=mycitecolor,
urlcolor=myurlcolor}

\usepackage[draft]{fixme}
\usepackage{amsmath,bbm}
\usepackage{scrextend}
\usepackage{graphicx}
\usepackage{amsfonts}
\usepackage{amssymb}
\usepackage{amsmath, amssymb, amsthm,verbatim,graphicx,bbm}
\usepackage{mathrsfs}
\usepackage{color,xcolor,longtable}

\def\be{\begin{equation}}
\def\ee{\end{equation}}
\def\ben{\begin{eqnarray}}
\def\een{\end{eqnarray}}
\def\eea{\end{array}}
\def\bea{\begin{array}}

\newcommand{\Tr}[1]{\mathrm{Tr}#1}
\newcommand{\bei}{\begin{itemize}}
\newcommand{\eei}{\end{itemize}}
\newcommand{\ket}[1]{|#1\rangle}
\newcommand{\bra}[1]{\langle#1|}

\newcommand{\proj}[1]{\ket{#1}\!\bra{#1}}

\newcommand{\I}{\mathbbm{1}}

\newcommand{\p}{\Vec{p}}

\renewcommand{\emph}[1]{\textbf{#1}}

\makeatletter
\newtheorem*{rep@theorem}{\rep@title}
\newcommand{\newreptheorem}[2]{%
\newenvironment{rep#1}[1]{%
 \def\rep@title{#2 \ref{##1}}%
 \begin{rep@theorem}}%
 {\end{rep@theorem}}}
\makeatother

\theoremstyle{plain}
\newtheorem{thm}{Theorem}
\newtheorem*{thm*}{Theorem}
\newreptheorem{thm}{Theorem}
\newtheorem{lem}{Lemma}
\newtheorem{fakt}{Fact}

\newtheorem{cor}[thm]{Corollary}

\theoremstyle{definition}

\theoremstyle{remark}

\usepackage[T1]{fontenc}

\begin{document}


\title{Self-testing composite measurements and bound entangled state in a single quantum network}
\author{Shubhayan Sarkar}
\affiliation{Center for Theoretical Physics, Polish Academy of Sciences, Aleja Lotnik\'{o}w 32/46, 02-668 Warsaw, Poland}
\affiliation{Laboratoire d’Information Quantique, Université libre de Bruxelles (ULB), Av. F. D. Roosevelt 50, 1050 Bruxelles, Belgium}

\author{Chandan Datta}
\affiliation{Centre for Quantum Optical Technologies, Centre of New Technologies, University of Warsaw, Banacha 2c, 02-097 Warsaw, Poland}
\affiliation{Institute for Theoretical Physics III, Heinrich Heine University D\"{u}sseldorf, Universit\"{a}tsstra{\ss}e 1, D-40225 Düsseldorf, Germany}
\affiliation{Department of Physics, Indian Institute of Technology Jodhpur, Jodhpur 342030, India}

\author{Saronath Halder}
\affiliation{Centre for Quantum Optical Technologies, Centre of New Technologies, University of Warsaw, Banacha 2c, 02-097 Warsaw, Poland}

\author{Remigiusz Augusiak}
\affiliation{Center for Theoretical Physics, Polish Academy of Sciences, Aleja Lotnik\'{o}w 32/46, 02-668 Warsaw, Poland}

\begin{abstract}	

Within the quantum networks scenario we introduce a single scheme allowing to certify three different types of composite projective measurements acting on a three-qubit Hilbert space: one constructed from genuinely entangled GHZ-like states, one constructed from fully product vectors that exhibit the phenomenon of nonlocality without entanglement (NLWE), and a hybrid measurement obtained from an unextendible product basis (UPB). Noticeably, we  certify a basis exhibiting NLWE in the smallest dimension capable of supporting this phenomenon. On the other hand, the possibility of certification of a measurement obtained from a UPB has an interesting implication that one can also self-test a bound entangled state in the considered quantum network. Such a possibility does not seem to exist in the standard Bell scenario. Furthermore, we also analyse the robustness of our scheme towards experimental errors.
\end{abstract}


\maketitle

\textit{Introduction.---}With the advancement of new technologies such as quantum cryptography \cite{Gisin_crypto}, device-independent (DI) certification of quantum devices is becoming increasingly important, allowing one to certify certain features of an underlying device in a ``black-box'' scenario \cite{DICrypto,Bancal_PhysRevLett.106.250404,Brunner_PhysRevLett.100.210503}, which requires basically no assumptions about the device except that it is governed by quantum theory. The key ingredient for DI certification is Bell nonlocality, i.e., the existence of quantum correlations that cannot be explained by local hidden variable models \cite{Bell,Bell66}. 
The most comprehensive form of DI certification is self-testing \cite{Mayers_selftesting}, which allows for almost complete characterization of the underlying quantum state and the measurements performed on it. From an application standpoint, this type of certification is crucial as it allows to verify whether a quantum device works as expected without knowing its internal mechanism. 

Since its introduction in \cite{Mayers_selftesting}, self-testing has been investigated in various scenarios and shown to have numerous advantages \cite{SupicReview}. However, while much attention has been paid to the self-testing of bipartite and multipartite states  \cite{Yao,Scarani,Reichardt_nature,Mckague_2014,Wu_2014,Bamps,All,chainedBell, Projection,Jed1,prakash,Armin1,sarkar,sarkaro2, Yang}, the problem of certifying quantum measurements, in particular those acting on composite Hilbert spaces, has been 
largely unexplored. Apart from a few classes of local measurements \cite{Armin1, Jed1, random1} and a few composite ones \cite{Marco, JW2, NLWEsupic}, no general scheme exists allowing to certify composite measurements even in the simplest case of multi-qubit Hilbert spaces.

Our aim here is to fill this gap and introduce a unified scheme in the quantum networks scenario that allows for self-testing, in a single experiment, three types of projective measurements: one composed of genuinely entangled states, one constructed from fully product states exhibiting nonlocality without entanglement (NLWE) \cite{Bennett99}, and a hybrid one which is constructed from an unextendible product basis (UPB) \cite{Bennett99-1} and a projector onto the completely entangled subspace orthogonal to the UPB. Notice that UPBs are interesting mathematical objects that have found numerous applications, for instance, in constructing bound entangled states \cite{Bennett99-1} or Bell inequalities with no quantum violation \cite{Augusiak_PhysRevLett.107.070401}.

While a network-based self-testing scheme of a two-qutrit product basis exhibiting NLWE was previously introduced in \cite{NLWEsupic}, our approach enables one to self-test such a basis in the smallest possible dimension capable of supporting the notion of NLWE, which is eight. Furthermore, our scheme utilizes fewer measurements and is thus more efficient as compared to \cite{NLWEsupic}. Finally, it also allows to simultaneously self-test a measurement of a hybrid type, which is constructed from a UPB. An interesting implication of this fact is that one can 
also (indirectly) certify in the network, a mixed bound entangled state 
constructed from the UPB. While a network-based certification scheme for pure states
has already been introduced in Ref. \cite{Allst}, our work seems to be the first to address the question of certification of mixed entangled states (see nevertheless Refs. \cite{subspaces1,subspaces2}).

\textit{Preliminaries.---}Consider a Hilbert space $\mathcal{H}=\mathcal{H}_1\otimes\ldots\otimes\mathcal{H}_N$ and a set of mutually orthogonal fully product states from $\mathcal{H}$,
%
 $   \mathcal{S}=\{|\psi_i^1\rangle\otimes\ldots\otimes|\psi_i^N\rangle\}_{i=1}^k$,
%
where $|\psi_i^m\rangle\in\mathcal{H}_m$ and $k\leq D=\dim\mathcal{H}$. Following Ref. \cite{Bennett99} we say that this set exhibits NLWE if the vectors $|\psi_i\rangle$ cannot be perfectly distinguished by local operations and classical communications (LOCC). An exemplary such set is the following basis of the three-qubit Hilbert space $\mathcal{H}=(\mathbbm{C}^2)^{\otimes 3}$ \cite{Bennett99}:
\begin{eqnarray}\label{NLWEbasism}
\ket{\delta_0}&=&\ket{\overline{0}}\ket{1}\ket{+},\,\, \ket{\delta_1}=\ket{\overline{0}}\ket{1}\ket{-},\,\,
\ket{\delta_2}=\ket{\overline{+}}\ket{0}\ket{1},\nonumber\\ \ket{\delta_3}&=&\ket{\overline{-}}\ket{0}\ket{1},\,\,\,
\ket{\delta_4}=\ket{\overline{1}}\ket{+}\ket{0},\,\, \ket{\delta_5}=\ket{\overline{1}}\ket{-}\ket{0},\nonumber\\
\ket{\delta_6}&=&\ket{\overline{0}}\ket{0}\ket{0},\,\,\,\, \ket{\delta_7}=\ket{\overline{1}}\ket{1}\ket{1},
\end{eqnarray}
where $\ket{i}$, $\ket{\overline{i}}$ $(i=0,1)$ and $\ket{\pm}$, 
$\ket{\overline{\pm}}$ are the eigenvectors of $Z$, $(X+Z)/\sqrt{2}$ and $X$, $(X-Z)/\sqrt{2}$, respectively, where 
$Z$ and $X$ are the Pauli matrices.

While the above set forms a complete basis in the corresponding Hilbert space, there also exist sets $\mathcal{S}$ exhibiting NLWE that do not span the underlying Hilbert space. These are called UPB and were introduced to provide one of the first constructions of bound entangled states, which are entangled states from which no pure entanglement can be distilled \cite{Bennett99}. To be more precise, a collection $\mathcal{S}$ is a UPB if $k<D$, i.e., $\mathcal{S}$ spans a proper subspace $\mathcal{V}$ in $\mathcal{H}$, and the subspace complementary to $\mathcal{V}$ is completely entangled \cite{CES1,CES2}, i.e., contains no fully product vectors. An excellent example of a UPB in $\mathcal{H}=(\mathbbm{C}^2)^{\otimes 3}$ are the following four vectors
\begin{equation}\label{UPBm}
\begin{split}
\ket{\tau_0}&=\ket{\overline{0}}\ket{1}\ket{+},\qquad \ket{\tau_1}=\ket{\overline{+}}\ket{0}\ket{1},\\
\ket{\tau_2}&=\ket{\overline{1}}\ket{+}\ket{0},\qquad \ket{\tau_3}=\ket{\overline{-}}\ket{-}\ket{-}
\end{split}
\end{equation}
which are equivalent under local unitary transformations to the Shifts UPB introduced in \cite{Bennett99-1}. The mixed state $\rho=\Gamma/4$, where 
\begin{eqnarray}\label{UPB1m}
\Gamma&=&\I-\sum_{i=0}^3\proj{\tau_i}
\end{eqnarray}
stands for the projector onto a subspace complementary to the UPB,
is bound entangled; in fact, by the very construction it is entangled and all its partial transpositions are nonnegative \cite{HorodeckiBE}.

The last set of vectors that we consider here
are the following GHZ-like pure states
\begin{equation}\label{GHZvecsm}
\ket{\phi_l}=\frac{1}{\sqrt{2}}(\ket{l_1l_2l_3}+(-1)^{l_{1}}|\overline{l}_1\overline{l}_2\overline{l}_3\rangle),
\end{equation}
where $l\equiv l_1l_2l_3$ with $l_1,l_2,l_3=0,1$ and $\overline{l}_i$ is the negation of the bit $l_i$, i.e., $\overline{l}_i=1-l_i$. It is worth noting that, unlike the previous vectors $\ket{\delta_i}$ or $\ket{\tau_i}$, the GHZ states are all genuinely multipartite entangled.

\textit{Composite measurements.---} From each of the considered sets of vectors one can construct a projective measurement acting on $\mathcal{H}_3=(\mathbbm{C}^2)^{\otimes 3}$. First, the product basis $\ket{\delta_i}$
gives rise to a separable eight-outcome measurement $M_{\mathrm{NLWE}}=\{\proj{\delta_i}\}_{i=0}^7$. This measurement cannot be implemented in terms of the 
LOCC, and, simultaneously, cannot produce entanglement if applied to a state. On the other extreme, we have the eight-outcome measurement 
$M_{\mathrm{GHZ}}=\{\proj{\phi_l}\}_{l=0}^{7}$ 
constructed from the GHZ-like states which are all entangled. Thus, unlike $M_{\mathrm{NLWE}}$, this measurement leads to entangled states for all outcomes $l$.

Let us finally move to the UPB in Eq. \eqref{UPB1m}. It allows constructing a five-outcome hybrid measurement $M_{\mathrm{UPB}}=\{\proj{\tau_i}\}_{i=0}^3\cup\{\Gamma\}$ that lies in between the GHZ measurement and the separable measurement. In fact, four of its outcomes correspond to projections onto fully product states 
$\ket{\tau_i}$, whereas the last outcome is represented by $\Gamma$ that projects onto a completely (but not genuinely) entangled four-dimensional subspace.

\textit{Setting the scenario.---}We consider a quantum network scenario consisting of three external parties Alice, Bob and Charlie and a central party Eve  (see Fig. \ref{fig1}). The scenario also comprises of three independent sources $P_i$ that distribute bipartite quantum states among the parties. We denote these states by $\rho_{s\overline{s}}$ with $s=A,B,C$, where the subsystems $A$, $B$ and $C$ belong to the external parties, whereas the other three systems $\overline{A}$, 
$\overline{B}$ and $\overline{C}$ go to Eve; in what follows we simplify the notation by using $E:=\overline{A}\overline{B}\overline{C}$. On their shares of the joint state $\rho_{ABCE}=\rho_{A\overline{A}}\otimes\rho_{B\overline{B}}\otimes\rho_{C\overline{C}}$, each party can choose to perform one of the measurements $A_x$, $B_y$, $C_z$ and $E_e$, where the measurement choices are labelled $x,y,z,e=0,1,2$. We assume that each of the external party's measurements has two outcomes, denoted $a,b,c=0,1$. The first two Eve's measurements yield eight outputs, whereas the third one results in five outcomes. During the experiment, the parties cannot communicate classically.  

The correlations obtained by repeatedly performing these measurements are captured by a set of probability distributions $\vec{p}=\{p(abcl|xyze)\}$, where each $p(abcl|xyze)$ is the probability of observing outcomes $a$, $b$, $c$ and $l$ by Alice, Bob, Charlie and Eve after performing measurements labeled by $x$, $y$, $z$, and $e$, respectively; it is given by the well-known formula
\begin{equation}\label{probs}
    p(abcl|xyze)=\Tr\left[\rho_{ABCE}N^A_{a|x}\otimes N^B_{b|y}\otimes N^C_{c|z}\otimes N^E_{l|e}\right],
\end{equation}
where $N_{a|x}^A$, $N_{b|y}^B$ etc. are the 
measurement elements representing the measurements of the observers; these are positive semi-definite and satisfy $\sum_{a}N_{a|x}^s=\mathbbm{1}$ for every measurement choice $x$ and every party $s$.

It will be beneficial to use another representation of the observed correlations, that is, in terms of the expectation values of observables of the external parties, which are defined as
\begin{equation}\label{obspic}
    \langle A_{x} B_{y}C_z N^{E}_{l|e} \rangle = \sum_{a,b,c=0,1} (-1)^{a+b+c} p(abcl|xyze).
\end{equation}
Notice that by employing Eq. \eqref{probs} one can express them as $\langle A_{x} B_{y}C_z N^{E}_{l|e} \rangle=\Tr[(A_x \otimes B_{y}\otimes C_z\otimes N^{E}_{l|e})\rho_{ABCE}]$, where $A_x$, $B_y$ and $C_z$ are quantum operators that are defined through the measurement elements as $s_k=N^{s}_{0|k}-N^{s}_{1|k}$, where $s=A,B,C$ and $k=x,y,z$. In the particular case of projective measurements these operators $s_k$ become unitary and thus represent the standard quantum observables. 

\textit{Self-testing.---}
The quantum networks scenario has recently been harnessed to propose self-testing schemes for few quantum measurements defined in composite Hilbert spaces such as the measurement corresponding to the two-qubit Bell basis composed of four maximally entangled vectors \cite{Marco}, or the nine-outcome projective measurement corresponding to a complete basis in $\mathbbm{C}^3\otimes\mathbbm{C}^3$ that exhibits NLWE \cite{NLWEsupic}. Our aim here is to employ these ideas to design a general framework for quantum networks-based device-independent (NDI) certification of various interesting types of quantum measurements, concentrating on the particular case of three-qubit Hilbert spaces.

\begin{figure}[t]
    \centering
    \includegraphics[scale=0.145]{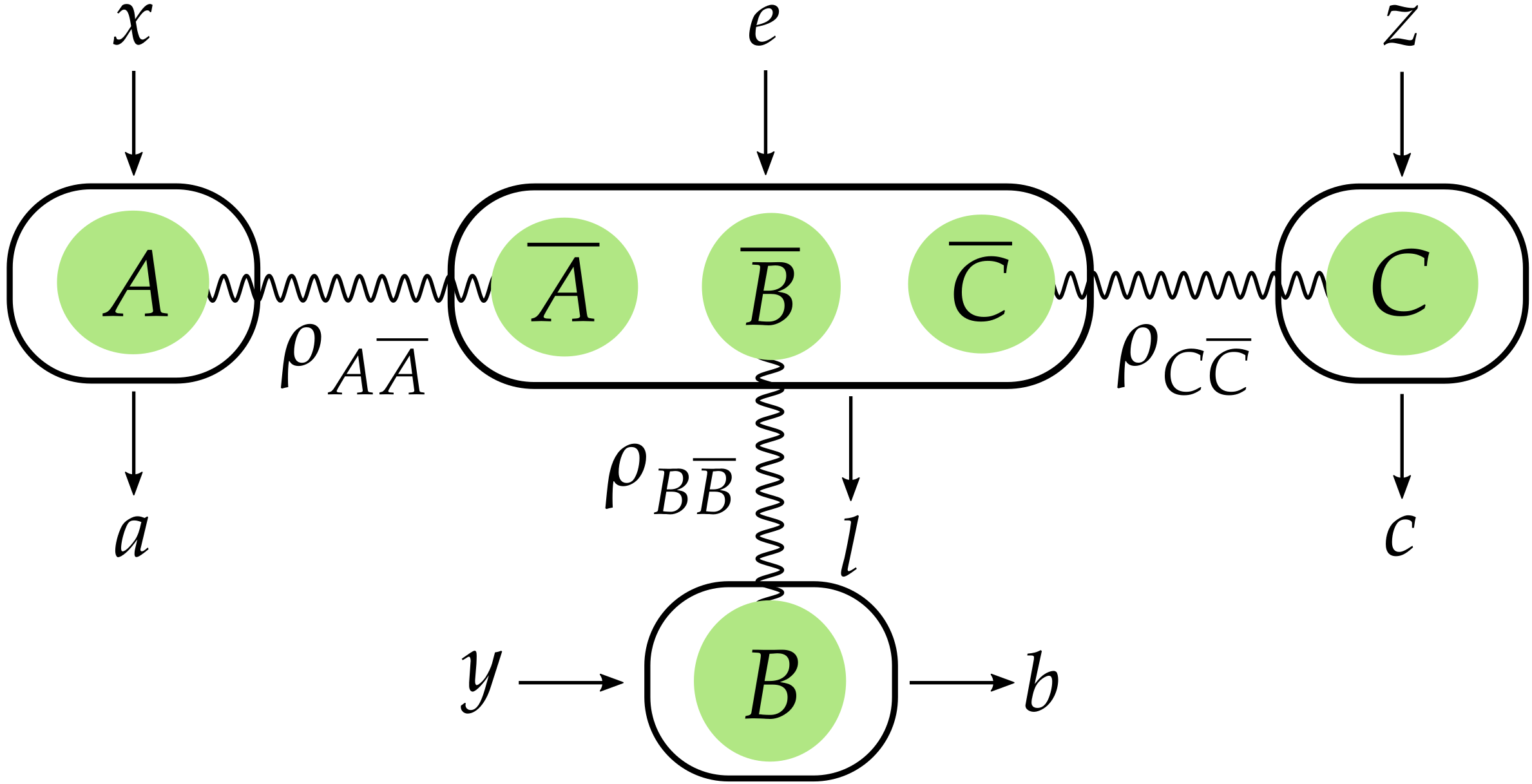}
    \caption{\textbf{Schematic of the considered quantum network scenario.} 
    It consists of four parties $A$, $B$, $C$ and $E$ and three independent sources distributing bipartite quantum states $\rho_{S\overline{S}}$ $(S=A,B,C)$ among the parties as shown on the figure.  The central party $E$ shares quantum states with each of the other parties. Each party performs one of the available measurements on their share of the state obtaining an outcome. The obtained correlations $\{p(abcl|xyze)\}$ are used to certify that each source distributes the maximally entangled state of two qubits and that $E$'s measurements are $M_{\mathrm{GHZ}}$, $M_{\mathrm{UPB}}$ and $M_{\mathrm{NLWE}}$.}
    
    \label{fig1}
\end{figure}

To define the task of self-testing in more precise terms, let us consider again the scenario depicted on Fig. \ref{fig1}, but now we assume that 
both the states $\rho_{s\overline{s}}\in\mathcal{L}(\mathcal{H}_s\otimes\mathcal{H}_{\overline{s}})$ and the measurements performed by the parties are unknown. Due to the fact that the dimensions of the underlying Hilbert spaces $\mathcal{H}_s\otimes\mathcal{H}_{\overline{s}}$ are unspecified we can employ the standard dilation arguments and assume that the shared states are pure, that is, $\rho_{s\overline{s}}=\proj{\psi_{s\overline{s}}}$ and that the measurements are projective. These states and measurements generate correlations that we denote by $\p$.

Consider then a reference experiment involving some known pure states $\ket{\psi'_{s\overline{s}}}\in\mathcal{H}_{s'}\otimes\mathcal{H}_{\overline{s}'}$ and known projective measurements represented by the observables $A_x'$, $B_y'$, $C_z'$ and $E_e'$ that generate the same correlations $\p$. We say that both experiments are equivalent, or, alternatively, that $\ket{\psi'_{s\overline{s}}}$ and $A'_x$, $B'_y$, $C'_z$, $E'_e$ are self-tested from $\p$ if one can prove that the local Hilbert spaces admit the product form $\mathcal{H}_s=\mathcal{H}_{s'}\otimes\mathcal{H}_{s''}$ and $\mathcal{H}_{\overline{s}}=\mathcal{H}_{\overline{s}'}\otimes\mathcal{H}_{\overline{s}''}$ $(s=A,B,C)$ for some auxiliary Hilbert spaces $\mathcal{H}_{s''}$ and $\mathcal{H}_{\overline{s}''}$, and that there are local unitary operations  $U_s:\mathcal{H}_s\to \mathcal{H}_{s'}\otimes\mathcal{H}_{s''}$ and $V_{\overline{s}}:\mathcal{H}_{\overline{s}}\to \mathcal{H}_{\overline{s}'}\otimes\mathcal{H}_{\overline{s}''}$
\begin{equation}
 (U_s\otimes V_{\overline{s}})\ket{\psi_{s\overline{s}}}=\ket{\psi'_{s'\overline{s}'}}\otimes\ket{\mathrm{junk}_{s''\overline{s}''}},
\end{equation}
where $\ket{\mathrm{junk}_{s''\overline{s}''}}$ belongs to $\mathcal{H}_{s''}\otimes\mathcal{H}_{\overline{s}''}$ and
\begin{equation}
    U_s\,s_{i}\,U_s^{\dagger}=s_{i}'\otimes\mathbbm{1}_{s''},\quad V_{\overline{E}}\, E_e\,V_{\overline{E}}^{\dagger}=E_e'\otimes\I_{E''},
\end{equation}
where $\mathbbm{1}_{E''}$ is the identity acting on the auxiliary systems $\mathcal{H}_{\overline{s}''}$ and $V_{\overline{E}}=U_{\overline{A}}\otimes U_{\overline{B}}\otimes U_{\overline{C}}$; recall that $E''=\overline{A}''\overline{B}''\overline{C}''$.

It is important to note here that, since the measurements can only be characterized on the local states, a natural assumption that we make throughout this work is that the latter are full-rank. Moreover, we do not assume that Eve's measurements are projective for our self-testing proof.

\textit{Results.---}
We propose a scheme that allows to device-independently certify in a single experiment the three different measurements introduced above, i.e., $M_{\mathrm{GHZ}}$, $M_{\mathrm{NLWE}}$, and the hybrid one $M_{\mathrm{UPB}}$. Consider again the network scenario in which 
the unknown pure states $|\psi_{s\overline{s}}\rangle$ with $s=A,B,C$
are distributed by independent sources among four parties who perform unknown measurements $A_x$, $B_y$, $C_z$ and $E_e$ on their shares of those states. Now, their aim is to exploit the observed correlations $\p$ to certify in the network the ideal reference experiment in which 
each source distributes the maximally entangled state of two qubits
$\ket{\psi'_{s\overline{s}}}=\ket{\phi^+}=(\ket{00}+\ket{11})/\sqrt{2}$, 
the external observers measure the following observables on their shares of the joint state
\begin{equation}\label{idealABC}
A'_{0/1/2}=\frac{X\pm Z}{\sqrt{2}}/Y,\quad s'_{0/1/2}=Z/X/Y\ \ (s=B,C).
\end{equation}
and Eve performs the three measurements mentioned above, i.e.,
$E'_0=M_{\mathrm{GHZ}}$, $E'_1=M_{\mathrm{NLWE}}$ and $E'_2=M_{\mathrm{UPB}}$.

To make our considerations easier to follow we divide them into three parts, each devoted to one of Eve's measurements. We begin with $E_0=\{R_{l|0}\}_{l=0}^7$. To certify that it is equivalent to the GHZ measurement $M_{\mathrm{GHZ}}$, the observed correlations $\p$ must be such that for each outcome $l$ of $E_0$, 
$\{p(abcl|xyz0)\}$ $(x,y,z=0,1)$ maximally violate the Bell inequality 
\begin{equation}\label{BE1m}
\sqrt{2}(-1)^{l_1}\left\langle2\widetilde{A}_1 B_1 C_1+(-1)^{l_2}\widetilde{A}_0B_0
+(-1)^{l_3}\widetilde{A}_0C_0\right\rangle\leq 4,
\end{equation}
where $\widetilde{A}_{0/1}=(A_0\pm A_1)/\sqrt{2}$ and $l\equiv l_1l_2l_3$ with $l_1,l_2,l_3=0,1$ is the binary representation of $l$, and
the probability of observing the outcome $l$ by Eve must obey $\overline{P}(l|e=0)=1/8$.

The Bell inequality corresponding to $l=0$ was introduced in \cite{Flavio}
and is maximally violated by $\ket{\phi_0}$, whereas those corresponding to $l\neq 0$ are its modifications that are adjusted to be maximally violated by the remaining GHZ states $|\phi_{l}\rangle$ and the quantum observables given in Eq. (\ref{idealABC}). In fact, these Bell violations can be achieved in the reference quantum network described above.

Let us now state our first result on self-testing the GHZ measurement (see Appendix B of \cite{SupMat} for proof).
\begin{thm}\label{theorem1m}
Assume that the observed correlations $\p$ obtained in the network are such that the Bell inequalities in Eq.~\eqref{BE1m} are maximally violated for each outcome $l$ of Eve's measurement $E_0$ and that each outcome occurs with probability $\overline{P}(l|e=0)=1/8$. 
Then,
 (i) the Hilbert spaces decompose as $\mathcal{H}_{s}=\mathcal{H}_{s'}\otimes\mathcal{H}_{s''}$  and $\mathcal{H}_{\overline{s}}=\mathcal{H}_{\overline{s}'}\otimes\mathcal{H}_{\overline{s}''}$; (ii) There exist local unitary transformations $U_{s}:\mathcal{H}_{s}\rightarrow\mathcal{H}_{s}$ and $V_{\overline{s}}:\mathcal{H}_{\overline{s}}\rightarrow\mathcal{H}_{\overline{s}}$ 
such that 
\begin{equation}
(U_{s}\otimes V_{\overline{s}})\ket{\psi_{s\overline{s}}}=|\phi^+_{s'\overline{s}'}\rangle\otimes\ket{\xi_{s''\overline{s}''}}    
\end{equation}
for some $\ket{\xi_{s''\overline{s}''}}\in\mathcal{H}_{s''}\otimes\mathcal{H}_{\overline{s}''}$, and the measurements of all parties are certified as
\begin{equation}
\overline{V}\,R_{l|0}\,\overline{V}^{\dagger}=\proj{\phi_l}_{E'}\otimes\I_{E''},\quad U_{s}\,s_{i}\,U_{s}^{\dagger}=s'_i\otimes\I_{s''} 
\end{equation}
for all $l$ and $i=0,1$ where $\overline{V}=\otimes_s V_{\overline{s}}$ such that $s=A,B,C$ and $E=ABC$. The states $\ket{\phi_l}$ and the observables $s'_i$ are given in Eqs. \eqref{GHZvecsm} and \eqref{idealABC} respectively.
\end{thm}

In what follows we build on this result to show how to certify the other Eve's measurements $E_1$ and $E_2$ and also the third measurements of the external parties $A_2$, $B_2$ and $C_2$.
Let us then consider Eve's second measurement $E_1=\{R_{l|1}\}_{l=0}^7$. In order to certify that it is equivalent to the separable measurement $M_{\mathrm{NLWE}}$, the observed correlations $\p$, apart from the conditions stated in Theorem \ref{theorem1m}, must additionally satisfy
\begin{eqnarray}\label{NLWEstatm}
&&\hspace{-0.5cm}p(0100|0011)=p(0111|0011)=p(0012|1001)\nonumber\\
&&=p(1013|1001)=p(1004|0101)=p(1105|0101)\nonumber\\
&&=p(0006|0001)=p(1117|0001)=\frac{1}{8}.
\end{eqnarray}
Notice that these conditions are met in the ideal experiment outlined above.

Let us state formally our second result 
that together with Theorem \ref{theorem1m} provides 
a scheme for DI certification of
the separable measurement exhibiting NLWE in the least possible dimension (cf. Appendix C of \cite{SupMat} for a proof).
\begin{thm}\label{theorem3m}
Suppose that $\p$ generated in the network satisfies the assumptions of Theorem \ref{theorem1m} as well as the conditions in Eq. (\ref{NLWEstatm}). 
Then, for any $l$ it holds that $\overline{V}\, R_{l|1}\,
\overline{V}^{\dagger} =\proj{\delta_l}_{E'}\otimes\I_{E''}$
where $\overline{V}$ is the same unitary 
as in Theorem \ref{theorem1m} and $E=ABC$.
\end{thm}

Before proceeding to the final result which is self-testing of $M_{\mathrm{UPB}}$ in $E_2$, we need to 
introduce another condition that is necessary to prove a 
self-testing statement for $A_2$, $B_2$ and $C_2$. Precisely, the correlations $\{p(abc0|xyz0)\}$ with $x,y,z=1,2$ corresponding to the situation in which Eve observes the first outcome $l=0$ of $E_0$, the following condition is satisfied
\begin{equation}\label{BE2m}
\langle \widetilde{A}_1B_1 C_1-\widetilde{A}_1B_2 C_2-A_2 B_1C_2-A_2B_2 C_1\rangle=4,
\end{equation}
%
where $\widetilde{A}_1=(A_0-A_1)/\sqrt{2}$ and the above functional is inspired by the Mermin inequality \cite{Mermin}. This along with Theorem \ref{theorem1m} implies that (see Appendix B  of \cite{SupMat}) 
\begin{eqnarray}\label{stmea3m}
    U_s\,s_2\,U_{s}^{\dagger}=\pm Y_{s'}\otimes\I_{s''}\qquad (s=A,B,C).
\end{eqnarray}
With the above characterization at hand, we can finally move onto showing how to certify $M_{\mathrm{UPB}}$ in $E_2=\{R_{l|2}\}_{l=0}^4$. To this end, the observed correlations must satisfy
%
%
\begin{eqnarray}\label{betstat0m}
p(0100|0012)&=&p(0011|1002)=p(1002|0102)\nonumber\\ 
&=&p(1113|1112)=\frac{1}{8}
\end{eqnarray}
along with four other conditions stated in Appendix D of \cite{SupMat} as Eqs. (D4a)-(D4d);
we refer to them as Pr$_2$.
Notice again that correlations obtained within the reference experiments fulfill the above conditions.
Let us now state the following theorem.
\begin{thm}\label{theorem4m}
Assume that the assumptions of Theorem \ref{theorem3m} and the conditions in Eq. \eqref{betstat0m} and Pr$_2$ are satisfied. Then, 
the measurement $E_2=\{R_{l|2}\}$ is certified as
%
 $\overline{V}\, R_{l|2}\,\overline{V}^{\dagger} =\proj{\tau_l}_{E'}\otimes\I_{E''}$  
%
for $l=0,1,2,3,$ and,
%
 $\overline{V}\, R_{4|2}\,\overline{V}^{\dagger} =\Gamma_{E'}\otimes\I_{E''},  $ 
%
where $\ket{\tau_l}$ and $\Gamma$ are defined in Eqs. \eqref{UPBm} and \eqref{UPB1m}, respectively, and $\overline{V}$ is the same unitary operation as in Theorem \ref{theorem1m}. 
\end{thm}
The proof can be found in Appendix D of \cite{SupMat}. This final result shows that the hybrid separable-entangled measurement constructed from a UPB can also be self-tested using our scheme. Most importantly, this is the minimal scenario possible to self-test such a measurement.

\textit{Bound entangled state.---}An interesting consequence of Theorem \ref{theorem4m} is that the considered network allows one to self-test a bound entangled state shared between the external parties. Assume that the states and Eve's measurement $E_2$ are certified as in Eq. \eqref{idealABC} and as in Theorem \ref{theorem4m}, respectively. The post-measurement state shared by the external parties that corresponds to the last outcome of $E_2$ is then given by [see Appendix D of \cite{SupMat}]
\begin{eqnarray}
U\,\rho_{ABC}\,U^{\dagger}=\frac{1}{4}\Gamma_{A'B'C'}\otimes \tilde{\rho}_{A''B''C''},
\end{eqnarray}
where $U=\bigotimes_sU_s$ and the unitaries $U_s$ are the same as in Theorem \ref{theorem1m}. As mentioned above, the state $\Gamma/4$ is bound entangled  \cite{Bennett99} and can be prepared by Eve in the external parties labs with a simple post-processing strategy. She first broadcasts her outcome of the measurement $E_2$ and then the external parties discard those runs of the experiment for which Eve observes any other outcome than the last one. 

\textit{Robustness.---} In a situation when the observed correlations are not ideal, we have also analyzed the noise robustness of our scheme for the particular Eve's measurement corresponding to $e=0$. Our findings can be summarized as the following theorem whose proof is deferred to Appendix E of \cite{SupMat}.
%
%
\begin{thm} Consider the Bell inequalities \eqref{BE1m} $\mathcal{I}_{l}$ is violated $\varepsilon$-close to the maximal violation, that is, $\mathcal{I}_{{l}}\geq\beta_Q-\varepsilon$ along with the probabilities of the central party for $e=0$ given by $ \left |\overline{P}(l|e=0)-1/8\right|\leq\varepsilon$. Then, the ideal Eve's measurement $\proj{\phi_l}_{E'}\otimes\I_{E''}$ is close to the actual one $R_{l|0}$ as
       \begin{eqnarray}\label{robures1}
        \left\|\Tr_{E''}(\Tilde{R}_l)- \proj{\phi_l}_{E'}\right\|\leq 17\left(\varepsilon+ 9\sqrt{2\varepsilon}\right)
    \end{eqnarray}
    for all $l$, where $\Tilde{R}_l=(\bigotimes_{s}\,V_{\overline{s}})\ R_l\ (\bigotimes_{s} V_{\overline{s}}^{\dagger})$.
\end{thm}
%
Let us notice that the proof of this statement exploits the robustness analysis of the network-based self-testing scheme for the GHZ-like states \eqref{GHZvecsm} derived recently in Ref. \cite{sarkar2025}.
In an analogous way, one can perform a similar analysis for the other two measurements of Eve.
%

\textit{Outlook.---}Inspired by the notion of NLWE, we can identify a larger set of orthogonal projectors which can be referred to as nonlocality without distillable entanglement (NLWDE). NLWDE is defined as a set of projectors that cannot be perfectly distinguished using local operations and classical communication, such that their normalized versions are positive under partial transpose with respect to every bipartition. In this work, the measurement composed of UPB and bound entanglement falls under this category. It will be interesting to further analyze the properties of such sets. 

Several other follow-up problems arise from our work. Designing certification schemes for GHZ bases and product bases exhibiting NLWE for $N$ qubits will be straightforward from our work. A more interesting question would be to generalize our scheme to certify any hybrid or even any composite projective measurement in the case of any number of qubits. An even more challenging problem will be to propose a quantum-networks-based scheme for non-projective composite measurements. Our work also establishes a way to self-test mixed entangled states. A natural question here would be to construct schemes to self-test any mixed entangled state; a possibility that does not seem to exist in the standard Bell scenario. This has now been achieved in \cite{sarkar2024universal}.

\textit{Acknowledgments---}
 This project was funded within the QuantERA II Programme (VERIqTAS project) that has received funding from the European Union’s Horizon 2020 research and innovation programme under Grant Agreement No 101017733 and from the Polish National Science Center (project No 2021/03/Y/ST2/00175). C.D. and S.H. acknowledge the support by the ``Quantum Optical Technologies'' project, carried out within the International Research Agendas program of the Foundation for Polish Science co-financed by the European Union under the European Regional Development Fund. C.D. also acknowledges the support from the German Federal
Ministry of Education and Research (BMBF) within the funding program ``quantum technologies -- from basic research to market'' in the joint project QSolid (grant
number 13N16163).

\vspace{2cm}
\appendix




\onecolumngrid
\begin{center}
    \textbf{SUPPLEMENTARY MATERIALS}
\end{center}
\tableofcontents
\vspace{.5cm}
We first prove some basic mathematical lemmas that are useful towards the presented self-testing statement. Then, we self-test the measurement corresponding to Eve's input $e=0$ along with the sources and the measurements of other parties in the quantum network. Using the certified states and measurements we then self-test Eve's measurement corresponding to inputs $e=1,2$. Finally, we analyse the robustness of our self-testing scheme towards experimental imperfections.


\section{General results}
Before proceeding to the proofs of the main results, we introduce an important lemma that is required to derive our results.
\begin{lem}\label{theorem1.1}
Consider a positive semi-definite matrix $M$ such that $M\leq\I$.  Also, consider a density matrix $\rho$ that satisfies $\Tr(M\rho)=1$. Then, $M$ is an identity matrix acting on the support of $\rho$.
\end{lem}
\begin{proof}
Let us consider the eigendecomposition of $\rho$,  $\rho=\sum_kp_k\proj{\psi_k}$ such that $p_k>0$. After putting it into the condition $\Tr(M\rho)=1$, one obtains 
\begin{eqnarray}
\sum_kp_k\bra{\psi_k}M\ket{\psi_k}=1,
\end{eqnarray}
which by employing the fact that $\sum_kp_k=1$ can further be rewritten as
\begin{eqnarray}
\sum_kp_k(1-\bra{\psi_k}M\ket{\psi_k})=0.
\end{eqnarray}
Due to the facts that $0\leq M\leq \I$ and $p_k> 0$, the above equation can hold true only if $\bra{\psi_k}M\ket{\psi_k}=1$. Thus, every $\ket{\psi_k}$ is an eigenstate of $M$ with eigenvalue $1$. As a result, $M=\sum_k\proj{\psi_k}=\I_{\rho}$,
where $\I_{\rho}$ is an identity acting on the support of $\rho$.
\end{proof}
\begin{lem}\label{lem2}
Consider two unitary observables $A_0$ and $A_1$ such that $A_i^2=\mathbbm{1}$ acting on a Hilbert space $\mathcal{H}$. Consider then projections $\Pi_j$ onto $k$ subspaces $\mathcal{K}_j$ of $\mathcal{H}$ such that $\Pi_1+\dots+\Pi_k>0$. Let us denote $\overline{A}_{i}^{(j)}=\Pi_jA_i\Pi_j$ and assume that 
\begin{equation}\label{dupablada}
 \overline{A}_{0}^{(j)}\overline{A}_{1}^{(j)}+  \overline{A}_{1}^{(j)}\overline{A}_{0}^{(j)} =0\qquad (j=1,\ldots,k)
\end{equation}
and that $\overline{A}_{i}^{(j)}$ are unitary on the subspaces $\mathcal{K}_j$.
Then, the matrices $A_i$ anticommute, 
\begin{equation}
 \{A_0,A_1\}=0.
\end{equation}
%
%
%
\end{lem}
\begin{proof}
Let us first show that $A_i=\overline{A}_i^{(j)}\oplus \widetilde{E}_{i,j}$, where $\widetilde{E}_{i,j}$ with $i=0,1$ are matrices acting on the complementary subspace to $\mathcal{K}_j$. For this purpose, let us fix $j=1$ and express the observable $A_0$ as
\begin{equation}
    A_0=
    \begin{pmatrix}
    \overline{A}_0^{(1)}& B\\ C & \widetilde{E}_{0,1}
    \end{pmatrix},
\end{equation}
where $\widetilde{E}_{0,1}$, $B$ and $C$ are some matrices; in particular, 
$\widetilde{E}_{0,1}$ acts on $\mathcal{K}_1^{\perp}$.

Now, exploiting the conditions that $A_1A_1^{\dagger}=\I$ and $ \overline{A}_0^{(1)}$ is unitary we can conclude that $B=0$. Similarly, using  $A_0^{\dagger}A_0=\I$ one obtains that $C=0$, which gives us the desired form of $A_0$. Similar analyses with any $i,j$ gives us $A_i=\overline{A}_i^{(j)}\oplus \Tilde{E}_{i,j}$. 

The facts that $A_i$ decompose into blocks with respect to the 
decompositions $\mathcal{H}=\mathcal{K}_j\oplus\mathcal{K}_j^{\perp}$ and that with respect to the same decompositions, $\Pi_j$ can be expressed as
\begin{equation}
    \Pi_j=
    \begin{pmatrix}
    \mathbbm{1}_{K_j}& 0\\ 0 & 0
    \end{pmatrix},
\end{equation}
where $\mathbbm{1}_{K_j}$ is the identity on $\mathcal{K}_j$, we can 
rewrite the conditions (\ref{dupablada}) as
\begin{equation}\label{dupablada1}
    \{\overline{A}_{0}^{(j)},\overline{A}_{1}^{(j)}\}=\Pi_j\{A_0,A_1\}=0.
\end{equation}

After summing \eqref{dupablada1} over all subspaces we obtain
\begin{equation}\label{dupablada5}
    (\Pi_1+\ldots+\Pi_k)\{A_0,A_1\}=0,
\end{equation}
which by using the fact that $\Pi_1+\ldots+\Pi_k$ is a full rank positive semi-definite matrix, implies that $A_0$ and $A_1$ anticommute on $\mathcal{H}$.
%
\end{proof}

\section{Self-testing of GHZ bases, measurements of external parties, and the states prepared by the preparation devices}
\label{AppB}

In this section, we explain the certification of the GHZ basis, measurements of the external parties and the states distributed by the sources.

\subsection{The GHZ basis and measurements $A_x$, $B_y$ and $C_z$ with $x,y,z=0,1$}

Let us first consider the following eight Bell inequalities 
\begin{eqnarray}\label{BE1}
\mathcal{I}_{l}=2(-1)^{l_1}\left\langle(A_0+A_1) B_1 C_1\right\rangle+(-1)^{l_2+l_1}\left\langle(A_0-A_1) B_0\right\rangle+(-1)^{l_3+l_1}\left\langle(A_0-A_1) C_0\right\rangle\leq 4,
\end{eqnarray}
where $l=l_1l_2l_3$ with $l_i\in\{0,1\}$ for each $i=1,2,3$. 
The inequality for $l_1=l_2=l_3=0$ was constructed in \cite{Flavio}, whereas the remaining seven are its variants obtained by making the signs in front of each expectation value depend on the parameters $l_i$. We can now state the following fact which is concerned with Tsirelson's bound of the inequalities in Eq. \eqref{BE1}.

\begin{fakt}The maximal quantum value of the Bell expression $\mathcal{I}_l$ is $4\sqrt{2}$ and it is achieved by the following observables 
\begin{eqnarray}\label{GHZObs}
A_0=\frac{X+Z}{\sqrt{2}},\quad A_1= \frac{X-Z}{\sqrt{2}},\quad B_0=Z,\quad B_1=X,\quad C_0=Z,\quad C_1=X
\end{eqnarray}
as well as the GHZ-like state 
\begin{equation}\label{GHZvecs}
\ket{\phi_l}=\frac{1}{\sqrt{2}}(\ket{l_1l_2l_3}+(-1)^{l_{1}}|\overline{l}_1\overline{l}_2\overline{l}_3\rangle),
\end{equation}
where $l\equiv l_1l_2l_3$ with $l_1,l_2,l_3=0,1$ and $\overline{l}_i$ is the negation of the bit $l_i$, i.e., $\overline{l}_i=1-l_i$. 
\end{fakt}
\begin{proof}
We follow the results of Ref. \cite{Flavio}. First, to each of the Bell expressions $\mathcal{I}_l$ we associate a Bell operator of the following form
%
\begin{eqnarray}\label{BE1Nop}
\mathcal{\hat{I}}_{l_1l_2l_3}=2(-1)^{l_1}(A_0+A_1)\otimes B_1\otimes C_1+(-1)^{l_2+l_1}(A_0-A_1)\otimes B_0+(-1)^{l_3+l_1}(A_0-A_1)\otimes C_0,
\end{eqnarray}
where $A_x$, $B_y$ and $C_z$ are arbitrary $\pm1$-valued quantum observables of 
arbitrary dimensions. Second, for each $\mathcal{\hat{I}}_l$ we construct the following sum-of-squares (SOS) decomposition,
\begin{eqnarray}\label{SOS1}
   4\sqrt{2}\,\I-\mathcal{\hat{I}}_{l_1l_2l_3}=\sqrt{2}\left(\I-P_{1,l_1}\right)^2+\frac{1}{\sqrt{2}}\left[\left(\I-P_{2,l_1,l_2}\right)^2+\left(\I-P_{3,l_1,l_3}\right)^2\right], 
\end{eqnarray}
where 
\begin{subequations}\label{SOS2}
\begin{equation}\label{SOS2_1}
    P_{1,l_1}=(-1)^{l_1}\frac{A_{0}+A_{1}}{\sqrt{2}} \otimes B_1\otimes C_1,
\end{equation}
\begin{equation}\label{SOS2_2}
    P_{2,l_1,l_2}=(-1)^{l_1+l_2}\frac{A_{0}-A_{1}}{\sqrt{2}}\otimes B_0,
\end{equation}
\begin{equation}\label{SOS2_3}
    P_{3,l_1,l_3}=(-1)^{l_1+l_3}\frac{A_{0}-A_{1}}{\sqrt{2}}\otimes C_0.
\end{equation}
\end{subequations}
It directly follows from Eq. \eqref{SOS1} that 
$4\sqrt{2}\,\mathbbm{1}-\mathcal{\hat{I}}_l\geq 0$ and thus $4\sqrt{2}$ is an upper bound on the maximal quantum value of $\mathcal{I}_{l}$, which means that  $\langle\psi|\mathcal{\hat{I}}_l|\psi\rangle\leq 4\sqrt{2}$ for arbitrary state $\ket{\psi}$. To finally show that the latter inequality is tight
and that $4\sqrt{2}$ is in fact Tsirelson's bound of the inequalities
in Eq. \eqref{BE1} it suffices to observe that $\mathcal{I}_l$ achieves the value $4\sqrt{2}$ for the GHZ-like state $\ket{\phi_l}$ and the observables given in Eq. (\ref{GHZObs}).  
\end{proof}

Crucially, the SOS decomposition in Eq. \eqref{SOS1} implies that any state $\ket{\psi}$ and any observables $A_x$, $B_y$ and $C_z$ that achieve the quantum bound $\beta_Q=4\sqrt{2}$ of $\mathcal{I}_l$ must satisfy the following relations
\begin{eqnarray}\label{relSOS}
   P_{1,l_1}\ket{\psi}=\ket{\psi}\qquad\text{and}\qquad P_{i,l_1,l_i}\ket{\psi}=\ket{\psi}\qquad (i=2,3).
\end{eqnarray}
%
which, by virtue of the relations in Eqs. (\ref{SOS2_1}), (\ref{SOS2_2}) and (\ref{SOS2_3}),
can be rewritten as 
\begin{subequations}\label{SOSrel}
\begin{equation}\label{SOSrel1}
   (-1)^{l_1}\frac{A_{0}+A_{1}}{\sqrt{2}}\otimes B_1\otimes C_1\ket{\psi}=\ket{\psi},
\end{equation}
\begin{equation}\label{SOSrel2}
   (-1)^{l_1+l_2}\frac{A_{0}-A_{1}}{\sqrt{2}}\otimes B_0\ket{\psi}=\ket{\psi},
\end{equation} 
\begin{equation}\label{SOSrel3}
   (-1)^{l_1+l_3}\frac{A_{0}-A_{1}}{\sqrt{2}}\otimes C_0\ket{\psi}=\ket{\psi}.
\end{equation}
\end{subequations}
%
The above relations are used in the proof of Theorem 1 stated in the main text as well as in the preceding subsection.

\subsection{Proof of self-testing}

\setcounter{thm}{0}
\begin{thm}\label{theorem1}
Consider the network scenario outlined in the main text and assume that the observed correlations $\vec{p}$ achieve the maximal quantum value of $\mathcal{I}_l$ in Eq. \eqref{BE1} for each outcome $l$ of Eve's first measurement $E_0$ and that each outcome $l$ occurs with probability $\overline{P}(l|e=0)=1/8$. 
Then,
\begin{itemize}
    \item[(i)] All six Hilbert spaces decompose as $\mathcal{H}_{s}=\mathcal{H}_{s'}\otimes\mathcal{H}_{s''}$,
    and $\mathcal{H}_{\overline{s}}=\mathcal{H}_{\overline{s}'}\otimes\mathcal{H}_{\overline{s}''}$  with $s=A,B,C$, where $\mathcal{H}_{s'}$ and $\mathcal{H}_{\overline{s}'}$ are one-qubit Hilbert spaces.
  
    \item[(ii)] There exist local unitary transformations $U_{s}:\mathcal{H}_{s}\rightarrow\mathcal{H}_{s}$ and $V_{\overline{s}}:\mathcal{H}_{\overline{s}}\rightarrow\mathcal{H}_{\overline{s}}$ 
    such that 
\begin{eqnarray}\label{statest1}
U_{s}\otimes V_{\overline{s}}\ket{\psi_{s\overline{s}}}=|\phi^+_{s'\overline{s}'}\rangle\otimes\ket{\xi_{s''\overline{s}''}}
\end{eqnarray}
for each $s=A,B,C$.
\item[(iii)] Then, Eve's first measurement $E_0=\{R_{l|0}\}$ satisfies
\begin{eqnarray}\label{stmea1}
 (V_{\overline{A}}\otimes V_{\overline{B}}\otimes V_{\overline{C}})\,  R_{l|0}\,(V_{\overline{A}}\otimes V_{\overline{B}}\otimes V_{\overline{C}})^{\dagger} =\proj{\phi_l}_{E'}\otimes\I_{E''},
\end{eqnarray}
where $E=\overline{A}\overline{B}\overline{C}$, $\ket{\phi_l}$ are the GHZ-like states given in Eq. \eqref{GHZvecs} and the measurements of all other parties are given by
\begin{equation}\label{stmea2}
\begin{split}
U_{A}\,A_{0}\,U_{A}^{\dagger}&=\left(\frac{X+Z}{\sqrt{2}}\right)_{A'}\otimes\I_{A''}, \quad U_{A}\,A_{1}\,U_{A}^{\dagger}=\left(\frac{X-Z}{\sqrt{2}}\right)_{A'}\otimes\I_{A''},\\
U_{B}\,B_0\,U_{B}^{\dagger}&=Z_{B'}\otimes\I_{B''},\quad\quad\ \ \ \ \  U_{B}\,B_{1}\,U_{B}^{\dagger}=X_{B'}\otimes\I_{B''},\\
U_{C}\,C_0\,U_{C}^{\dagger}&=Z_{C'}\otimes\I_{C''},\quad\quad\ \ \ \ \  U_{C}\,C_{1}\,U_{C}^{\dagger}=X_{C'}\otimes\I_{C''}.
\end{split}
\end{equation}
\end{itemize}
\end{thm}
\begin{proof} Before we proceed with the proof we first notice that the post-measurement state that $A$, $B$ and $C$ share after Eve performs her first measurement $E_0$ and obtains the outcome $l$ is given by
\begin{eqnarray}\label{54}
\rho^l_{{ABC}}=\frac{1}{\overline{P}(l)}\Tr_{\overline{ABC}}\left[\left(\I_{ABC}\otimes R_{l|0}\right)\bigotimes_{s=A,B,C}\proj{\psi_{s\overline{s}}}\right].
\end{eqnarray}

We divide the proof into a few steps and the first one is concerned with determining the form of the states $\rho^l_{{ABC}}$ for any $l$ from the observed maximal violations of the inequalities in Eq. \eqref{BE1}. Building on this result we then find the form of the states generated by the sources $\ket{\psi_{s\overline{s}}}$ for any $s=A,B,C$. Finally, using both of these results we obtain the form of the entangled measurement $\{R_{l|0}\}$. For simplicity, in the rest of the proof, we represent $R_{l|0}$ as $R_{l}$.\\

\noindent\textbf{(a) Post-measurement states $\rho^l_{{ABC}}$.} 
To determine the form of the post-measurement states $\rho^l_{{ABC}}$ we exploit the relations given in Eq. \eqref{SOSrel}. First, let us consider a purification of $\rho^l_{{ABC}}$, denoted $|\psi_l\rangle_{ABCE}$, which is a pure state satisfying 
\begin{eqnarray}
   \rho^l_{{ABC}}=\Tr_E\left(|\psi_l\rangle\!\langle\psi_l|_{{ABCE}}\right).
\end{eqnarray}
For simplicity, in what follows we drop the subscript from the above state. 
From the assumption that $\rho^l_{{ABC}}$ maximally violates the Bell inequality in Eq. \eqref{BE1} it follows that the relations in Eq. \eqref{SOSrel} are satisfied by the purification $\ket{\psi_l}$, which after taking into account that $B_y^2=C_z^2=\I$ can be stated as
\begin{subequations}\label{SOSrelp}
\begin{eqnarray}
   && (-1)^{l_1}\frac{A_{0}+A_{1}}{\sqrt{2}} \ket{\psi_l}=B_1\otimes C_1\ket{\psi_l}\label{SOSrel1p}\\
 &&  (-1)^{l_1+l_2}\frac{A_{0}-A_{1}}{\sqrt{2}} \ket{\psi_l}=B_0\ket{\psi_l},\label{SOSrel2p}\\ && (-1)^{l_1+l_3}\frac{A_{0}-A_{1}}{\sqrt{2}} \ket{\psi_l}=C_0\ket{\psi_l}.\label{SOSrel3p}
\end{eqnarray}
\end{subequations}
In the above equations as well as in the following considerations we omit the identities acting on the remaining subsystems, including the $G$ one. 

Let us now denote by $\Pi_l^A$ projection onto the supports of the local reduced density matrix of Alice $\rho_A^l=\Tr_{BCE}\psi_l$. Let us then consider 
Eqs. (\ref{SOSrel1p}) and (\ref{SOSrel2p}) and project them onto $\Pi_l^A$, which results in the following equations
\begin{equation}\label{dupalinder1}
    (-1)^{l_1}\left(\overline{A}_{0}^{(l)}+\overline{A}_{1}^{(l)}\right) \ket{\psi_l}=\sqrt{2}\,B_1\otimes C_1\ket{\psi_l}
\end{equation}
and
\begin{equation}\label{dupalinder2}
   (-1)^{l_1+l_2}\left(\overline{A}_{0}^{(l)}-\overline{A}^{(l)}_{1}\right) \ket{\psi_l}=\sqrt{2}\,B_0\ket{\psi_l}
\end{equation}
for any $l$, where $\overline{A}_{i}^{(l)}=\Pi_l^A \,A_i\,\Pi_l^A$.
By applying then $B_1\otimes C_1$ and $B_0$ to both sides of 
Eq. (\ref{dupalinder1}) and (\ref{dupalinder2}), respectively, 
we arrive at
\begin{eqnarray}\label{SOSrel4p}
  \left(\overline{A}^{(l)}_{0}+\overline{A}^{(l)}_{1}\right)^2 \ket{\psi_l}=2\ket{\psi_l},
\end{eqnarray}
and
\begin{eqnarray}\label{SOSrel4p}
  \left(\overline{A}^{(l)}_{0}-\overline{A}^{(l)}_{1}\right)^2 \ket{\psi_l}=2\ket{\psi_l},
\end{eqnarray}
which directly imply that

%
%
\begin{eqnarray}\label{EXTRA1}
\left(\overline{A}_{0}^{(l)}+\overline{A}_{1}^{(l)}\right)^2=2\,\Pi_l^A,\qquad \left(\overline{A}_{0}^{(l)}-\overline{A}_{1}^{(l)}\right)^2=2\,\Pi_l^A.
\end{eqnarray}
Expanding the above formulas 
we get that 
\begin{equation}
 \left(\overline{A}_0^{(l)}\right)^2=\left(\overline{A}_1^{(l)}\right)^2=\Pi_l^A  
\end{equation}
and
\begin{equation}
  \left\{\overline{A}_{0}^{(l)},\overline{A}_{1}^{(l)}\right\}=0  
\end{equation}
for all $l$. 

Now, since $\sum_l \Pi_l^A$ is a full rank operator on $\mathcal{H}_A$, which is a consequence of the facts that the states $\rho_{ABC}^l$ come from a quantum measurement and that the reduced density matrices of the initial state are full rank, Lemma \ref{lem2} implies that 
\begin{eqnarray}
    \{A_0,A_1\}=0,
\end{eqnarray}
which by employing the standard arguments implies that 
the local Hilbert space $\mathcal{H}_A$ is even dimensional, i.e., 
$\mathcal{H}_A=\mathcal{H}_{A'}\otimes\mathcal{H}_{A''}$, where
$\mathcal{H}_{A'}$ is a qubit Hilbert space and 
there exist a unitary operation $U_{A}:\mathcal{H}_{A}\rightarrow\mathcal{H}_{A}$ such that
%
%
%
\begin{eqnarray}\label{A1}
U_{A}\,A_{0}\,U_{A}^{\dagger}&=&\frac{X+Z}{\sqrt{2}}\otimes\I_{A''}, \qquad U_{A}\,A_{1}\,U_{A}^{\dagger}=\frac{X-Z}{\sqrt{2}}\otimes\I_{A''}.
\end{eqnarray}
%

Let us now move on to characterizing the other parties' observables and use Eq. \eqref{A1} to rewrite the relations in Eqs. \eqref{SOSrelp} as
\begin{subequations}\label{SOSrelp1}
\begin{eqnarray}
  && (-1)^{l_1}\,X_{A'} \ket{\psi_l'}=\overline{B}_1^{(l)}\otimes C_1\ket{\psi_l'},\label{SOSrel5p}\\
 &&   (-1)^{l_1+l_2}\,Z_{A'} \ket{\psi_l'}=\overline{B}_0^{(l)}\ket{\psi_l'},\label{SOSrel6p}\\ && (-1)^{l_1+l_3}Z_{A'}\ket{\psi_l'}=C_0\ket{\psi_l'}.\label{SOSrel7p}
\end{eqnarray}
\end{subequations}
where $\ket{\psi_l'}=U_{A}\ket{\psi_l}$, $\overline{B}_i^{(l)}=\Pi_{l}^B\,B_i\,\Pi_{l}^B$ such that $\Pi_{l}^B$ represents the projector onto the support of $\rho_B^l=\Tr_{ACE}\psi^l$. For simplicity, we also omitted the identity acting on the $A''$ subsystem. Let us then consider the relation in Eq. \eqref{SOSrel5p} and multiply it with $Z_{A'}$. Then using the relation in Eq. \eqref{SOSrel6p} on the right-hand side of the obtained expression, we get
\begin{eqnarray}\label{SOSrel8p}
    (-1)^{l_2}(ZX)_{A'}\ket{\psi_l'}= \overline{B}^{(l)}_1\overline{B}^{(l)}_0\otimes C_1\ket{\psi_l'}
\end{eqnarray}
Then, after multiplying Eq. \eqref{SOSrel6p} with $X_{A'}$ and using Eq. \eqref{SOSrel5p} on the right-hand side of the obtained expression, we get
\begin{eqnarray}\label{SOSrel9p}
    (-1)^{l_2}(XZ)_{A'} \ket{\psi_l'}= (-1)^{l_2}(ZX)_{A'}\ket{\psi_l'}= \overline{B}^{(l)}_0\overline{B}^{(l)}_1\otimes C_1\ket{\psi_l'}.
\end{eqnarray}
Adding Eqs. \eqref{SOSrel8p} and \eqref{SOSrel9p}, and using the fact that $ZX+XZ=0$, we finally arrive at 
\begin{eqnarray}
  \left \{ \overline{B}^{(l)}_0,\overline{B}^{(l)}_1\right\}\otimes C_1\ket{\psi_l'}=0.
\end{eqnarray}
Exploiting the facts that $C_1$ is unitary, we can remove it from the above equation, and further conclude that we get that 
\begin{equation}
    \left\{ \overline{B}^{(l)}_0,\overline{B}^{(l)}_1\right\}=0.
\end{equation}
Also, applying $X_{A'}$ to Eq. \eqref{SOSrel5p} and $Z_{A'}$ to Eq. \eqref{SOSrel6p}, one can easily conclude that $(\overline{B}_i^{(l)})^2=\Pi_l^B$.
Thus, we can again exploit Lemma \ref{lem2} to infer that 
\begin{equation}
    \{B_0,B_1\}=0,
\end{equation}
which together with $B_0^2=B_1^2=\mathbbm{1}_B$ implies that 
$\mathcal{H}_B=(\mathbbm{C}^2)_{B'}\otimes\mathcal{H}_{B''}$ and that there exists a unitary operation $U_{B}:\mathcal{H}_{B}\rightarrow\mathcal{H}_{B}$ for which
\begin{eqnarray}\label{An}
U_{B}\,B_{0}\,U_{B}^{\dagger}=Z_{B'}\otimes\I_{B''},\qquad  U_{B}\,B_{1}\,U_{B}^{\dagger}=X_{B'}\otimes\I_{B''}.
\end{eqnarray}

Proceeding exactly in the same manner, we can conclude that $C_0$ and $C_1$ that appear in Eqs. \eqref{SOSrel5p} and \eqref{SOSrel7p} also anticommute. Thus, as before, $\mathcal{H}_{C}=(\mathbbm{C}^2)_{C'}\otimes\mathcal{H}_{C''}$ and there exists a unitary $U_{C}:\mathcal{H}_{C}\rightarrow\mathcal{H}_{C}$ such that
\begin{eqnarray}\label{An2}
   U_C\,C_{0}\,U_{C}^{\dagger}=Z_{C'}\otimes\I_{C''},\qquad  U_C\,C_{1}\,U_{C}^{\dagger}=X_{C'}\otimes\I_{C''}.
\end{eqnarray}

Let us now characterise the state $\rho^l_{ABC}$. For this purpose, we plug the forms of the observables from Eqs. \eqref{A1}, \eqref{An} and \eqref{An2} into Eqs. \eqref{relSOS} to obtain
\begin{subequations}\label{SOSrelp2}
\begin{eqnarray}
  && (-1)^{l_1} X_{A'}\otimes X_{B'}\otimes X_{C'}\ket{\tilde{\psi}_l}=\ket{\tilde{\psi}_l}, \label{SOSrel10}\\
  && (-1)^{l_1+l_2}Z_{A'}\otimes Z_{B'}\ket{\tilde{\psi}_l}=\ket{\tilde{\psi}_l}, \label{SOSrel11}\\
  &&(-1)^{l_1+l_3}Z_{A'}\otimes Z_{C'}\ket{\tilde{\psi}_l}=\ket{\tilde{\psi}_l},\label{SOSrel12}
\end{eqnarray}
\end{subequations}
where $\ket{\tilde{\psi}_l}=U_A\otimes U_B\otimes U_C\ket{\psi_l}$. As already concluded above each local Hilbert space $\mathcal{H}_{s}$ $(s=A,B,C)$ decomposes as $\mathcal{H}_{s}=(\mathbbm{C}^2)_{s'}\otimes\mathcal{H}_{s''}$. Thus, $\ket{\tilde{\psi}_l}$ can be decomposed as 
\begin{eqnarray}\label{genstate4}
\ket{\tilde{\psi}_l}=\sum_{i_1,i_2,i_3=0,1}\ket{i_1i_2i_3}_{A'B'C'}|\phi^l_{i_1i_2i_3}\rangle_{A''B''C''G},
\end{eqnarray}
where the normalisation factors are included in $|\phi^l_{i_1i_2i_3}\rangle$. Putting the above form of $\ket{\tilde{\psi}_l}$ in Eqs. \eqref{SOSrel11} and \eqref{SOSrel12}, we obtain
\begin{eqnarray}
   (-1)^{l_1+l_2}\sum_{i_1,i_2,i_3=0,1}(-1)^{i_1+i_2}\ket{i_1i_2i_3}|\phi^l_{i_1i_2i_3}\rangle=\sum_{i_1,i_2, i_3=0,1}\ket{i_1i_2 i_3}|\phi^l_{i_1i_2 i_3}\rangle,
\end{eqnarray}
and
\begin{eqnarray}
   (-1)^{l_1+l_3}\sum_{i_1,i_2,i_3=0,1}(-1)^{i_1+i_3}\ket{i_1i_2i_3}|\phi^l_{i_1i_2i_3}\rangle=\sum_{i_1,i_2, i_3=0,1}\ket{i_1i_2 i_3}|\phi^l_{i_1i_2 i_3}\rangle.
\end{eqnarray}
For brevity, we dropped subscripts denoting the subsystems. Projecting both the above formulas on $\bra{i_1i_2i_3}$, we obtain the following relations
\begin{eqnarray}\label{ststate2}
   (-1)^{l_1+l_2}(-1)^{i_1+i_2}|\phi^l_{i_1i_2i_3}\rangle=|\phi^l_{i_1i_2i_3}\rangle \qquad\mbox{and} \qquad (-1)^{l_1+l_3}(-1)^{i_1+i_3}|\phi^l_{i_1i_2i_3}\rangle=|\phi^l_{i_1i_2i_3}\rangle,
\end{eqnarray}
which allow us to conclude that $|\phi^l_{i_1i_2i_3}\rangle=0$ whenever $(l_1+l_{2}+i_1+i_2)$ mod $2=1$ and $(l_1+l_{3}+i_1+i_3)$ mod $2=1$. 
Thus, the state in Eq. \eqref{genstate4} that satisfies the conditions in Eq. \eqref{ststate2} must be of the form
\begin{eqnarray}\label{genstate5}
   \ket{\tilde{\psi}_l}=\ket{l_1l_2 l_3}|\phi_{l_1l_2 l_3}\rangle+|\overline{l}_1\overline{l}_2 \overline{l}_3\rangle|\phi_{\overline{l}_1\overline{l}_2 \overline{l}_3}\rangle
\end{eqnarray}
where $l_i=0,1$ for any $i=1,2,3$ and $\overline{l}_i=1-l_i$. Now, putting this state in Eq. \eqref{SOSrel10}, we obtain the following relation
\begin{eqnarray}
   (-1)^{l_1}|\overline{l}_1\overline{l}_2 \overline{l}_3\rangle\ket{\phi_{l_1l_2l_3}}+ (-1)^{l_1}\ket{l_1l_2 l_3}|\phi_{\overline{l}_1\overline{l}_2\overline{l}_3}\rangle=\ket{l_1l_2 l_3}\ket{\phi_{l_1l_2 l_3}}+|\overline{l}_1\overline{l}_2 \overline{l}_3\rangle|\phi_{\overline{l}_1\overline{l}_2 \overline{l}_3}\rangle,
\end{eqnarray}
which implies that
\begin{eqnarray}
  (-1)^{l_1}\ket{\phi_{l_1l_2 l_3}}=|\phi_{\overline{l}_1\overline{l}_2\overline{l}_3}\rangle.
\end{eqnarray}
Thus, from Eq. \eqref{genstate5} we conclude that the state $\ket{\tilde{\psi}_l}$, by putting the appropriate normalisation constant, is given by
\begin{eqnarray}
   \ket{\tilde{\psi}_l}=\frac{1}{\sqrt{2}}\left[\ket{l_1l_2 l_3}+(-1)^{l_1}|\overline{l}_1\overline{l}_2 \overline{l}_3\rangle\right]_{A'B'C'}\otimes\ket{\phi_{l_1l_2l_3}}_{A''B''C''E}.
\end{eqnarray}
Tracing out the ancillary subsystem $E$, we finally obtain 
\begin{eqnarray}\label{projstate1}
U_{A}\otimes U_B\otimes U_C\,\rho^l_{ABC}\,
(U_{A}\otimes U_B\otimes U_C)^{\dagger}=\proj{\phi_l}_{A'B'C'}\otimes\tilde{\rho}^{l}_{A''B''C''},
\end{eqnarray}
where $\tilde{\rho}^l_{A''B''C''}$ denotes the auxiliary state acting on $\mathcal{H}_{A''}\otimes \mathcal{H}_{B''}\otimes\mathcal{H}_{C''}$. Putting the above relation back into Eq. \eqref{54} and taking the unitaries to the right-hand side, we see that
\begin{eqnarray}\label{cond1}
\proj{\phi_l}_{A'B'C'}\otimes\tilde{\rho}^l_{A''B'' C''}=8\ \Tr_{\overline{ABC}}\left[\left(\I_{ABC}\otimes R_l\right)\bigotimes_{s=A,B,C}\proj{\tilde{\psi}_{s\overline{s}}}\right],
\end{eqnarray}
where we denoted $\ket{\tilde{\psi}_{s\overline{s}}}= U_{s}\ket{\psi_{s\overline{s}}}$ and used the fact that $\overline{P}(l)=1/8$. \\

\noindent\textbf{(b) States $\ket{\psi_{s\overline{s}}}$ 
    generated by the sources.}\label{SGS}
Let us now exploit the relations in Eq. (\ref{cond1}) to determine the form of the the states $\ket{\tilde{\psi}_{s\overline{s}}}$ $(s=A,B,C)$. To this end, we first express them using their Schmidt decompositions, 
\begin{eqnarray}\label{thestate}
\ket{\tilde{\psi}_{s\overline{s}}}=\sum_{i=0}^{d_{s}-1}\alpha_i^s\ket{e_i}_{s}\ket{f_i}_{\overline{s}},
\end{eqnarray}
where $d_{s}$ denotes the dimension of the Hilbert space $\mathcal{H}_{s}$ and $\{\ket{e_i}_{s}\}$ and $\{\ket{f_i}_{\overline{s}}\}$ are some local bases corresponding to the subsystems $s$ and $\overline{s}$, respectively; recall that, as proven above, $d_{s}$ is even for any $s$ as the Hilbert spaces of the external parties decompose as $\mathcal{H}_{s}=\mathbbm{C}^2\otimes\mathcal{H}_{s''}$. Moreover, the Schmidt coefficients satisfy $\alpha_i^s> 0$ and $\sum_i(\alpha_i^s)^2=1$.

Let us now consider a unitary operation $V_{\overline{s}}:\mathcal{H}_{\overline{s}}\rightarrow\mathcal{H}_{\overline{s}}$ such that $V_{\overline{s}}\ket{f_i}_{\overline{s}}=|e_i^*\rangle_{\overline{s}}$ for any $i$, where the asterisk stands for the complex conjugation. By using this unitary we can bring Eq.  \eqref{thestate} to the following form
%
\begin{eqnarray}\label{genstate2}
\ket{\overline{\psi}_{s\overline{s}}}=\I_{s}\otimes V_{\overline{s}}\ket{\tilde{\psi}_{s\overline{s}}}=\sum_{i=0}^{d_{s}-1}\alpha^s_i\ket{e_i}_{s}\ket{e_i^*}_{\overline{s}}.
\end{eqnarray}
Introducing then the following matrix
\begin{eqnarray}
P_{\overline{s}}=\sum_{i=0}^{d_{s}-1}\alpha_i^s\proj{e_i^*}_{\overline{s}},
\end{eqnarray}
we can rewrite the state in Eq. \eqref{genstate2} as
\begin{eqnarray}\label{genstate3}
\ket{\overline{\psi}_{s\overline{s}}}=\sqrt{d_{s}}\ (\I_{s}\otimes P_{\overline{s}})\,|\phi^{+}_{d_{s}}\rangle_{s\overline{s}}\qquad (s=A,B,C),
\end{eqnarray}
where $|\phi^{+}_{d_s}\rangle_{s\overline{s}}$ denotes the maximally entangled state of local dimension $d_{s}$, that is,
\begin{eqnarray}\label{dupa3}
|\phi^{+}_{d_s}\rangle_{s\overline{s}}=\frac{1}{\sqrt{d_{s}}}\sum_{i=0}^{d_{s}-1}\ket{e_i}_{s}\ket{e_i^*}_{\overline{s}}=\frac{1}{\sqrt{d_{s}}}\sum_{i=0}^{d_{s}-1}\ket{i}_{s}\ket{i}_{\overline{s}},
\end{eqnarray}
where $\ket{e_i^*}$ is the complex conjugate of $\ket{e_i}$ and the last equality follows from the fact that the maximally entangled state is invariant under the action of $U\otimes U^{*}$ for any unitary operation $U$. 

Putting now Eq. (\ref{dupa3}) back into Eq. \eqref{cond1}, we conclude that
\begin{eqnarray}\label{cond2}
\proj{\phi_l}_{A'B'C'}\otimes\tilde{\rho}^l_{A''B''C''}=8\Tr_{E}\left[\left(\I_{ABC}\otimes \overline{R}_l\right)\bigotimes_{s}\proj{\phi^{+}_{d_s}}_{s\overline{s}}\right],
\end{eqnarray}
where we denoted 
\begin{eqnarray}\label{rl}
\overline{R}_l=d_{A}d_Bd_C\left(\bigotimes_{s}P_{\overline{s}}\,V_{\overline{s}}\right)\ R_l\ \left(\bigotimes_{s} V_{\overline{s}}^{\dagger}P_{\overline{s}}\right).
\end{eqnarray}
Now, notice that the tensor product of the three maximally entangled states appearing in (\ref{cond2}) can also be understood as a single maximally entangled state across the bipartition $ABC|\overline{A}\overline{B}\overline{C}$ whose local dimension is $d_Ad_Bd_C$, that is,
\begin{eqnarray}\label{45}
|\phi^{+}_{d_{s}}\rangle_{A\overline{A}}\otimes |\phi^{+}_{d_{s}}\rangle_{B\overline{B}}\otimes |\phi^{+}_{d_{s}}\rangle_{C\overline{C}}
%
=|\phi^{+}_{d_{A}d_{B} d_{C}}\rangle_{ABC|\overline{A}\overline{B}\overline{C}},
\end{eqnarray}
\begin{eqnarray}
|\phi^{+}_{d_{A}d_Bd_C}\rangle_{ABC|\overline{ABC}}=\frac{1}{\sqrt{d_{A}d_Bd_C}}\sum_{i=0}^{d_{A}d_Bd_C-1}\ket{i}_{ABC}\ket{i}_{\overline{ABC}}.
\end{eqnarray}

Substituting Eq. \eqref{45} in Eq. \eqref{cond2} and using the well known fact that $\I_A\otimes Q_{B}|\phi^{+}_{D}\rangle_{A|B}=Q^T_A\otimes\I_B|\phi^{+}_{D}\rangle_{A|B}$ for any matrix $Q$, we can rewrite \eqref{cond2} as
\begin{eqnarray}
\proj{\phi_l}_{A'B'C'}\otimes\tilde{\rho}^l_{A''B''C''}=8\Tr_{\overline{ABC}}\left[\left(\overline{R}^T_{l}\right)_{ABC}\otimes\I_{\overline{ABC}}\proj{\phi^{+}_{d_{A}d_Bd_C}}_{ABC|\overline{ABC}}\right],
\end{eqnarray}
which implies that
\begin{eqnarray}\label{cond3}
\proj{\phi_l}_{A'B'C'}\otimes\tilde{\rho}^l_{A''B''C''}=\frac{8}{d_{A}d_Bd_C}\overline{R}^T_{l}.
\end{eqnarray}
Now, applying the transposition map to both sides and then using Eq. \eqref{rl}, we arrive at
\begin{eqnarray}\label{cond4}
\proj{\phi_l}_{A'B'C'}\otimes\left(\tilde{\rho}^{l}_{A''B''C''}\right)^T=8\left(\bigotimes_{s}P_{s}V_{\overline{s}}\right)\ R_l\ \left(\bigotimes_{s}V_{\overline{s}}^{\dagger} P_{s}\right),
\end{eqnarray}
which after summing over $l$ and 
using the fact that $\sum_{l=0}^{7}R_l=\I$ gives 
\begin{eqnarray}\label{dupa7}
\sum_{l=0}^{7}\proj{\phi_l}_{A'B'C'}\otimes\left(\tilde{\rho}^{l}_{A''B''C''}\right)^T=8\,P_{A}^2\otimes P_B^2\otimes P_C^2.
\end{eqnarray}

We can now employ the fact that $\Tr(P_B^2)=\Tr(P^2_C)=1$ to 
take a partial trace of the above expression over the subsystems $BC$ and obtain
\begin{eqnarray}
\frac{\I_{A'}}{2}\otimes\sigma_{A''}=P_{A}^2,
\end{eqnarray}
which directly implies that 
\begin{eqnarray} P_{A}=\I_{A'}\otimes\sqrt{\frac{\sigma_{A''}}{2}},
\end{eqnarray}
where
\begin{eqnarray}
 \sigma_{A''}=\frac{1}{8}\sum_{l=0}^{7}\Tr_{B''C''}\left[\left(\tilde{\rho}^{l}_{A''B''C''}\right)^T\right].
\end{eqnarray}

Analogously we can determine the other matrices $P_s$ with $s=B,C$. Precisely, by taking partial traces of Eq. (\ref{dupa7}) 
over all subsystems except the $s$th one, one obtains
\begin{eqnarray}\label{Pj}
P_s=\I_{s'}\otimes\sqrt{\frac{\sigma_{s''}}{2}}\qquad (s=A,B,C)
\end{eqnarray}
with
\begin{eqnarray}\label{sigma1}
 \sigma_{s''}=\frac{1}{8}\sum_{l=0}^{7}\Tr_{A''B''C'' \setminus\{s''\}}\left[\left(\tilde{\rho}^{l}_{A''B''C''}\right)^T\right],
\end{eqnarray}
where $\Tr_{A''B''C''\setminus\{s''\}}$ represents the partial trace over the systems $A''B''C''$ except the one labelled by $s''$; For instance $\Tr_{A''B''C''\setminus\{A''\}}\equiv\Tr_{B''C''}$. 

Now, we can substitute $P_{s}$ given by Eq. (\ref{Pj}) into Eq. \eqref{genstate3} and also use the fact that $d_{s}=2d_s''$ for some positive integer $d_s''$ to obtain 
\begin{eqnarray}
 \ket{\overline{\psi}_{s\overline{s}}}=\sqrt{d_{s}''}\ \I_{s'\overline{s}'}\otimes \I_{s''}\otimes\sqrt{\sigma_{\overline{s}''}}\,|\phi^{+}_{d_{s}}\rangle_{s\overline{s}}
\end{eqnarray}
for any $s$. Using again the fact that $|\phi^{+}_{d_{s}}\rangle_{s\overline{s}}=|\phi^{+}\rangle_{s'\overline{s}'}|\phi^{+}_{d_{s}''}\rangle_{s''\overline{s}''}$, we finally conclude that 
the states generated by the sources admit the following form
\begin{eqnarray}
  U_{s}\otimes V_{\overline{s}}\ket{\psi_{s\overline{s}}}=\ket{\overline{\psi}_{s\overline{s}}}=|\phi^{+}\rangle_{s'\overline{s}'}\ket{\xi_{s''\overline{s}''}}\qquad (s=A,B,C),
\end{eqnarray}
where the auxiliary state
$\ket{\xi_{s''\overline{s}''}}$ is given by
\begin{eqnarray}\label{junkst1}
  \ket{\xi_{s''\overline{s}''}}= \sqrt{d_s''}\,\I_{s''}\otimes\sqrt{\sigma_{\overline{s}''}}\,|\phi^{+}_{d_{s}''}\rangle_{s''\overline{s}''}.
\end{eqnarray}
\ \\

\noindent\textbf{(c) Entangled measurement  $E_0$.} Let us now characterise the measurement $E_0=\{R_l\}$. For this purpose, we notice 
that the states $\sigma_{s''}$ are invertible because we assumed 
all the reduced density matrices of the joint states to be full rank. 
This implies that the matrices $P_s$ [cf. Eq. (\ref{Pj})] are invertible too and therefore we can act with $P_{s}^{-1}$ on both sides of  Eq. (\ref{cond4})
to bring to the following form

\begin{eqnarray}
  \left(\bigotimes_{s}V_{\overline{s}}\right)\ R_l\ \left(\bigotimes_{s} V_{\overline{s}}^{\dagger}\right)=\proj{\phi_l}_{A'B'C'}\otimes\left(\bigotimes_{s}\sigma^{-1/2}_{s''}\right)\tilde{\rho}^{l,T}_{A''B''C''}\left(\bigotimes_{s}\sigma^{-1/2}_{s''}\right).
\end{eqnarray}
We thus conclude that
\begin{eqnarray}\label{Rl}
 \left(\bigotimes_{s}V_{\overline{s}}\right)\ R_l\ \left(\bigotimes_{s} V_{\overline{s}}^{\dagger}\right)=\proj{\phi_l}_{A'B'C'}\otimes\left(\tilde{R}_{l}\right)_{A''B''C''},
\end{eqnarray}
for all $l=0,\ldots,7$, where $\tilde{R}_{l}$ is defined as
\begin{eqnarray}\label{Rl1}
\tilde{R}_{l}=\left(\bigotimes_{s}\sigma^{-1/2}_{s''}\right)\left(\tilde{\rho}^{l}_{A''B''C''}\right)^T\left(\bigotimes_{s}\sigma^{-1/2}_{s''}\right).
\end{eqnarray}
Notice that $\tilde{R}_l\geq 0$. Moreover, the fact that $R_l\leq\I$ implies via Eq. \eqref{Rl} that $ \tilde{R}_l\leq\I$. Taking then the sum over $l$ on both sides of Eq. \eqref{Rl} and employing the fact that $\sum_lR_l=\I$ allows us to conclude that
\begin{eqnarray}
  \I_{ABC}=\sum_{l=0}^{7}\proj{\phi_l}_{A'B'C'}\otimes\left(\tilde{R}_{l}\right)_{A''B''C''}.
\end{eqnarray}
Now, we can take advantage of the fact that 
the GHZ-like states are mutually orthogonal 
and therefore by projecting the $A'B'C'$
subsystems onto $\ket{\phi_k}$ we 
deduce that for every $k$,
\begin{eqnarray}\label{Rl2}
  \tilde{R}_{k}=\I_{A''B''C''}.
\end{eqnarray}
After plugging the above relation into Eq. \eqref{Rl} we finally conclude that there exist local unitary transformations such that the measurement operators $R_l$ admit the following form
\begin{eqnarray}
 (V_{\overline{A}}\otimes V_{\overline{B}}\otimes V_{\overline{C}})\,  R_{l|0}\,(V_{\overline{A}}\otimes V_{\overline{B}}\otimes V_{\overline{C}})^{\dagger}=\proj{\phi_l}_{A'B'C'}\otimes\I_{A''B''C''}\qquad (l=0,\ldots,7),
\end{eqnarray}
which is exactly what we promised in Eq. (\ref{stmea1}). This completes the proof.
\end{proof}

An interesting consequence of the above theorem is that the states $\tilde{\rho}^{l}_{A''B''C''}$ in Eq. \eqref{projstate1} are separable for all $l$.
\setcounter{thm}{0}
\begin{cor}
Assume that the sources $P_1, P_2, P_3$ generate states that are certified as in Eq. \eqref{statest1} and the measurements of all the parties are certified as in Eqs. \eqref{stmea1} and \eqref{stmea2}. Then, the post-measurement state $\rho^l_{ABC}$ when Eve gets an outcome $l$ is given by
\begin{eqnarray}\label{projstate2}
\left(\bigotimes_{s} U_{s} \right)\rho^l_{{ABC}}\left(\bigotimes_{s} U_{s}^{\dagger} \right)=\overline{\rho}^l_{{ABC}}=\proj{\phi_l}_{A'B'C'}\otimes\tilde{\rho}_{A''}\otimes\tilde{\rho}_{B''}\otimes\tilde{\rho}_{C''}\qquad \forall l.
\end{eqnarray}
\end{cor}
\begin{proof}
Let us first observe that by combining Eqs. \eqref{Rl1} and \eqref{Rl2} we can write 
\begin{eqnarray}
\left(\bigotimes_{s}\sigma^{-1/2}_{s''}\right)\left(\tilde{\rho}^{l}_{A''B''C''}\right)^T\left(\bigotimes_{s}\sigma^{-1/2}_{s''}\right)=\tilde{R}_{l}=\I_{A''B''C''},
\end{eqnarray}
which after rearranging the terms and taking transposition gives us
\begin{eqnarray}
   \tilde{\rho}^{l}_{A''B'' C''}=\bigotimes_{s}\sigma^T_{s''} \qquad (l=0,\ldots,7).
\end{eqnarray}
Recall that $\sigma_{s''}$ from Eq. \eqref{sigma1} are valid quantum states as they positive semi-definite and satisfy $\Tr(\sigma_{s''})=1$ for any $s=A,B,C$. This completes the proof.
\end{proof}

\subsection{Self-testing of the extra measurements $A_2$, $B_2$ and $C_2$}\label{appendix_extra_measurement}

Now using the above theorem, let us self-test the additional measurements $A_2$, $B_2$ and $C_2$. For this purpose, we consider a functional inspired by the well-known Mermin Bell inequality \cite{Mermin}:
\begin{eqnarray}\label{BE2}
\mathcal{I}_{2}=\left\langle \frac{1}{\sqrt{2}}(A_0+A_1)B_1 C_1-\frac{1}{\sqrt{2}}(A_0+A_1) B_2 C_2-A_2 B_1C_2-A_2B_2 C_1\right\rangle.
\end{eqnarray}
%

\begin{cor}\label{theorem2}
 Assume that the sources $P_1, P_2, P_3$ generate states that are certified as in Eq. \eqref{statest1} and the measurements of all the external parties are certified as in Eqs. \eqref{stmea1} and \eqref{stmea2}. If $\mathcal{I}_2$ achieves the value four for the state shared by $A$, $B$ and $C$ that corresponds to the outcome $l=000$ of Eve's first measurement $E_0$, then the observables $A_2$, $B_2$ and $C_2$ can have two possible forms
\begin{equation}\label{stmea3}
U_{A}\,A_2\,U_{A}^{\dagger}=Y_{A'}\otimes\I_{A''},\quad U_{B}\,B_2\,U_{B}^{\dagger}=Y_{B'}\otimes\I_{B''},\quad
U_{C}\,C_2\,U_{C}^{\dagger}=Y_{C'}\otimes\I_{C''},
\end{equation}
or
\begin{equation}\label{stmea32}
U_{A}\,A_2\,U_{A}^{\dagger}=-Y_{A'}\otimes\I_{A''},\quad U_{B}\,B_2\,U_{B}^{\dagger}=-Y_{B'}\otimes\I_{B''},\quad
U_{C}\,C_2\,U_{C}^{\dagger}=-Y_{C'}\otimes\I_{C''}.
\end{equation}
\end{cor}
\begin{proof}Let us first consider the functional $\mathcal{I}_2$ for the particular observables 
$A_{0}$, $A_1$, $B_1$ and $C_1$ that are given by Eq. \eqref{stmea2} as well as a particular state 
corresponding to the $l=000$ outcome of Eve's first measurement, i.e., $\rho_{ABC}^{000}$,
\begin{eqnarray}\label{BE3}
\mathcal{I}'_{2}&=&\langle X_{A'}\otimes X_{B'}\otimes X_{C'}\otimes\I_{A''B''C''}\rangle- \langle X_{A'}\otimes B'_2\otimes C'_2\otimes\I_{A''}\rangle-\langle A'_2\otimes X_{B'}\otimes C'_2\otimes\I_{B''}\rangle\nonumber\\ 
&&-\langle A'_2\otimes B'_2\otimes X_{C'}\otimes\I_{C''}\rangle_{\overline{\rho}^{000}_{ABC}},
\end{eqnarray}
where $A_2'=U_A\,A_2\,U^{\dagger}_A$ etc., and $\overline{\rho}^{000}_{ABC}$ is a 'rotated version' of $\rho_{ABC}^{000}$ and is given in Eq. \eqref{projstate2}; to simplify the notation from now on we denote $\varrho_0\equiv \overline{\rho}^{000}_{ABC}$.

Now, taking into account the fact that the observables $A'_2$, $B'_2$ and $C'_2$ are unitary it is not difficult to realise that $\mathcal{I}_2'$ attains the value four if and only if the first expectation value in Eq. (\ref{BE3}) equals $1$ whereas the remaining three are $-1$. Again, given that $A'_2$, $B'_2$ and $C'_2$ are unitary this implies that the following conditions are satisfied,
\begin{eqnarray}
 (X_{A'}\otimes X_{B'}\otimes X_{C'}\otimes\I_{A''B''C''})\, \varrho_0&=& \varrho_0,\label{Ystab1}\\
 (X_{A'}\otimes B'_2 \otimes C'_2 \otimes \mathbbm{1}_{A''}) \ \varrho_0 &=&-  \varrho_0,\label{Ystab2}\\
 (A'_2\otimes X_{B'}\otimes C'_2\otimes\I_{B''})\ \varrho_0    &=&-\ \varrho_0,\label{Ystab3}\\
 (A'_2\otimes B'_2\otimes X_{C'}\otimes\I_{C''})\ \varrho_0    &=&-\ \varrho_0.\label{Ystab4}
 \end{eqnarray}
Let us now consider the relation in Eq. \eqref{Ystab2}. After acting on it with $ X_{A'}\otimes X_{B'}\otimes X_{C'}$ and then using Eq. \eqref{Ystab1}, we obtain
\begin{eqnarray}
\left[\I_A\otimes \left( X_{B'}\otimes\I_{B''}\right)B'_2\otimes \left( X_{C'}\otimes\I_{C''}\right)C'_2\right]\ \varrho_0
 &=&-\ \varrho_0.
\end{eqnarray}
Next, after multiplication by $ C'_2\left( X_{C'}\otimes\I_{C''}\right)$, the above relation can be brought to
\begin{eqnarray}\label{Yst1}
\left( X_{B'}\otimes\I_{B''}\right)B'_2\ \varrho_0
 &=&-C'_2\left( X_{C'}\otimes\I_{C''}\right)\ \varrho_0.
\end{eqnarray}
Let us then consider the relation in Eq. \eqref{Ystab1}. After multiplying it with $ X_{A'}\otimes B'_2\otimes C'_2$, using Eq. \eqref{Ystab2}, and then multiplying the resulting relation with $\left( X_{C'}\otimes\I_{C''}\right)C'_2$, we arrive at
\begin{eqnarray}\label{Yst2}
B'_2\left( X_{B'}\otimes\I_{B''}\right) \ \varrho_0
 &=&-\left( X_{C'}\otimes\I_{C''}\right)C'_2\ \varrho_0.
\end{eqnarray}
By adding Eqs. \eqref{Yst1} and \eqref{Yst2}, one obtains
\begin{eqnarray}\label{Yst5}
\{B'_2,\left( X_{B'}\otimes\I_{B''}\right)\}\ \varrho_0=-\{C'_2,\left( X_{C'}\otimes\I_{C''}\right)\}\ \varrho_0.
\end{eqnarray}

Now, we consider Eq. \eqref{Ystab3}. By acting on it with $A'_2\otimes B'_2\otimes X_{C'}$, using Eq. \eqref{Ystab4}, and then multiplying the resulting formula with $ C'_2\left( X_{C'}\otimes\I_{C''}\right)$, we obtain 
\begin{eqnarray}\label{Yst3}
B'_2\left( X_{B'}\otimes\I_{B''}\right)\ \varrho_0
 &=&C'_2\left( X_{C'}\otimes\I_{C''}\right)\ \varrho_0.
\end{eqnarray}
Similarly, we then consider Eq. \eqref{Ystab4}, multiply it with $A'_2\otimes (X_{B'}\otimes\I_{B''})\otimes C'_2$, use Eq. \eqref{Ystab3}, and finally multiply resulting formula with $ (X_{C'}\otimes\I_{C''})C'_2$ to obtain 
\begin{eqnarray}\label{Yst4}
\left( X_{B'}\otimes\I_{B''}\right)B'_2 \ \varrho_0
 &=&\left( X_{C'}\otimes\I_{C''}\right)C'_2\ \varrho_0.
\end{eqnarray}
Adding Eqs. \eqref{Yst3} and \eqref{Yst4}, we get
\begin{eqnarray}\label{Yst6}
\{B'_2,\left( X_{B'}\otimes\I_{B''}\right)\}\ \varrho_0=\{C'_2,\left( X_{C'}\otimes\I_{C''}\right)\}\ \varrho_0
\end{eqnarray}

We have thus obtained two similar relations, \eqref{Yst5} and \eqref{Yst6}, but with the opposite signs. By adding them we thus conclude that 
\begin{eqnarray}
\{B'_2,\left( X_{B'}\otimes\I_{B''}\right)\}\ \varrho_0=0,
\end{eqnarray}
which, by taking into account that all reduced density matrices of $\varrho_0$ are full rank eventually implies that
\begin{eqnarray}
\{B'_2,\left( X_{B'}\otimes\I_{B''}\right)\}=0.
\end{eqnarray}

In the exact same manner, one can use Eqs. (\ref{Ystab1})--(\ref{Ystab4}) to obtain similar relations for the observables $A'_2$, and $C'_2$:
\begin{eqnarray}\label{Ycon1}
\{A'_2,\left( X_{A'}\otimes\I_{A''}\right)\}=0,\qquad\{C'_2,\left( X_{C'}\otimes\I_{C''}\right)\}=0.
\end{eqnarray}

Let us exploit the above anticommutation relations to determine the forms of 
$A'_2$, $B'_2$ and $C'_2$. We begin with $A'_2$. As characterized before, the Hilbert space of Alice is given by $\mathcal{H}_A=\mathbbm{C}^2\otimes\mathcal{H}_{A''}$. Thus, any observable acting on such a Hilbert space can be decomposed as
\begin{eqnarray}
A'_2=\I_{2}\otimes Q_{0}+Z\otimes Q_{1}+X\otimes Q_{2}+Y\otimes Q_{3},
\end{eqnarray}
where for simplicity we have omitted the subscripts $A'$ and $A''$.
Putting it back into Eq. \eqref{Ycon1}, we get that
\begin{eqnarray}
X\otimes Q_{0}+\I\otimes Q_{2}=0,
\end{eqnarray}
which implies that $Q_{0}=Q_{2}=0$, and consequently $A_2$ expresses as
\begin{equation}\label{A2form}
 A'_2=Z\otimes Q_{1}+Y\otimes Q_{3},
\end{equation}
where, due to the fact that $A_2^2=\mathbbm{1}$, the matrices $Q_1$ and $Q_3$ obey 
the following relations
\begin{equation}
    Q_1^2+Q_3^2=\mathbbm{1},\qquad [Q_1,Q_3]=0.
\end{equation}

One can similarly find that 
\begin{equation}
 B'_2=Z\otimes R_{1}+Y\otimes R_{3}  ,\qquad C'_2=Z\otimes S_{1}+Y\otimes S_{3}   
\end{equation}
for some matrices $R_1$, $R_3$, $S_1$ and $S_3$ such that $R_1^2+R_3^2=\mathbbm{1}$
and $[R_1,R_3]=0$, and $S_1^2+S_3^2=\mathbbm{1}$ and $[S_1,S_3]=0$.

Now, after putting these forms of the observables into Eq. \eqref{Ystab2}, we get
\begin{eqnarray}
 [X_{A'}\otimes \left(Z_{B'}\otimes R_{1,B''}+Y_{B'}\otimes R_{3,B''}\right)\otimes \left(Z_{C'}\otimes S_{1,C''}+Y_{C'}\otimes S_{3,C''}\right)\otimes\I_{A''}]\ \varrho_0
 &=&-\ \varrho_0,
\end{eqnarray}
which on expansion and substituting the state $\tilde{\rho}_{ABC}^{000}$ from Eq. \eqref{projstate2} gives
\begin{eqnarray}\label{Yst8}
\left[ X_{A'}Z_{B'}Z_{C'}\otimes \I_{A''}R_{1,B''}S_{1,C''}+X_{A'}Z_{B'}Y_{C'}\otimes \I_{A''}R_{1,B''}S_{3,C''}+X_{A'}Y_{B'}Z_{C'}\otimes \I_{A''}R_{3,B''}S_{1,C''}\right.\nonumber\\\left.+X_{A'}Y_{B'}Y_{C'}\otimes \I_{A''}R_{3,B''}S_{3,C''}\right]\ \proj{\phi_{0}}_{A'B'C'}\otimes\tilde{\rho}_{A''}\tilde{\rho}_{B''}\tilde{\rho}_{C''} 
 =-\proj{\phi_{0}}_{A'B'C'}\otimes\tilde{\rho}_{A''}\tilde{\rho}_{B''}\tilde{\rho}_{C''} ,
\end{eqnarray}
where for simplicity, we are representing the index $000$ as $0$ and the symbol of the tensor products are removed.
Notice that the following relations hold true
\begin{subequations}\label{Yst7}
\begin{eqnarray}
X_{A'}Z_{B'}Z_{C'}\ket{\phi_{000}}_{A'B'C'}&=&\ket{\phi_{011}}_{A'B'C'},\\
X_{A'}Z_{B'}Y_{C'}\ket{\phi_{000}}_{A'B'C'}&=&\mathbbm{i}\ket{\phi_{010}}_{A'B'C'},\\
X_{A'}Y_{B'}Z_{C'}\ket{\phi_{000}}_{A'B'C'}&=&\mathbbm{i}\ket{\phi_{001}}_{A'B'C'},\\
X_{A'}Y_{B'}Y_{C'}\ket{\phi_{000}}_{A'B'C'}&=&-\ket{\phi_{000}}_{A'B'C'}
\end{eqnarray}
\end{subequations}
where $\ket{\phi_l}$ for any $l$ can be found in Eqs. \eqref{GHZvecs}. Using these relations and the condition in Eq. \eqref{Yst8}, we get four different relations
\begin{eqnarray}\label{Yst9}
R_{1,B''}\tilde{\rho}_{B''}\otimes S_{1,C''}\tilde{\rho}_{C''}=0,\qquad R_{1,B''}\tilde{\rho}_{B''}\otimes S_{3,C''}\tilde{\rho}_{C''}=0,\qquad
R_{3,B''}\tilde{\rho}_{B''}\otimes S_{1,C''}\tilde{\rho}_{C''}=0,
\end{eqnarray}
and
\begin{eqnarray}\label{Yst10}
\tilde{\rho}_{A''}\otimes R_{3,B''}\tilde{\rho}_{B''}\otimes S_{3,C''}\tilde{\rho}_{C''}=\tilde{\rho}_{A''}\otimes\tilde{\rho}_{B''}\otimes \tilde{\rho}_{C''}.
\end{eqnarray}
All three reduced density matrices $\tilde{\rho}_{s''}$ $(s=A,B,C)$ are full rank, it follows from Eq. (\ref{Yst10}) that $R_3$ and $S_3$ are non-zero. 
Consequently, the last two relations in Eq. (\ref{Yst9}) imply that 
$S_1=0$ and $R_1=0$. 
Analogously, after plugging $A_2$ as given in Eq. (\ref{A2form}) and $C_2=Y\otimes S_3$ into Eq. (\ref{Ystab3}) we infer that $Q_{1}=0$ and
\begin{eqnarray}\label{Yst12}
 Q_{3,A''}\tilde{\rho}_{A''}\otimes S_{3,C''}\tilde{\rho}_{C''}=\tilde{\rho}_{A''}\otimes \tilde{\rho}_{C''}.
\end{eqnarray}
Thus, we obtain that
\begin{eqnarray}
A'_2=Y\otimes Q,\qquad B'_2=Y\otimes R,\qquad C'_2=Y\otimes S,
\end{eqnarray}
where $Q$, $R$ and $S$ are some hermitian matrices such that $Q^2=\mathbbm{1}$, 
$R^2=\mathbbm{1}$ and $S^2=\mathbbm{1}$, which makes them also unitary; for simplicity we dropped the subscripts from them. 

Let us now exploit the fact that $Q$, $R$ and $S$ are hermitian and square to the identity to decompose them as 
\begin{equation*}
Q=Q_{+}-Q_-, \qquad R=R_{+}-R_-,\qquad S=S_+-S_- ,   
\end{equation*}
where $Q_{\pm}$, $R_{\pm}$ and $S_{\pm}$ are projectors onto the
eigenspaces of $A'_2$, $B'_2$ and $C'_2$ corresponding to eigenvalues $\pm1$. 
Now using the fact that $Q^2=Q_++Q_-=\mathbb{1}$ and $S^2=S_++S_-=\mathbb{1}$ and tracing the $A''$ out, we obtain from Eqs. \eqref{Yst10} that
 \begin{eqnarray}\label{Yst101}
2(R_{+}\tilde{\rho}_{B''})\otimes (S_{+}\tilde{\rho}_{C''})=(R_+\tilde{\rho}_{B''})\otimes\tilde{\rho}_{C''}+
\tilde{\rho}_{B''}\otimes(S_+\tilde{\rho}_{C''}).
\end{eqnarray}   
One concludes from the above relation that $(R_+\tilde{\rho}_{B''})\otimes(S_-\tilde{\rho}_{C''})=0$ and $(R_-\tilde{\rho}_{B''})\otimes(S_+\tilde{\rho}_{C''})=0$, which by taking into 
account the fact that both $\tilde{\rho}_{B''}$ and $\tilde{\rho}_{C''}$ are 
full rank, implies that either $R_+=0$ and $S_+=0$ or $R_-=0$ and $S_-=0$.
In the same spirit, one can exploit the second relation in Eq. \eqref{Yst12} to 
conclude that either $Q_+=0$ and $S_+=0$ or $Q_-=0$ and $S_-=0$. Taking into account all the possibilities listed above, one deduces that either 
$Q_-=R_-=S_-=0$ (in which case $Q_+=\mathbbm{1}_{A''}$, $R_+=\mathbbm{1}_{B''}$ and $S_+=\mathbbm{1}_{C''}$) or $Q_+=R_+=S_+=0$ (in which case $Q_-=\mathbbm{1}_{A''}$, $R_-=\mathbbm{1}_{B''}$ and $S_-=\mathbbm{1}_{C''}$), which directly leads us to 
two possible forms that the observables $A'_2$, $B'_2$ and $C'_2$ can take:
\begin{eqnarray}
A_2=Y_{A'}\otimes\I_{A''},\qquad B_2=Y_{B'}\otimes\I_{B''},\qquad C_2=Y_{C'}\otimes\I_{C''}
\end{eqnarray} 
or
\begin{eqnarray}
A_2=-Y_{A'}\otimes\I_{A''},\qquad B_2=-Y_{B'}\otimes\I_{B''},\qquad C_2=-Y_{C'}\otimes\I_{C''},
\end{eqnarray} 
from which one recovers Eqs. (\ref{stmea3}) and (\ref{stmea32}). This completes the proof.

\end{proof}

\section{Self-testing the three-qubit NLWE basis}
\label{AppC}

In this section, we show that using the certified states in Eq. \eqref{statest1} generated by the sources $P_i$ $(i=1,2,3)$ and measurements in Eqs. \eqref{stmea1} and \eqref{stmea2} along with some additional statistics, one can self-test the measurement corresponding to the input $e=1$ with the central party to be NLWE basis given below in Eq. \eqref{NLWEbasis} upto some additional degrees of freedom. 

Before proceeding, let us denote the eigenvectors of $(X+Z)/\sqrt{2}$ as
\begin{eqnarray}\label{X+Z}
\ket{\overline{0}}=\cos(\pi/8)\ket{0}+\sin(\pi/8)\ket{1},\qquad \ket{\overline{1}}=-\sin(\pi/8)\ket{0}+\cos(\pi/8)\ket{1},
\end{eqnarray}
and the eigenvectors of
$(X-Z)/\sqrt{2}$ as 
\begin{eqnarray}\label{X-Z}
\ket{\overline{+}}=\sin(\pi/8)\ket{0}+\cos(\pi/8)\ket{1},\qquad \ket{\overline{-}}=\cos(\pi/8)\ket{0}-\sin(\pi/8)\ket{1}.
\end{eqnarray}
Recall also that the NLWE measurement $M_{\mathrm{NLWE}}=\{\proj{\delta_l}\}$ is defined via the following fully product vectors 
\begin{equation}\label{NLWEbasis}
\begin{split}
\ket{\delta_0}&=\ket{\overline{0}}\ket{1}\ket{+},\qquad \ket{\delta_1}=\ket{\overline{0}}\ket{1}\ket{-},\\
\ket{\delta_2}&=\ket{\overline{+}}\ket{0}\ket{1},\qquad \ket{\delta_3}=\ket{\overline{-}}\ket{0}\ket{1},\\
\ket{\delta_4}&=\ket{\overline{1}}\ket{+}\ket{0},\qquad \ket{\delta_5}=\ket{\overline{1}}\ket{-}\ket{0},\\
\ket{\delta_6}&=\ket{\overline{0}}\ket{0}\ket{0},\qquad\  \ket{\delta_7}=\ket{\overline{1}}\ket{1}\ket{1}.
\end{split}
\end{equation}
Notice that the above vectors are equivalent to the standard NLWE vectors introduced in Ref. \cite{Bennett99} up to a local unitary transformation applied to the first qubit.

To certify the second Eve's measurement $E_1$, the correlations observed by the parties $\{p(a,b,c,l|x,y,z,1)\}$ must satisfy the following conditions,
\begin{equation}\label{NLWEstat}
\begin{split}
p(0,1,0,0|0,0,1,1)=\frac{1}{8},\qquad p(0,1,1,1|0,0,1,1)=\frac{1}{8},\\
p(0,0,1,2|1,0,0,1)=\frac{1}{8},\qquad p(1,0,1,3|1,0,0,1)=\frac{1}{8},\\
p(1,0,0,4|0,1,0,1)=\frac{1}{8},\qquad p(1,1,0,5|0,1,0,1)=\frac{1}{8},\\
p(0,0,0,6|0,0,0,1)=\frac{1}{8},\qquad p(1,1,1,7|0,0,0,1)=\frac{1}{8}.
\end{split}
\end{equation}
Notice that the above distribution can be realized if the sources $P_i$ $(i=1,2,3)$ generate the maximally entangled state of two qubits $\ket{\phi^+}=(1/\sqrt{2})(\ket{00}+\ket{11})$ and the measurement $E_1$ is exactly $M_{\mathrm{NLWE}}=\{\proj{\delta_l}\}$, whereas the external parties perform the following measurements
\begin{eqnarray}
A_0=\frac{X+Z}{\sqrt{2}},\quad A_1= \frac{X-Z}{\sqrt{2}},\quad B_0=Z,\quad B_1=X,\quad C_0=Z,\quad C_1=X.
\end{eqnarray}

Next, we show that the above probabilities along with certification of the states and measurements presented in Theorem \ref{theorem1} are sufficient to fully characterize the unknown measurement $E_{1}=\{R_{l|1}\}$ where $R_l$ denotes the measurement element corresponding to outcome $l$. The theorem stated below is labeled as Theorem 2 in the manuscript.

\setcounter{thm}{1}
\begin{thm}\label{theorem3}
Assume that the correlations $\vec{p}$ generated in the
network satisfy the assumptions of Theorem \ref{theorem1} as well as the conditions in Eq. (\ref{NLWEstat}). Then, for any $l$ it holds that 
\begin{equation}\label{stmeaNLWE1}
 (V_{\overline{A}}\otimes V_{\overline{B}}\otimes V_{\overline{C}})\,  R_{l|1}\,(V_{\overline{A}}\otimes V_{\overline{B}}\otimes V_{\overline{C}})^{\dagger} =|\delta_l \rangle
\!\langle\delta_l |_{E'} \otimes \mathbbm{1}_{E''},   
\end{equation}
where $V_{\overline{s}}$ are the same unitary as in Theorem \ref{theorem1}
and $E=\overline{A}\overline{B}\overline{C}$.
\end{thm}

\begin{proof}
For simplicity, we represent $R_{l|1}$ as $R_{l}$ throughout the proof. Let us first consider the first relation in Eq. (\ref{NLWEstat}), that is, 
\begin{eqnarray}
p(0,1,0,0|0,0,1,1)=\frac{1}{8}
\end{eqnarray}
and then expand it by using the fact that the observables $A_i,B_i,C_i$ for $i=0,1$ are certified as in Eq. \eqref{stmea2}. This gives us
\begin{eqnarray}\label{stcond1NLWE}
\left\langle\left(\bigotimes_{s=A,B,C} U_{s}^{\dagger}\right) \left(\proj{\overline{0}}_{A'}\otimes\proj{1}_{B'}\otimes\proj{+}_{C'}\otimes\I_{A''B''C''}\right)\left(\bigotimes_{s=A,B,C} U_{s}\right)\otimes R_0\right\rangle_{\psi_{ABCE}}=\frac{1}{8}.
\end{eqnarray}
Let us then use Theorem \ref{theorem1} to represent the global state $\ket{\psi_{ABCE}}$ as [cf. Eq. \eqref{statest1}] 
\begin{eqnarray}\label{state1}
\ket{\psi_{ABCE}}=\bigotimes_s\ket{\psi_{s\overline{s}}}=\bigotimes_sU_s^{\dagger}\otimes V_{\overline{s}}^{\dagger}|\phi^+_{s'\overline{s}'}\rangle\otimes\ket{\xi_{s''\overline{s}''}}.
\end{eqnarray}
Notice also that by virtue of Eq. \eqref{junkst1} the junk states $\ket{\xi_{s''\overline{s}''}}$ can be represented as
\begin{eqnarray}\label{junkst2}
 \ket{\xi_{s''\overline{s}''}}= (\I_{s''}\otimes P_{\overline{s}''} )|\phi^{+}_{d_{s}''}\rangle_{s''\overline{s}''},
\end{eqnarray}
where $P_{\overline{s}''}=\sqrt{d_s''\sigma_{\overline{s}''}}$.
The joint state in Eq. \eqref{state1} can be further written as [c.f. Eq. \eqref{45}],
\begin{eqnarray}\label{115}
\left(\bigotimes_{s=A,B,C}U_s\otimes V_{\overline{s}}\right)|\psi_{ABC|E}\rangle=\left(P_{\overline{A}''}\otimes P_{\overline{B}''}\otimes P_{\overline{C}''}\right)|\phi^{+}_{8d_A''d_B''d_C''}\rangle_{ABC|E}
\end{eqnarray}
where the bipartition is between the subsystem $ABC$ and $E\equiv\overline{A}\overline{B}\overline{C}$ and the local dimension of the state is $8d_A''d_B''d_C''$. After plugging this state into the Eq. \eqref{stcond1NLWE}, and then using the fact that $(\I\otimes Q)|\phi^{+}_{d}\rangle=(Q^T\otimes\I)|\phi^{+}_{d}\rangle$ for any matrix $Q$, we get that
\begin{eqnarray}\label{116}
\left\langle\left(\proj{\overline{0}}\otimes\proj{1}\otimes\proj{+}\otimes \left(P_{A''}^T\right)^2\otimes \left(P_{B''}^T\right)^2\otimes \left(P_{C''}^T\right)^2\right)\overline{R}_0^T\otimes\I_{E}\right\rangle_{|\phi^{+}_{d_Ad_Bd_C}\rangle_{ABC|E}}=\frac{1}{8},
\end{eqnarray}
where 
\begin{eqnarray}\label{R0eq}
\overline{R}_0=\left(\bigotimes_{s=A,B,C} V_{\overline{s}} \right) R_0\left(\bigotimes_{s=A,B,C} V_{\overline{s}}^{\dagger}\right).
\end{eqnarray}
Expanding the above term, we arrive at 
\begin{eqnarray}
\frac{1}{d_A''d_B''d_C''}\Tr\left[\left(\proj{\overline{0}}\otimes\proj{1}\otimes\proj{+}\otimes \left(P_{A''}^T\right)^2\otimes \left(P_{B''}^T\right)^2\otimes \left(P_{C''}^T\right)^2\right)\overline{R}_0^T\right]=1.
\end{eqnarray}
Then using the fact that $P_{\overline{s}''}=\sqrt{d_s''\sigma_{\overline{s}''}}$ for any $s$, we obtain
\begin{eqnarray}\label{nlwest1}
\Tr\left[\left(\proj{\overline{0}}\otimes\proj{1}\otimes\proj{+}\otimes \sigma^T_{A''}\otimes \sigma^T_{B''}\otimes \sigma^T_{C''}\right)\overline{R}_0^T\right]=1.
\end{eqnarray}
Recall that one can characterize measurements only on the support of the state or equivalently the states $\sigma_{s''}$ are full-rank. Since $\overline{R}_0^T$ acts on the Hilbert space $\mathbbm{C}^8\otimes\mathcal{H}_{{A''B''C''}}$, we can express it using the basis given in Eq. \eqref{NLWEbasis} as
\begin{eqnarray}
\overline{R}_0^T=\sum_{l,l'=0}^7\ket{\delta_l}\!\!\bra{\delta_{l'}}\otimes \tilde{R}_{l,l'}
\end{eqnarray}
where $\tilde{R}_{l,l'}$ act on $\mathcal{H}_{{A''B''C''}}$, and, moreover, 
$\tilde{R}_{l,l'}$ are such that $0\leq \tilde{R}_{l,l'}\leq\mathbbm{1}$.
Then, using the cyclic property of trace and Lemma \ref{theorem1.1} we conclude from Eq. \eqref{nlwest1} that $\tilde{R}_{0,0}=\I$ and thus we finally get that
\begin{eqnarray}\label{120}
\overline{R}_0^T&=&\proj{\overline{0}}\otimes\proj{1}\otimes\proj{+}\otimes\I_{A''B''C''}+\mathbbm{L}_0\nonumber\\
&=&\proj{\delta_0}_{A'B'C'}\otimes\I_{A''B''C''}+\mathbbm{L}_0,
\end{eqnarray}
where $\mathbbm{L}_0$ stands for an operator given by
\begin{equation}\label{L0}
\mathbbm{L}_0=\sum_{\substack{l,l'=0\\l\ne0|l'=0\ l'\ne0|l=0\ }}^7\ket{\delta_l}\!\!\bra{\delta_{l'}}\otimes \tilde{R}_{l,l'},
\end{equation}
where $l\ne0|l'=0$ denotes that $l\ne 0$ if $l'=0$. 
Let us now show that $\mathbbm{L}_0$ is positive semi-definite. For this purpose, we consider the state $\ket{\delta_0}\otimes\ket{\xi} $, where $\ket{\xi}$ is an arbitrary state from $\mathcal{H}_{A''B''C''}$, and act on this state with $\overline{R}_0^T$. Taking into account Eqs. \eqref{120} and \eqref{L0}, we obtain
\begin{equation}\label{Lo1}
\overline{R}_0^T\ket{\delta_0}\ket{\xi}=\ket{\delta_0}\ket{\xi}+\sum_{\substack{l=1}}^7\ket{\delta_l}\tilde{R}_{l,0}\ket{\xi}.
\end{equation}
Now, multiplying the above formula in Eq. \eqref{Lo1} with its complex conjugate we obtain that
\begin{eqnarray}
\bra{\delta_0}\bra{\xi}\left(\overline{R}_0^T\right)^2\ket{\delta_0}\ket{\xi}=1+\sum_{\substack{l=1}}^7\bra{\xi}\tilde{R}^{\dagger}_{l,0}\tilde{R}_{l,0}\ket{\xi}
\end{eqnarray}
which after using the fact that $\left(\overline{R}_0^T\right)^2\leq\overline{R}_0^T\leq\I$, gives us
\begin{equation}
\sum_{\substack{l=1}}^7\bra{\xi}\tilde{R}^{\dagger}_{l,0}\tilde{R}_{l,0}\ket{\xi}\leq0.
\end{equation}
As $\tilde{R}^{\dagger}_{l,0}\tilde{R}_{l,0}\geq 0$,  it follows that $\bra{\xi}\tilde{R}^{\dagger}_{l,0}\tilde{R}_{l,0}\ket{\xi}=0$ for any $\ket{\xi}\in\mathcal{H}_{{A''B''C''}}$, and consequently $\tilde{R}_{l,0}=0$ for any $l=1,\ldots,7$. Since $\mathbbm{L}_0$ is hermitian, the above implies
also that $\tilde{R}_{0,l}=0$ for $l=1,\ldots,7$, and hence
\begin{eqnarray}
\mathbbm{L}_0=\sum_{\substack{l,l'=1}}^7\ket{\delta_l}\!\!\bra{\delta_{l'}}\otimes \tilde{R}_{l,l'}.
\end{eqnarray}
This means that $\overline{R}_0^T$ can now be expressed in the block form as
\begin{eqnarray}\label{upR0}
\overline{R}_0^T=\proj{\delta_0}\otimes\mathbbm{1}_{A''B''C''}+\mathbbm{L}_0=
\begin{pmatrix}
1 & 0\\
0 & \mathbbm{L}_0
\end{pmatrix},
\end{eqnarray}
where both components act on orthogonal subspaces of the three-qubit Hilbert space.
Owing to the fact that $\overline{R}_0^T\geq 0$, we thus obtain $\mathbbm{L}_0\geq 0$.

A similar analysis using all the other probabilities in Eq. \eqref{NLWEstat} can be done for the other operators $\overline{R}_l^T$, and thus we arrive at
\begin{eqnarray}\label{stcond2NLWE}
\overline{R}_l^T=\proj{\delta_l}\otimes\I_{A''B''C''}+\mathbbm{L}_l\qquad (l=0,\ldots,7),
\end{eqnarray}
where $\mathbbm{L}_l$ is a positive semi-definite operator that with respect to $A'B'C'$ subsystem is defined a subspace of $\mathbbm{C}^8$ which is orthogonal to $\ket{\delta_l}$. 

Now, the fact that $\sum_l\overline{R}_l^T=\mathbbm{1}$ implies that
\begin{eqnarray}
\sum_l\mathbbm{L}_l=0.
\end{eqnarray}
As each $\mathbbm{L}_l$ is positive semi-definite, the only way the above condition is satisfied is that $\mathbbm{L}_l=0$ for any $l$. Thus, we finally obtain from Eq. \eqref{stcond2NLWE} that
\begin{eqnarray}
\overline{R}_l=\proj{\delta_l}\otimes\I,
\end{eqnarray}
which by taking into account Eq. (\ref{R0eq})
gives the desired result
form Eq. \eqref{stmeaNLWE1}, completing the proof.
\end{proof}

\section{Self-testing the measurement constructed from the UPB}
\label{AppD}

Let us finally provide the proof of Theorem 4 stated in the main text. 
To this end, we recall that the measurement constructed from the UPB
is defined as $M_{\mathrm{UPB}}=\{\proj{\tau_0},\proj{\tau_1},\proj{\tau_2},\proj{\tau_3},\Gamma\}$, where
\begin{equation}\label{UPB}
\begin{split}
\ket{\tau_0}&=\ket{\overline{0}}\ket{1}\ket{+},\qquad \ket{\tau_1}=\ket{\overline{+}}\ket{0}\ket{1}\\
\ket{\tau_2}&=\ket{\overline{1}}\ket{+}\ket{0},\qquad \ket{\tau_3}=\ket{\overline{-}}\ket{-}\ket{-},
\end{split}
\end{equation}
and
\begin{equation}
    \Gamma=\I-\sum_{i=0}^3\proj{\tau_i}.
\end{equation}
Notice that $\ket{\tau_i}$ form a four-element UPB which is obtained from the UPB introduced in \cite{Bennett99-1} by applying a local unitary to Alice's qubit. 

Now, to self-test the five-outcome measurement $E_2$, the correlations observed
by the parties $\{p(a,b,c,l|x,y,z,2)\}$ must satisfy 
%
%
\begin{equation}\label{betstat0}
\begin{split}
p(0100|0012)=\frac{1}{8},\qquad p(0011|1002)=\frac{1}{8},\\
p(1002|0102)=\frac{1}{8},\qquad p(1113|1112)=\frac{1}{8},
\end{split}
\end{equation}
as well as the following conditions that are expressed in the correlation picture:
%
\begin{subequations}\label{betstatfull}
\begin{eqnarray}
\langle\left(\I+A_0B_0+B_0C_0+A_0C_0+A_1B_1C_1-A_1B_2C_2-A_2B_1C_2-A_2B_2C_1\right)R_4\rangle&=&1\label{betstat1}\\
\left\langle(\I+ A_0)\left[-B_1C_1+B_2C_2+B_1(\I+C_0)-\frac{1}{2}(\I-B_0)C_1+(\I+B_0)C_0+\frac{1}{2}B_0+\frac{3}{2}\I\right] R_4\right\rangle&=&\frac{3}{2},\label{betstat2}\\
\left\langle(\I+ B_0)\left[-A_1C_1+A_2C_2+(\I+A_0)C_1-\frac{1}{2}A_1(\I-C_0)+A_0(\I+C_0)+\frac{1}{2}C_0+\frac{3}{2}\I\right] R_4\right\rangle&=&\frac{3}{2},\label{betstat3}\\
\left\langle(\I+ C_0)\left[-A_1B_1+A_2B_2+A_1(\I+B_0)-\frac{1}{2}(\I-A_0)B_1+(\I+A_0)B_0+\frac{1}{2}A_0+\frac{3}{2}\I\right] R_4\right\rangle&=&\frac{3}{2},\label{betstat4}
\end{eqnarray}
\end{subequations}
where $E_2=\{R_l\}_{l=0}^4$.

Notice that the above conditions are met in a situation in which the sources $P_i$ $(i=1,2,3)$ distribute the state $\ket{\phi^+_2}$ and Eve's measurements $E_2$ is the ideal measurement $M_{\mathrm{UPB}}$ given in Eq. \eqref{UPB}, whereas the external parties measure the following observables
\begin{eqnarray}
A_0=\frac{X+Z}{\sqrt{2}},\quad A_1= \frac{X-Z}{\sqrt{2}},\quad A_2=Y,\quad B_0=Z,\quad B_1=X,\quad B_2=Y,\quad C_0=Z,\quad C_1=X\quad C_2=Y.
\end{eqnarray}

In what follows we show that Theorem \ref{theorem1} and Theorem \ref{theorem2} 
together with the conditions in Eqs. (\ref{betstat0}) and (\ref{betstatfull}) enable 
self-testing $M_{\mathrm{UPB}}$ is $E_{2}=\{R_l\}$ where $R_l$ denotes the measurement element corresponding to outcome $l$. The theorem stated below is labeled as Theorem 3 in the manuscript.

\begin{thm}\label{theorem4}
Assume that the correlations $\vec{p}$ observed in the network
satisfy the assumptions of Theorems \ref{theorem1} and \ref{theorem2} as well as conditions in Eqs. 
(\ref{betstat0}) and (\ref{betstatfull}). Then,
\begin{eqnarray}\label{stmeaupb1}
 (V_{\overline{A}}\otimes V_{\overline{B}}\otimes V_{\overline{C}})\, R_{l|2}\,(V_{\overline{A}}\otimes V_{\overline{B}}\otimes V_{\overline{C}} )^{\dagger}=\proj{\tau_l}_{E'}\otimes\I_{E''}\qquad  (l=0,1,2,3),
\end{eqnarray}
and,
\begin{eqnarray}\label{stmeaupb2}
(V_{\overline{A}}\otimes V_{\overline{B}}\otimes V_{\overline{C}})\,R_{4|2}\,(V_{\overline{A}}\otimes V_{\overline{B}}\otimes V_{\overline{C}})=\Gamma_{E'}\otimes\I_{E''},
\end{eqnarray}
where the unitary operations $V_{\overline{s}}$ $(s=A,B,C)$ are the same as in Theorem 1.
\end{thm}
\begin{proof}
For simplicity, we represent $R_{l|2}$ as $R_{l}$ throughout the proof. Let us first consider the conditions in Eq. \eqref{betstat0}. As was done in the previous section in the proof of Theorem \ref{theorem3}, we can conclude from these conditions that for an unknown measurement $\{R_l\}$ we have that
\begin{eqnarray}\label{R0123}
\overline{R}_l^T=\proj{\tau_l}\otimes\I_{A''B''C''}+\mathbbm{L}_l\qquad l=0,1,2,3,
\end{eqnarray}
where
\begin{eqnarray}\label{unitary}
\overline{R}_l=\left(\bigotimes_{s=A,B,C} V_{\overline{s}} \right) R_l\left(\bigotimes_{s=A,B,C} V_{\overline{s}}^{\dagger}\right)\qquad l=0,1,2,3,4
\end{eqnarray}
and $\mathbbm{L}_l\geq0$. Let us now consider Eq. \eqref{betstat1}
\begin{eqnarray}\label{Raimat1}
\langle\left(\I+A_0B_0+B_0C_0+A_0C_0+A_1B_1C_1-A_1B_2C_2-A_2B_1C_2-C_1A_2B_2\right)R_4\rangle=1
\end{eqnarray}
and expand it by using the fact that the observables $A_i,B_i,C_i$  for $i=0,1,2$ are certified as in Eqs. \eqref{stmea2} and \eqref{stmea3}. This gives us
\begin{eqnarray}\label{stcond1}
\left\langle\left(\bigotimes_{s=A,B,C} U_{s}^{\dagger}\right) S_{1,A'B'C'}\otimes\I_{A''B''C''}\left(\bigotimes_{s=A,B,C} U_{s}\right)\otimes R_4\right\rangle=1,
\end{eqnarray}
where 
\begin{eqnarray}\label{stab11}
    S_{1,A'B'C'}=\I_{{A'B'C'}}+\frac{(X+Z)_{A'}}{\sqrt{2}}Z_{B'}+Z_{B'}Z_{C'}+\frac{(X+Z)_{A'}}{\sqrt{2}} Z_{C'}+\frac{(X-Z)_{A'}}{\sqrt{2}}X_{B'}X_{C'}&-&\frac{(X-Z)_{A'}}{\sqrt{2}}Y_{B'}Y_{C'}-Y_{A'}X_{B'}Y_{C'}\nonumber\\&-&Y_{A'}Y_{B'}X_{C'}.
\end{eqnarray}
The states generated by the sources $P_1,P_2,P_3$ have already been certified to be of the form in Eq. \eqref{statest1}. Thus, the joint state of all the parties can be represented as
\begin{eqnarray}
\ket{\psi_{ABCE}}=\bigotimes_s\ket{\psi_{s\overline{s}}}=\bigotimes_s
(U_s^{\dagger}\otimes V_{\overline{s}}^{\dagger})\ket{\phi^+}_{s'\overline{s}'}\ket{\xi_{s''\overline{s}''}}
\end{eqnarray}
which can be simplified to
\begin{eqnarray}\label{138}
\left(\bigotimes_sU_s\otimes V_{\overline{s}}\right)\ket{\psi_{ABC|E}}=\left(P_{\overline{A}''}\otimes P_{\overline{B}''}\otimes P_{\overline{C}''}\right)\ket{\phi^{+}_{8d_A''d_B''d_C''}}_{ABC|E},
\end{eqnarray}
where $P_{\overline{s}''}=\sqrt{d_s''\sigma_{\overline{s}''}}$ [c. f. Eqs. \eqref{state1}-\eqref{115}]. Now, following exactly the same steps from Eqs. \eqref{116}-\eqref{nlwest1}, we obtain from Eq. \eqref{stcond1} that
\begin{eqnarray}\label{stab1}
\frac{1}{8}\Tr\left[ \left(S_{1,A'B'C'}\otimes \sigma^T_{A''}\otimes \sigma^T_{B''}\otimes \sigma^T_{C''}\right) \overline{R}_4^T\right]=1.
\end{eqnarray}
%

Let us now consider a state given by
\begin{eqnarray}\label{psi14}
\ket{\psi_{1,4}}=\frac{1}{\sqrt{2}}(\ket{\overline{0}00}+\ket{\overline{1}11}).
\end{eqnarray}
One checks that  
\begin{equation}\label{S1proj}
S_{1,A'B'C'}=8\proj{\psi_{1,4}}.    
\end{equation}
The above fact also explains why we imposed the condition in Eq. (\ref{Raimat1}).
Taking into account Eq. (\ref{S1proj}) we can rewrite Eq. (\ref{stab1}) as
\begin{eqnarray}\label{D18}
\Tr\left[\left(\proj{\psi_{1,4}}_{A'B'C'}\otimes\sigma_{A''}\otimes \sigma_{B''}\otimes \sigma_{C''}\right)\overline{R}_4\right]=1,
\end{eqnarray}
where we have also used the fact that $\Tr[X^TY^T]=\Tr[XY]$ for any pair
of matrices $X$ and $Y$ and that the state $\ket{\psi_{1,4}}$ is real.

Consider now an orthonormal basis $\{\ket{\varphi_l}\}$ in $\mathbbm{C}^8$ 
in which $\ket{\varphi_l}=\ket{\tau_l}$ for $l=0,1,2,3$, where $\ket{\tau_l}$
form the UPB given in Eq. \eqref{UPB}, and $\ket{\phi_l}=\ket{\psi_{1,l}}$ $(l=4,5,6,7)$, where
$\ket{\psi_{1,4}}$ is given in \eqref{psi14} 
whereas the remaining vectors are defined as
\begin{eqnarray}\label{psi15}
    \ket{\psi_{1,5}}&=&\frac{1}{\sqrt{6}}(-2\ket{\overline{0}00}-\ket{\overline{0}10}+\ket{\overline{0}11}),\nonumber\\ \ket{\psi_{1,6}}&=&\frac{1}{\sqrt{6}}\left(-2\ket{\overline{0}00}-\ket{\overline{0}01}+\ket{\overline{1}01}\right), \nonumber\\ \ket{\psi_{1,7}}&=&\frac{1}{\sqrt{6}}\left(-2\ket{\overline{0}00}-\ket{\overline{1}00}+\ket{\overline{1}10}\right).
\end{eqnarray}
Notice that $\sum_{l=4}^7\proj{\psi_{1,l}}=\Gamma$. Since $\overline{R}_4$ acts on the Hilbert space $\mathbbm{C}^8\otimes\mathcal{H}_{{A''B''C''}}$, we can express it using the above mentioned basis as
\begin{eqnarray}
\overline{R}_0=\sum_{l,l'=0}^7\ket{\varphi_l}\!\!\bra{\varphi_{l'}}\otimes \tilde{R}_{l,l'},
\end{eqnarray}
where $\tilde{R}_{l,l'}$ are some matrices acting on $\mathcal{H}_{{A''B''C''}}$.
Recalling that the local states are full-rank and then using the cyclic property of trace and Lemma \ref{theorem1.1} we conclude from Eq. \eqref{nlwest1} that $\tilde{R}_{0,0}=\I$ and thus we get that
\begin{eqnarray}\label{R41}
\overline{R}_4=\proj{\psi_{1,4}}\otimes\I_{A''B''C''}+\mathbbm{L}_{4},
\end{eqnarray}
where 
\begin{equation}\label{L0upb}
\mathbbm{L}_{4}=\sum_{\substack{l,l'=0\\l\ne4|l'=4\ l'\ne4|l=4}}^7\ket{\varphi_l}\!\!\bra{\varphi_{l'}}\otimes \tilde{R}_{l,l'}
\end{equation}
where $l\ne4|l'=4$ denotes $l\ne4$ if $l'=4$. For simplicity, we denote it later as $l=l'\ne k$ to express $l\ne k$ if $l'=k$.
Let us now consider Eq.  \eqref{betstat2}
\begin{eqnarray}
\left\langle(\I+ A_0)\left[-B_1C_1+B_2C_2+B_1(\I+C_0)-\frac{1}{2}(\I-B_0)C_1+(\I+B_0)C_0+\frac{1}{2}B_0+\frac{3}{2}\I\right] R_4\right\rangle=\frac{3}{2}
\end{eqnarray}
and then expand it by using the fact that the observables $A_i,B_i,C_i$  for $i=0,1,2$ are certified as in Eqs. \eqref{stmea2} and \eqref{stmea3}. This gives us
\begin{eqnarray}\label{stcond2}
\left\langle\left(\bigotimes_{s=A,B,C} U_{s}^{\dagger}\right) S_{2,A'B'C'}\otimes\I_{A''B''C''}\left(\bigotimes_{s=A,B,C} U_{s}\right)\otimes R_4\right\rangle=\frac{3}{4}
\end{eqnarray}
where 
\begin{equation}\label{stab21}
    S_{2,A'B'C'}=\proj{\overline{0}}_{A'}\left[-X_{B'}X_{C'}+Y_{B'}Y_{C'}+2X_{B'}\proj{0}_{C'}-\proj{1}_{B'}X_{C'}+2\proj{0}_{B'} Z_{C'}+\frac{1}{2}Z_{B'}+\frac{3}{2}\I_{A'B'C'}\right].
\end{equation}
It is direct to verify that $S_{2,A'B'C'}$ is proportional to 
the projection onto $\ket{\psi_{1,5}}$; precisely, 
\begin{equation}
        S_{2,A'B'C'}=6\proj{\psi_{1,5}}.
\end{equation}
Using then the above form of $S_{2,A'B'C'}$
as well as the fact that $\ket{\psi_{1,5}}$ is real,
we can simplify Eq. (\ref{stcond2}) to 
%
%
\begin{eqnarray}\label{stab2}
\Tr\left[ \left(\proj{\psi_{1,5}}\otimes\sigma_{A''}\otimes \sigma_{B''}\otimes \sigma_{C''}\right) \overline{R}_4\right]=1.
\end{eqnarray}
Now, expanding $\overline{R}_4$ using Eq. \eqref{R41}, we can conclude that
\begin{eqnarray}\label{R42}
\overline{R}_4=\proj{\psi_{1,4}}\otimes\I_{A''B''C''}+\proj{\psi_{1,5}}\otimes\I_{A''B''C''}+\mathbbm{L}_{5},
\end{eqnarray}
where 
\begin{equation}\label{L0upb2}
\mathbbm{L}_{5}=\sum_{\substack{l,l'=0\\l=l'\ne4,5}}^7\ket{\varphi_l}\!\!\bra{\varphi_{l'}}\otimes \tilde{R}_{l,l'}.
\end{equation}
Similarly, we can conclude from the other two conditions in Eqs. \eqref{betstat3} and \eqref{betstat4}, and using Eq. \eqref{R42} that
\begin{eqnarray}\label{R4}
\overline{R}_4&=&\sum_{i=4}^7\proj{\psi_{1,i}}\otimes\I_{A''B''C''}+\mathbbm{L}'\nonumber\\
&=&\Gamma\otimes\mathbbm{1}_{A''B''C''}+\mathbbm{L}',
\end{eqnarray}
where 
\begin{equation}\label{L0upb3}
\mathbbm{L}'=\sum_{\substack{l,l'=0\\l=l'\ne4,5,6,7}}^7\ket{\varphi_l}\!\!\bra{\varphi_{l'}}\otimes \tilde{R}_{l,l'}.
\end{equation}
Now, following the same steps as between Eqs. \eqref{Lo1}-\eqref{upR0} using the states $\ket{\varphi_l}\ket{\xi}$ for $l=4,5,6,7$ and any $\ket{\xi}\in\mathcal{H}_{A''B''C''}$ we can conclude that $\tilde{R}_{l,l'}=0$ for any $l\ne l'$ such that $l',l=4,5,6,7$, and $\mathbbm{L}'$ simplifies to 
\begin{equation}\label{L0upb4}
\mathbbm{L}'=\sum_{\substack{l,l'=0}}^3\ket{\varphi_l}\!\!\bra{\varphi_{l'}}\otimes \tilde{R}_{l,l'}.
\end{equation}
The fact that $\overline{R}_4\geq 0$ implies that $\mathbbm{L}'\geq0$. 
Now, adding Eqs. \eqref{R0123} and \eqref{R4} and using the fact that $\sum_l\overline{R}_l=\I$, 
we get that
\begin{eqnarray}
\sum_{l=0}^3\mathbbm{L}_l+\mathbbm{L}'=0,
\end{eqnarray}
which, after taking into account the fact that $\mathbbm{L}_l$ as well as $\mathbbm{L}'$ are positive semi-definite, implies that $\mathbbm{L}_l=\mathbbm{L}'=0$.

Thus, from Eqs. \eqref{R0123} and \eqref{R4}, we obtain
\begin{eqnarray}
\overline{R}_l=\proj{\tau_l}_{E'}\otimes\I_{E''}\qquad  (l=0,1,2,3)
\end{eqnarray}
and,
\begin{eqnarray}
\overline{R}_4=\left(\I-\sum_{i=0}^3\proj{\tau_i}\right)_{E'}\otimes\I_{E''}=\Gamma_{E'}\otimes\I_{E''}.
\end{eqnarray}
By virtue of Eq. (\ref{unitary}), we finally arrive at the desired forms of Eqs. \eqref{stmeaupb1} and \eqref{stmeaupb2}, completing the proof.
\end{proof}
An interesting consequence of the above theorem is the certification of a bound entangled state when Eve observes the final outcome of her measurement $E_2$.
\setcounter{thm}{2}
\begin{cor}
Assume that the states are certified as in Eq. \eqref{statest1} and Eve's measurement $E_2$ is certified as in Eqs. \eqref{stmeaupb1} and \eqref{stmeaupb2}. Consequently, when Eve observes the last outcome of her measurement, the post-measurement state with the external parties is given by
\begin{eqnarray}
U\,\rho_{ABC}\,U^{\dagger}=\frac{1}{4}\Gamma_{A'B'C'}\otimes \tilde{\rho}_{A''B''C''},
\end{eqnarray}
where $U=\bigotimes_sU_s$ and the unitaries $U_s$ are the same as in Eq. 
\eqref{statest1}.
\end{cor}
\begin{proof}
The post-measurement state when Eve observes the last outcome of her measurement $E_2$ is given by
    \begin{eqnarray}
        \rho_{{ABC}}=\frac{1}{\overline{P}(4|e=2)}\Tr_{E}\left[\left(\I_{ABC}\otimes R_{4|2}\right)\bigotimes_{s=A,B,C}\proj{\psi_{s\overline{s}}}\right].
    \end{eqnarray}
Now, substituting the states $\ket{\psi_{s\overline{s}}}$ from Eq. \eqref{statest1} and the measurement element $R_{4|2}$ from Eq. \eqref{stmeaupb2} and then using the fact that $\overline{P}(4|e=2)=1/2$, we get that
    \begin{eqnarray}
        U\,\rho_{{ABC}}\,U^{\dagger}=2\,\Tr_{E}\left[\left(\I_{ABC}\otimes \Gamma_{E'}\otimes\I_{E''}\right)\bigotimes_{s=A,B,C}\proj{\phi^+}_{s'\overline{s}'}\otimes \proj{\xi_{s''\overline{s}''}}\right].
    \end{eqnarray}
Again using the identity $(\I\otimes Q)\ket{\phi^{+}}=(Q^T\otimes\I)\ket{\phi^{+}}$, we get
\begin{eqnarray}
U\,\rho_{{ABC}}\,U^{\dagger}=2\,\Tr_{E}\left[\left(\Gamma_{{A'B'C'}}\otimes\mathbbm{1}_{A''B''C''}\otimes\I_{E}\right)\bigotimes_{s=A,B,C}\proj{\phi^+}_{s'\overline{s}'}\otimes \proj{\xi_{s''\overline{s}''}}\right],
\end{eqnarray}
where we also used the fact that 
$\Gamma^T=\Gamma$. After tracing the $E$ subsystem we arrive at
\begin{eqnarray}
    U\,\rho_{ABC}\,U^{\dagger}=\frac{1}{4}\Gamma_{A'B'C'}\otimes \tilde{\rho}_{A''B''C''},
\end{eqnarray}
where $\tilde{\rho}_{A''B''C''}=\Tr_{E''}\left[\bigotimes_{s=A,B,C}\proj{\xi_{s''\overline{s}''}}\right]$.
\end{proof}

\section{Robustness}

Let us now comment on the robustness of our self-testing scheme. Even if there exist a few techniques to analyze robustness of self-testing schemes, most of them are restricted to the simplest situations both in the standard Bell scenario with two-outcome measurements \cite{Scarani, Wu_2014, Flavio, Yang} or the network one \cite{Marco}. Although interesting from a theoretical perspective, 
turning a self-testing statement into one that takes into account robustness to noises is an extremely difficult challenge, in particular in the network scenario. Roughly speaking, a robust self-testing statement corresponds to a situation in which the value of a Bell functional is slightly lower than the maximal one, that is, $\mathcal{B}\geq\beta_Q-\varepsilon$ where $\mathcal{B},\beta_Q$ denote a Bell functional and its maximal quantum value respectively and $\varepsilon$ is a very small positive number, and allows to asses how far the state and the measurements depart, in some norm, from the optimal quantum realization. 

We now present the robust version of the above-described self-testing scheme using a similar definition as \cite{Scarani} when Eve's input is $e=0$. For this purpose, we first utilise the fidelity bounds of the ideal post-measurement states with the external parties to the actual ones that can be obtained as shown in \cite{sarkar2025}. We do not find them here, but one can straightaway employ it here \cite{sarkar2025}. Consequently, the robustness of the ideal post-measurement states to the actual ones can be stated as
\begin{fakt}\label{fact2}
    Consider the Bell inequalities \eqref{BE1} $\mathcal{I}_{l}$ is violated $\varepsilon$-close to the maximal violation, that is, $\mathcal{I}_{{l}}\geq\beta_Q-\varepsilon$. Then, the ideal post-measurement states are close to the actual ones as
    \begin{eqnarray}
        \left\| (U_A\otimes U_B\otimes U_C)\  \rho^l_{{ABC}}\, (U_A\otimes U_B\otimes U_C)^{\dagger}-\proj{\phi_l}_{A'B'C'}\otimes\tilde{\rho}^{l}_{A''B''C''}\right\|\leq f(\varepsilon)\qquad \forall l.
    \end{eqnarray}
\end{fakt}
From Theorem 4 in Ref. \cite{sarkar2025}, $f(\varepsilon)= 72\sqrt{2\varepsilon}$.
Let us now proceed towards robust self-testing of Eve's measurement corresponding to the input $e=0$. For our purpose, we consider an additional assumption on the junk state of the sources in the noiseless case.

\setcounter{thm}{3}

\begin{thm} Consider the Bell inequalities \eqref{BE1} $\mathcal{I}_{l}$ is violated $\varepsilon$-close to the maximal violation, that is, $\mathcal{I}_{{l}}\geq\beta_Q-\varepsilon$ along with the probabilities of the central party for $e=0$ given by $ \left |\overline{P}(l|e=0)-1/8\right|\leq\varepsilon$. Then, the ideal Eve's measurement $\proj{\phi_l}_{E'}\otimes\I_{E''}$ is close to the actual one $R_{l|0}$ as
       \begin{eqnarray}\label{robures1}
        \left\|\Tr_{E''}(\Tilde{R}_l)- \proj{\phi_l}_{E'}\right\|\leq 17\left(\varepsilon+ \frac{f(\varepsilon)}{8}\right)\qquad \forall l
    \end{eqnarray}
    where $\Tilde{R}_l=(\bigotimes_{s}\,V_{\overline{s}})\ R_l\ (\bigotimes_{s} V_{\overline{s}}^{\dagger})$.
\end{thm}

\begin{proof} For the proof, we refer $R_{l|0}$ as $R_l$.
    Let us first consider the expression \eqref{54}
    \begin{eqnarray}
   \overline{P}(l)\rho^l_{{ABC}}=\ \Tr_{\overline{ABC}}\left[\left(\I_{ABC}\otimes R_l\right)\bigotimes_{s=A,B,C}\proj{\psi_{s\overline{s}}}\right].
    \end{eqnarray}
    Using the simplified form of the above expression as in \eqref{cond3} and then using \eqref{rl}, we obtain from the above expression that
    \begin{eqnarray}\label{F4}
    \overline{P}(l)(\sigma^l_{{ABC}})^*=\left(\bigotimes_{s}P_{\overline{s}}\,V_{\overline{s}}\right)\ R_l\ \left(\bigotimes_{s} V_{\overline{s}}^{\dagger}P_{\overline{s}}\right)
    \end{eqnarray}
    where $\sigma^l_{{ABC}}=(U_A\otimes U_B\otimes U_C)\  \rho^l_{{ABC}}\, (U_A\otimes U_B\otimes U_C)^{\dagger}$.
    Now, plugging in the ideal states and measurements in the noiseless scenario from Theorem \ref{theorem1} in the above formula, we obtain that
    \begin{eqnarray}\label{F5}
    \frac{1}{8}\proj{\phi_l}_{A'B'C'}\otimes(\tilde{\rho}^{l}_{A''B''C''})^*=\bigotimes_{s}P^{\mathrm{id}}_{\overline{s}}\, \left(\proj{\phi_l}_{E'}\otimes\I_{E''}\right)\  \bigotimes_{s} P^{\mathrm{id}}_{\overline{s}}
    \end{eqnarray}
    where $P_{\overline{s}}^{\mathrm{id}}=\I\otimes\sqrt{\frac{\sigma_{s''}}{2}}\quad (s=A,B,C)$. 
    Now, subtracting \eqref{F4} from \eqref{F5} and rearranging the terms on the left gives us
    \begin{eqnarray}\label{E6}
         \left(\overline{P}(l)- \frac{1}{8}\right)(\sigma^l_{{ABC}})^*+\frac{1}{8}\left((\sigma^l_{{ABC}})^*-\proj{\phi_l}_{A'B'C'}\otimes(\tilde{\rho}^{l}_{A''B''C''})^*\right)=\left(\bigotimes_{s}P_{\overline{s}}\,V_{\overline{s}}\right)\ R_l\ \left(\bigotimes_{s} V_{\overline{s}}^{\dagger}P_{\overline{s}}\right)-\nonumber\\ \bigotimes_{s}P^{\mathrm{id}}_{\overline{s}}\, \left(\proj{\phi_l}_{E'}\otimes\I_{E''}\right)\  \bigotimes_{s} P^{\mathrm{id}}_{\overline{s}}.
    \end{eqnarray}
    Now, taking the norm on both sides and using the fact that $\|\sigma^l_{{ABC}}\|\leq1$ and $\|M\|=\|M^*\|$ gives us
    \begin{eqnarray}
       \left\|\left(\bigotimes_{s}P_{\overline{s}}\,V_{\overline{s}}\right)\ R_l\ \left(\bigotimes_{s} V_{\overline{s}}^{\dagger}P_{\overline{s}}\right)-\bigotimes_{s}P^{\mathrm{id}}_{\overline{s}}\, \left(\proj{\phi_l}_{E'}\otimes\I_{E''}\right)\  \bigotimes_{s} P^{\mathrm{id}}_{\overline{s}}\right\|\leq\nonumber\\ \left|\overline{P}(l)- \frac{1}{8}\right|+ \frac{1}{8}\|\sigma^l_{{ABC}}-\proj{\phi_l}_{A'B'C'}\otimes\tilde{\rho}^{l}_{A''B''C''}\|.
    \end{eqnarray}
     Now, taking partial trace over $E''$ of the operator on the left-hand side and using the fact that $||\Tr_{E''}(M_{E'E''})||\leq||M_{E'E''}||$, we obtain
     \begin{eqnarray}
       \left\|\Tr_{E''}\left(P_{E'E''}\, \Tilde{R}_l\ P_{E'E''}\right)- \proj{\phi_l}_{E'}\right\|\leq\varepsilon+ \frac{f(\varepsilon)}{8}
     \end{eqnarray}
     where $P_{E'E''}=\bigotimes_{s}P_{\overline{s}}$, and $\Tilde{R}_l=(\bigotimes_{s}\,V_{\overline{s}})\ R_l\ (\bigotimes_{s} V_{\overline{s}}^{\dagger})$. For simplicity, we will now represent $P_{E'E''}\equiv P$. Now using triangle inequality the above formula can be expressed as 
     \begin{eqnarray}
       \left\|\Tr_{E''}(\Tilde{R}_l)- \proj{\phi_l}_{E'}\right\|\leq\varepsilon+ \frac{f(\varepsilon)}{8}+ \left\|\Tr_{E''}[P\Tilde{R}_l(P-\I)]\right\|+\left\|\Tr_{E''}[(P-\I)\Tilde{R}_l]\right\|.
     \end{eqnarray}
     Using the fact that $M\leq|M|$ and thus $\Tr_{E''}M\leq\Tr_{E''}|M|$, we obtain that
     \begin{eqnarray}\label{E10}
          \left\|\Tr_{E''}(\Tilde{R}_l)- \proj{\phi_l}_{E'}\right\|\leq\varepsilon+ \frac{f(\varepsilon)}{8}+ \left\|\Tr_{E''}|P\Tilde{R}_l(P-\I)|\right\|+\left\|\Tr_{E''}|(P-\I)\Tilde{R}_l|\right\|.
     \end{eqnarray}
     Let us now observe from the above formula that
     \begin{eqnarray}
        |P\Tilde{R}_l(P-\I)|=\sqrt{(P-\I)\Tilde{R}_lP^2\Tilde{R}_l(P-\I)}\leq \sqrt{(P-\I)\Tilde{R}_l^2(P-\I)}\leq\sqrt{(P-\I)^2}=\I-P
     \end{eqnarray}
     where we use the fact that $(P-\I)\Tilde{R}_l(\I-P^2)\Tilde{R}_l(P-\I)\geq0$ and $(P-\I)(\I-\Tilde{R}_l^2)(P-\I)\geq0$ as $0\leq P\leq\I$ and $0\leq \tilde{R}_l\leq\I$ and then the theorem 2.2.6 from \cite{murphy2014c} which states that $0\leq A\leq B\implies\sqrt{A}\leq \sqrt{B}$ for any two positive matrices $A,B$. Thus, from \eqref{E10} we have that
     \begin{eqnarray}\label{E12}
          \left\|\Tr_{E''}(\Tilde{R}_l)- \proj{\phi_l}_{E'}\right\|\leq\varepsilon+ \frac{f(\varepsilon)}{8}+ 2\left\|\I-\Tr_{E''}P\right\|
     \end{eqnarray}

    Let us again consider the relation \eqref{E6} and sum it over all $l$ and taking the norm on both sides to obtain
    \begin{eqnarray}
         \sum_{l}\left|\overline{P}(l)- \frac{1}{8}\right|\|\sigma^l_{{ABC}}\|+\frac{1}{8}\sum_{l}\left\|\sigma^l_{{ABC}}-\proj{\phi_l}_{A'B'C'}\otimes(\tilde{\rho}^{l}_{A''B''C''})\right\|\geq\left\|\bigotimes_{s}P_{\overline{s}}^2 -\bigotimes_{s}(P^{\mathrm{id}}_{\overline{s}})^2\right\|
    \end{eqnarray}
    which on employing Fact \ref{fact2} and $\|\sigma^l_{{ABC}}\|\leq1$ gives us
    \begin{eqnarray}\label{E12}
         \left\|\bigotimes_{s}P_{\overline{s}}^2 -\bigotimes_{s}(P^{\mathrm{id}}_{\overline{s}})^2\right\|\leq8\varepsilon+f(\varepsilon).
    \end{eqnarray}
    Taking a partial trace over $E''$ on the left side of the above formula gives us
    \begin{eqnarray}
         \left\|\Tr_{E''}(P^2) -\I\right\|\leq8\varepsilon+f(\varepsilon).
    \end{eqnarray}
    Let us notice that $P^2\leq P\leq\I$, using which we get that $0\leq\I-\Tr_{E''}(P)\leq\I-\Tr_{E''}(P^2)$. Now, using theorem 2.2.5 of \cite{murphy2014c} which states that $0\leq A\leq B\implies ||A||\leq||B||$ and thus we obtain from the above formula that
    \begin{eqnarray}
        \left\|\Tr_{E''}(P) -\I\right\|\leq8\varepsilon+f(\varepsilon).
    \end{eqnarray}
Consequently, we finally obtain from \eqref{E12} that
     \begin{eqnarray}
           \left\|\Tr_{E''}(\Tilde{R}_l)- \proj{\phi_l}_{E'}\right\|\leq 17\left(\varepsilon+ \frac{f(\varepsilon)}{8}\right).
     \end{eqnarray}
     This completes the proof.
\end{proof}

Let us now proceed towards robust self-testing Eve's measurement corresponding to the input $e=1$. Unlike the previous case, where we find the error in the measurment from the error in statistics, here we take the opposite direction, that is, we find the error in statistics from the error in measurement.

\begin{thm}\label{theo5}  Consider again the network scenario outlined in the main text (see Fig. 1) 
with Eve's measurement $E_1=\{R_{l|1}\}$ being close to the ideal one as 
       \begin{eqnarray}\label{robures2}
        \left\|\ \bigotimes_{s}\,V_{\overline{s}}\ R_{l|1}\bigotimes_{s=A,B,C}\ket{\psi_{s\overline{s}}}-\, \left(\proj{\delta_l}_{E'}\otimes\I_{E''}\right)\ \bigotimes_{s=A,B,C}\ket{\psi^{\mathrm{id}}_{s\overline{s}}}\right\|\leq\varepsilon\qquad \forall l
    \end{eqnarray}
    where $\ket{\psi_{s\overline{s}}}$ is the actual state in the experiment and $\ket{\psi^{\mathrm{id}}_{s\overline{s}}}$ is the state certified in \eqref{statest1}.
    Then the probabilities of the central party for $e=1$ given by $ \left |p(a,b,c,l|x,y,z,1)-p_{\mathrm{id}}(a,b,c,l|x,y,z,1)\right|\leq2\varepsilon$ where $p_{\mathrm{id}}(a,b,c,l|x,y,z,1)$ are the probabilities obtained in the noiseless scenario as given in \eqref{NLWEstat}. 
\end{thm}

\begin{proof}
    Let us first consider $l=0$ and notice from \eqref{NLWEstat} that only $p_{\mathrm{id}}(0,1,0,0|0,0,1,1)=\frac{1}{8}$ and $p_\mathrm{id}(a,b,c,0|0,0,1,1)=0$ for rest of the $a,b,c$. Now, taking the formula \eqref{stcond1NLWE} to express the probabilities $p_\mathrm{id}(a,b,c,0|0,0,1,1)$ for any $a,b,c$ gives us
    \begin{eqnarray}\label{F10}
  \left\langle \left(N_{a|0}^A\otimes N_{b|0}^B\otimes N_{c|1}^C\right)\otimes R_0\right\rangle_{\psi_{ABCE}}=p(a,b,c,0|0,0,1,1).
    \end{eqnarray}
    Using the fact that $N_{i|j}^s\leq\I$, we obtain that
    \begin{eqnarray}\label{F11}
  \left\langle R_0\right\rangle_{\psi_{ABCE}}\geq p(0,1,0,0|0,0,1,1).
    \end{eqnarray}
    Similarly, in the noiseless case from Theorem \ref{theorem2}  and summing over $a,b,c$ gives us
    \begin{eqnarray}
         \left\langle \proj{\delta_0}_{E'}\otimes\I_{E''}\right\rangle_{\psi^{\mathrm{id}}_{ABCE}}=\sum_{a,b,c=0,1}p_{\mathrm{id}}(a,b,c,0|0,0,1,1)=\frac{1}{8}.
    \end{eqnarray}
   Now, expanding the condition \eqref{robures2} for $l=0$ gives us (for simplicity we drop the lower indices from the state)
   \begin{eqnarray}\label{F13}
       \left\langle R_0\right\rangle_{\psi}+ \left\langle \proj{\delta_0}_{E'}\otimes\I_{E''}\right\rangle_{\psi^{\mathrm{id}}}-2\mathrm{Re}\bra{\psi}R_0\bigotimes_{s}\,V_{\overline{s}}^{\dagger}\proj{\delta_0}_{E'}\otimes\I_{E''}\ket{\psi^{\mathrm{id}}} \leq \varepsilon^2
   \end{eqnarray}
   where we also assumed that $R_0$ is a projector. This is well justified as Eve's dimension is unrestricted. 
   Now using the Cauchy-Schwarz inequality, the term with the negative sign in the above formula can be upper-bounded as
   \begin{eqnarray}\label{F14}
   |\bra{\psi}R_0\bigotimes_{s}\,V_{\overline{s}}^{\dagger}\proj{\delta_0}_{E'}\otimes\I_{E''}\ket{\psi^{\mathrm{id}}}|\leq \sqrt{\left\langle R_0\right\rangle_{\psi}}\sqrt{\left\langle \proj{\delta_0}_{E'}\otimes\I_{E''}\right\rangle_{\psi^{\mathrm{id}}}}.
   \end{eqnarray}
   Using \eqref{F14}, we can thus lower bound the left-hand side of \eqref{F13} as
   \begin{eqnarray}
       \left\langle R_0\right\rangle_{\psi}+ \left\langle \proj{\delta_0}_{E'}\otimes\I_{E''}\right\rangle_{\psi^{\mathrm{id}}}-2\sqrt{\left\langle R_0\right\rangle_{\psi}}\sqrt{\left\langle \proj{\delta_0}_{E'}\otimes\I_{E''}\right\rangle_{\psi^{\mathrm{id}}}} \leq \varepsilon^2
   \end{eqnarray}
   which can be further simplified to
   \begin{eqnarray}
      \left( \sqrt{\left\langle R_0\right\rangle_{\psi}}-\sqrt{\left\langle \proj{\delta_0}_{E'}\otimes\I_{E''}\right\rangle_{\psi^{\mathrm{id}}}}\right)^2\leq\varepsilon^2.
   \end{eqnarray}
   Using Eqs. \eqref{F13} and \eqref{F14}, we obtain
   \begin{eqnarray}
       \left| \sqrt{p(0,1,0,0|0,0,1,1)}-\frac{1}{\sqrt{8}}\right|\leq\varepsilon
   \end{eqnarray}
   which can also be expressed as
    \begin{eqnarray}
       \left| p(0,1,0,0|0,0,1,1)-\frac{1}{8}\right|\leq\varepsilon \left| \sqrt{p(0,1,0,0|0,0,1,1)}+\frac{1}{\sqrt{8}}\right|\leq2\varepsilon.
   \end{eqnarray}
   Similarly, one can obtain the same value for any $l$. This completes the proof.
\end{proof}

Similarly, following the same lines as the above Theorem \ref{theo5}, one can obtain the robust version of Theorem \ref{theorem4}. 
\linebreak

Furthermore, from an experimental perspective, one always has an idea about different noises that might affect the experiment. To suffice this, we provide a general idea to compute the robustness of the scheme given any general noise model. Any noise affecting the state or measurements can be represented by a quantum change which can in turn be represented using Kraus operators $\{K_i^s\}$ with $\sum_i(K_i^s)^{\dagger}K_i^s=\I$. Note that the Kraus operators might map a $d-$dimensional system to $D-$dimensional system such that $D\geq d$. In general the Krauss operators $K_i^s$ act on both the subsystems. Thus, the states $\ket{\psi_{s\overline{s}}}$ can be represented as 
\begin{eqnarray}
    \proj{\psi_{s\overline{s}}}=\sum_i K^s_i\proj{\phi^+_{s\overline{s}}}(K^s_i)^{\dagger}\qquad \forall s.
\end{eqnarray}
Similarly, the measurements of Alice, Bob, Charlie, and Eve denoted by $N_{a|x}, N_{b|y},  N_{c|z},  N_{e|l}$ can be represented as 
\begin{eqnarray}
    N_{u|v}=\sum_i K^{u,v}_i\tilde{N}_{u|v}(K_i^{u,v})^{\dagger}\qquad \forall u,v,
\end{eqnarray}
where $\tilde{N}_{u|v}$ denotes the ideal two-dimensional projector that results in saturating the ideal statistics, that is, for instance when considering Charlie $\tilde{N}_{0|0}=\proj{0}, \tilde{N}_{1|0}=\proj{1}, \tilde{N}_{0|1}=\proj{+}, \tilde{N}_{1|1}=\proj{-}, \tilde{N}_{0|2}=\proj{+_y}, \tilde{N}_{1|2}=\proj{-_y}$.  Using these noisy states and measurements one can then compute the observed statistics. As an example, we show some plots for the following scenario.


Let us now consider an exemplary noise model represented by the depolarizing channel 
whose Kraus operators are
\begin{equation}\label{noise model}
    K_0=\sqrt{1-\frac{3p}{4}}\I, \,\, K_1=\sqrt{\frac{p}{4}} X,\,\, K_2=\sqrt{\frac{p}{4}} Y, \,\, \mathrm{and}\,\, K_3=\sqrt{\frac{p}{4}} Z
\end{equation}
with $\sum_iK_i^{\dagger}K_i=\I$. For simplicity, we consider a situation in which this noise is
applied only to Eve's subsystems ($\overline{A}$, $\overline{B}$, and $\overline{C}$) and moreover, we assume that each channel is characterized by the same $p$. However, as we also mentioned earlier, one can also consider a more general noise model for the central party. 

Using the above noise model in Eq. (\ref{noise model}), we investigate how the self-testing conditions depart from the ideal scenario ($p=0$). This is depicted in the following figures for various measurements we considered on the central party. In Fig. \ref{fig1}, we plot the variation of the Bell value $\mathcal{I}_{l_1l_2l_3}$ with respect to the noise model described above.  In Fig. \ref{fig2}, we plot the variation of the sum of the probabilities in Eq. \eqref{NLWEstat} given by
\begin{eqnarray}
    W_1=p(0,1,0,0|0,0,1,1)+p(0,1,1,1|0,0,1,1)+
p(0,0,1,2|1,0,0,1)+p(1,0,1,3|1,0,0,1)+
p(1,0,0,4|0,1,0,1)\nonumber\\+ p(1,1,0,5|0,1,0,1)+
p(0,0,0,6|0,0,0,1)+ p(1,1,1,7|0,0,0,1)\qquad
\end{eqnarray}
with respect to the noise model described above. In the left figure of Fig. \ref{fig3}, we plot the variation of the sum of the probabilities in Eq. \eqref{betstat0} given by
\begin{eqnarray}
    W_2=p(0100|0012)+ p(0011|1002)+
p(1002|0102)+p(1113|1112)
\end{eqnarray}
with respect to the noise model described above.  Similarly, in the right figure of Fig. \ref{fig3}, we plot the variation of the sum of all the correlators appearing in Eqs. \eqref{betstatfull} 
with respect to the noise model described above. In all the figures the topmost dotted line represents the value of the concerned quantities in the noiseless scenario.

%
\begin{figure}[H]
    \centering
    \includegraphics[scale=0.6]{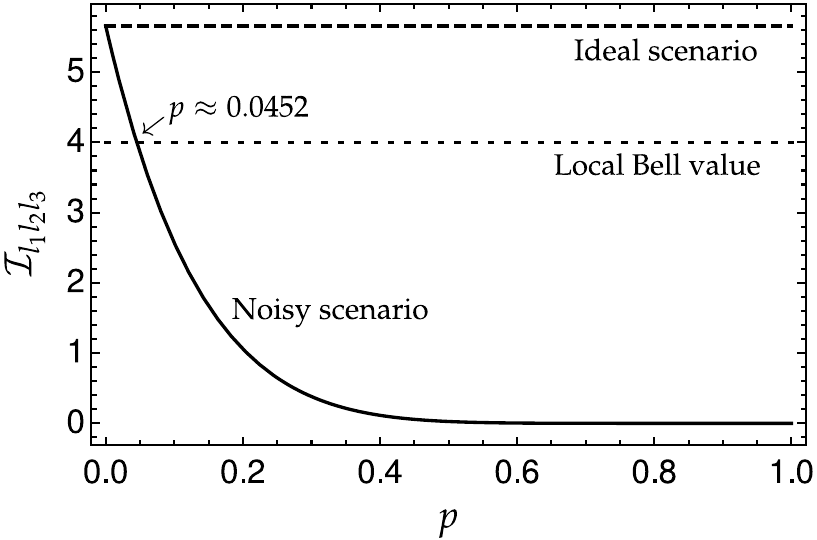}
    \caption{The variation of the Bell value $\mathcal{\hat{I}}_{l_1l_2l_3}$ with the noise parameter $p$ when the central party performs $M_{GHZ}$ for any $l_1,l_2,l_3=0,1$. The dashed, solid, and dotted lines represent the ideal scenario, noisy scenario, and local Bell value respectively.}
    
    \label{fig1}
\end{figure}

\begin{figure}[H]
    \centering
    \includegraphics[scale=0.6]{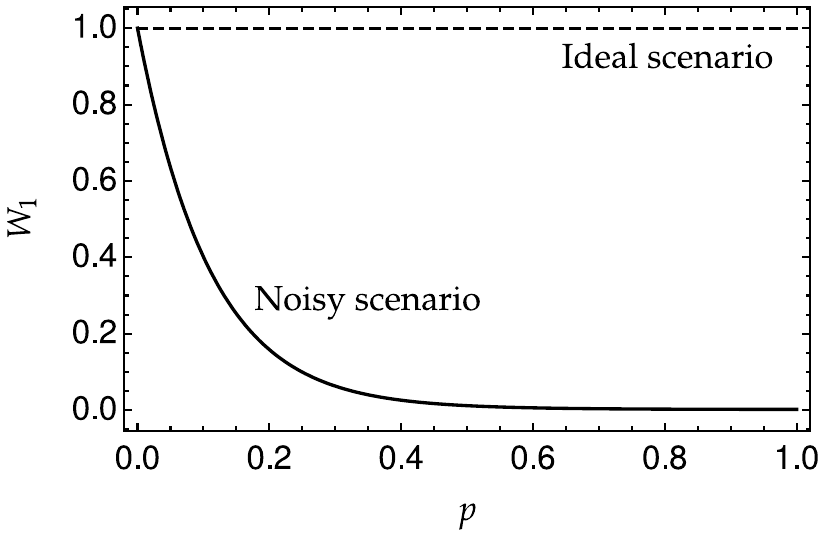}
    \caption{The variation of the sum of the probabilities in Eq. (\ref{NLWEstat}) with the noise parameter $p$ when the central party performs $M_{NLWE}$. The dashed and solid lines represent the ideal scenario and noisy scenario respectively.}
    
    \label{fig2}
\end{figure}

\begin{figure}[H]
    \centering
    \includegraphics[scale=0.6]{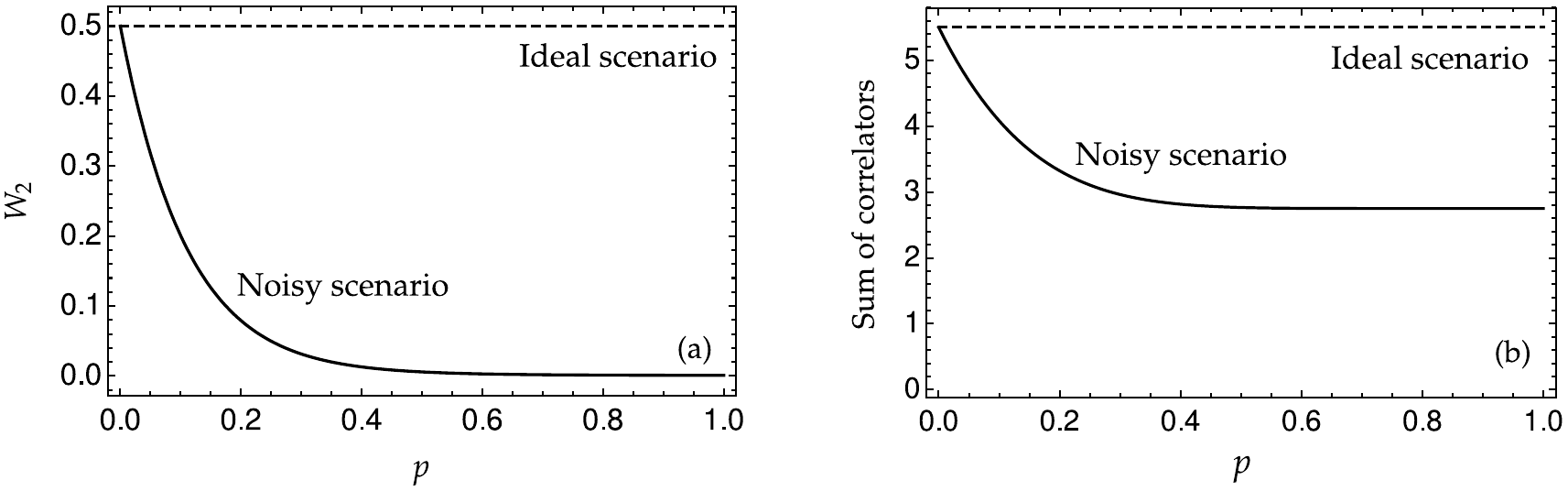}
    \caption{(a) The variation of the sum of the probabilities in Eq. (\ref{betstat0}) with the noise parameter $p$ when the central party performs $M_{UPB}$. (b) The variation of the sum of the correlators in Eq. (\ref{betstatfull}) with the noise parameter $p$ when the central party performs $M_{UPB}$. The dashed and solid lines represent the ideal scenario and noisy scenario respectively.}
    
    \label{fig3}
\end{figure}


\begin{thebibliography}{42}%
\makeatletter
\providecommand \@ifxundefined [1]{%
 \@ifx{#1\undefined}
}%
\providecommand \@ifnum [1]{%
 \ifnum #1\expandafter \@firstoftwo
 \else \expandafter \@secondoftwo
 \fi
}%
\providecommand \@ifx [1]{%
 \ifx #1\expandafter \@firstoftwo
 \else \expandafter \@secondoftwo
 \fi
}%
\providecommand \natexlab [1]{#1}%
\providecommand \enquote  [1]{``#1''}%
\providecommand \bibnamefont  [1]{#1}%
\providecommand \bibfnamefont [1]{#1}%
\providecommand \citenamefont [1]{#1}%
\providecommand \href@noop [0]{\@secondoftwo}%
\providecommand \href [0]{\begingroup \@sanitize@url \@href}%
\providecommand \@href[1]{\@@startlink{#1}\@@href}%
\providecommand \@@href[1]{\endgroup#1\@@endlink}%
\providecommand \@sanitize@url [0]{\catcode `\\12\catcode `\$12\catcode
  `\&12\catcode `\#12\catcode `\^12\catcode `\_12\catcode `\%12\relax}%
\providecommand \@@startlink[1]{}%
\providecommand \@@endlink[0]{}%
\providecommand \url  [0]{\begingroup\@sanitize@url \@url }%
\providecommand \@url [1]{\endgroup\@href {#1}{\urlprefix }}%
\providecommand \urlprefix  [0]{URL }%
\providecommand \Eprint [0]{\href }%
\providecommand \doibase [0]{https://doi.org/}%
\providecommand \selectlanguage [0]{\@gobble}%
\providecommand \bibinfo  [0]{\@secondoftwo}%
\providecommand \bibfield  [0]{\@secondoftwo}%
\providecommand \translation [1]{[#1]}%
\providecommand \BibitemOpen [0]{}%
\providecommand \bibitemStop [0]{}%
\providecommand \bibitemNoStop [0]{.\EOS\space}%
\providecommand \EOS [0]{\spacefactor3000\relax}%
\providecommand \BibitemShut  [1]{\csname bibitem#1\endcsname}%
\let\auto@bib@innerbib\@empty
\bibitem [{\citenamefont {Gisin}\ \emph {et~al.}(2002)\citenamefont {Gisin},
  \citenamefont {Ribordy}, \citenamefont {Tittel},\ and\ \citenamefont
  {Zbinden}}]{Gisin_crypto}%
  \BibitemOpen
  \bibfield  {author} {\bibinfo {author} {\bibfnamefont {N.}~\bibnamefont
  {Gisin}}, \bibinfo {author} {\bibfnamefont {G.}~\bibnamefont {Ribordy}},
  \bibinfo {author} {\bibfnamefont {W.}~\bibnamefont {Tittel}},\ and\ \bibinfo
  {author} {\bibfnamefont {H.}~\bibnamefont {Zbinden}},\ }\bibfield  {title}
  {\bibinfo {title} {Quantum cryptography},\ }\href
  {https://doi.org/10.1103/RevModPhys.74.145} {\bibfield  {journal} {\bibinfo
  {journal} {Rev. Mod. Phys.}\ }\textbf {\bibinfo {volume} {74}},\ \bibinfo
  {pages} {145} (\bibinfo {year} {2002})}\BibitemShut {NoStop}%
\bibitem [{\citenamefont {Ac\'{\i}n}\ \emph {et~al.}(2007)\citenamefont
  {Ac\'{\i}n}, \citenamefont {Brunner}, \citenamefont {Gisin}, \citenamefont
  {Massar}, \citenamefont {Pironio},\ and\ \citenamefont {Scarani}}]{DICrypto}%
  \BibitemOpen
  \bibfield  {author} {\bibinfo {author} {\bibfnamefont {A.}~\bibnamefont
  {Ac\'{\i}n}}, \bibinfo {author} {\bibfnamefont {N.}~\bibnamefont {Brunner}},
  \bibinfo {author} {\bibfnamefont {N.}~\bibnamefont {Gisin}}, \bibinfo
  {author} {\bibfnamefont {S.}~\bibnamefont {Massar}}, \bibinfo {author}
  {\bibfnamefont {S.}~\bibnamefont {Pironio}},\ and\ \bibinfo {author}
  {\bibfnamefont {V.}~\bibnamefont {Scarani}},\ }\bibfield  {title} {\bibinfo
  {title} {Device-independent security of quantum cryptography against
  collective attacks},\ }\href {https://doi.org/10.1103/PhysRevLett.98.230501}
  {\bibfield  {journal} {\bibinfo  {journal} {Phys. Rev. Lett.}\ }\textbf
  {\bibinfo {volume} {98}},\ \bibinfo {pages} {230501} (\bibinfo {year}
  {2007})}\BibitemShut {NoStop}%
\bibitem [{\citenamefont {Bancal}\ \emph {et~al.}(2011)\citenamefont {Bancal},
  \citenamefont {Gisin}, \citenamefont {Liang},\ and\ \citenamefont
  {Pironio}}]{Bancal_PhysRevLett.106.250404}%
  \BibitemOpen
  \bibfield  {author} {\bibinfo {author} {\bibfnamefont {J.-D.}\ \bibnamefont
  {Bancal}}, \bibinfo {author} {\bibfnamefont {N.}~\bibnamefont {Gisin}},
  \bibinfo {author} {\bibfnamefont {Y.-C.}\ \bibnamefont {Liang}},\ and\
  \bibinfo {author} {\bibfnamefont {S.}~\bibnamefont {Pironio}},\ }\bibfield
  {title} {\bibinfo {title} {Device-independent witnesses of genuine
  multipartite entanglement},\ }\href
  {https://doi.org/10.1103/PhysRevLett.106.250404} {\bibfield  {journal}
  {\bibinfo  {journal} {Phys. Rev. Lett.}\ }\textbf {\bibinfo {volume} {106}},\
  \bibinfo {pages} {250404} (\bibinfo {year} {2011})}\BibitemShut {NoStop}%
\bibitem [{\citenamefont {Brunner}\ \emph {et~al.}(2008)\citenamefont
  {Brunner}, \citenamefont {Pironio}, \citenamefont {Acin}, \citenamefont
  {Gisin}, \citenamefont {M\'ethot},\ and\ \citenamefont
  {Scarani}}]{Brunner_PhysRevLett.100.210503}%
  \BibitemOpen
  \bibfield  {author} {\bibinfo {author} {\bibfnamefont {N.}~\bibnamefont
  {Brunner}}, \bibinfo {author} {\bibfnamefont {S.}~\bibnamefont {Pironio}},
  \bibinfo {author} {\bibfnamefont {A.}~\bibnamefont {Acin}}, \bibinfo {author}
  {\bibfnamefont {N.}~\bibnamefont {Gisin}}, \bibinfo {author} {\bibfnamefont
  {A.~A.}\ \bibnamefont {M\'ethot}},\ and\ \bibinfo {author} {\bibfnamefont
  {V.}~\bibnamefont {Scarani}},\ }\bibfield  {title} {\bibinfo {title} {Testing
  the dimension of hilbert spaces},\ }\href
  {https://doi.org/10.1103/PhysRevLett.100.210503} {\bibfield  {journal}
  {\bibinfo  {journal} {Phys. Rev. Lett.}\ }\textbf {\bibinfo {volume} {100}},\
  \bibinfo {pages} {210503} (\bibinfo {year} {2008})}\BibitemShut {NoStop}%
\bibitem [{\citenamefont {Bell}(1964)}]{Bell}%
  \BibitemOpen
  \bibfield  {author} {\bibinfo {author} {\bibfnamefont {J.~S.}\ \bibnamefont
  {Bell}},\ }\bibfield  {title} {\bibinfo {title} {On the
  {E}instein-{P}odolsky-{R}osen paradox},\ }\href
  {https://doi.org/10.1103/PhysicsPhysiqueFizika.1.195} {\bibfield  {journal}
  {\bibinfo  {journal} {Physics Physique Fizika}\ }\textbf {\bibinfo {volume}
  {1}},\ \bibinfo {pages} {195} (\bibinfo {year} {1964})}\BibitemShut {NoStop}%
\bibitem [{\citenamefont {Bell}(1966)}]{Bell66}%
  \BibitemOpen
  \bibfield  {author} {\bibinfo {author} {\bibfnamefont {J.~S.}\ \bibnamefont
  {Bell}},\ }\bibfield  {title} {\bibinfo {title} {On the problem of hidden
  variables in quantum mechanics},\ }\href
  {https://doi.org/10.1103/RevModPhys.38.447} {\bibfield  {journal} {\bibinfo
  {journal} {Rev. Mod. Phys.}\ }\textbf {\bibinfo {volume} {38}},\ \bibinfo
  {pages} {447} (\bibinfo {year} {1966})}\BibitemShut {NoStop}%
\bibitem [{\citenamefont {Mayers}\ and\ \citenamefont
  {Yao}(1998)}]{Mayers_selftesting}%
  \BibitemOpen
  \bibfield  {author} {\bibinfo {author} {\bibfnamefont {D.}~\bibnamefont
  {Mayers}}\ and\ \bibinfo {author} {\bibfnamefont {A.}~\bibnamefont {Yao}},\
  }\bibfield  {title} {\bibinfo {title} {Quantum cryptography with imperfect
  apparatus},\ }in\ \href {https://doi.org/10.1109/SFCS.1998.743501} {\emph
  {\bibinfo {booktitle} {Proceedings 39th Annual Symposium on Foundations of
  Computer Science (Cat. No.98CB36280)}}}\ (\bibinfo {year} {1998})\ pp.\
  \bibinfo {pages} {503--509}\BibitemShut {NoStop}%
\bibitem [{\citenamefont {{\v{S}}upi\'c}\ and\ \citenamefont
  {Bowles}(2020)}]{SupicReview}%
  \BibitemOpen
  \bibfield  {author} {\bibinfo {author} {\bibfnamefont {I.}~\bibnamefont
  {{\v{S}}upi\'c}}\ and\ \bibinfo {author} {\bibfnamefont {J.}~\bibnamefont
  {Bowles}},\ }\bibfield  {title} {\bibinfo {title} {Self-testing of quantum
  systems: a review},\ }\href {https://doi.org/10.22331/q-2020-09-30-337}
  {\bibfield  {journal} {\bibinfo  {journal} {Quantum}\ }\textbf {\bibinfo
  {volume} {4}},\ \bibinfo {pages} {337} (\bibinfo {year} {2020})}\BibitemShut
  {NoStop}%
\bibitem [{\citenamefont {Mayers}\ and\ \citenamefont {Yao}(2004)}]{Yao}%
  \BibitemOpen
  \bibfield  {author} {\bibinfo {author} {\bibfnamefont {D.}~\bibnamefont
  {Mayers}}\ and\ \bibinfo {author} {\bibfnamefont {A.}~\bibnamefont {Yao}},\
  }\bibfield  {title} {\bibinfo {title} {Self testing quantum apparatus},\
  }\href {https://doi.org/doi.org/10.26421/QIC4.4} {\bibfield  {journal}
  {\bibinfo  {journal} {Quantum Inf. Comput.}\ }\textbf {\bibinfo {volume}
  {4}},\ \bibinfo {pages} {273} (\bibinfo {year} {2004})}\BibitemShut {NoStop}%
\bibitem [{\citenamefont {McKague}\ \emph {et~al.}(2012)\citenamefont
  {McKague}, \citenamefont {Yang},\ and\ \citenamefont {Scarani}}]{Scarani}%
  \BibitemOpen
  \bibfield  {author} {\bibinfo {author} {\bibfnamefont {M.}~\bibnamefont
  {McKague}}, \bibinfo {author} {\bibfnamefont {T.~H.}\ \bibnamefont {Yang}},\
  and\ \bibinfo {author} {\bibfnamefont {V.}~\bibnamefont {Scarani}},\
  }\bibfield  {title} {\bibinfo {title} {Robust self-testing of the singlet},\
  }\href {https://doi.org/10.1088/1751-8113/45/45/455304} {\bibfield  {journal}
  {\bibinfo  {journal} {J. Phys. A: Math. Theor.}\ }\textbf {\bibinfo {volume}
  {45}},\ \bibinfo {pages} {455304} (\bibinfo {year} {2012})}\BibitemShut
  {NoStop}%
\bibitem [{\citenamefont {Reichardt}\ \emph {et~al.}(2013)\citenamefont
  {Reichardt}, \citenamefont {Unger},\ and\ \citenamefont
  {Vazirani}}]{Reichardt_nature}%
  \BibitemOpen
  \bibfield  {author} {\bibinfo {author} {\bibfnamefont {B.}~\bibnamefont
  {Reichardt}}, \bibinfo {author} {\bibfnamefont {F.}~\bibnamefont {Unger}},\
  and\ \bibinfo {author} {\bibfnamefont {U.}~\bibnamefont {Vazirani}},\
  }\bibfield  {title} {\bibinfo {title} {Classical command of quantum
  systems},\ }\href {https://doi.org/10.1038/nature12035} {\bibfield  {journal}
  {\bibinfo  {journal} {Nature}\ }\textbf {\bibinfo {volume} {496}},\ \bibinfo
  {pages} {456–460} (\bibinfo {year} {2013})}\BibitemShut {NoStop}%
\bibitem [{\citenamefont {McKague}(2014)}]{Mckague_2014}%
  \BibitemOpen
  \bibfield  {author} {\bibinfo {author} {\bibfnamefont {M.}~\bibnamefont
  {McKague}},\ }\bibfield  {title} {\bibinfo {title} {Self-testing graph
  states},\ }in\ \href {https://doi.org/10.1007/978-3-642-54429-3_7} {\emph
  {\bibinfo {booktitle} {Theory of Quantum Computation, Communication, and
  Cryptography}}},\ \bibinfo {editor} {edited by\ \bibinfo {editor}
  {\bibfnamefont {D.}~\bibnamefont {Bacon}}, \bibinfo {editor} {\bibfnamefont
  {M.}~\bibnamefont {Martin-Delgado}},\ and\ \bibinfo {editor} {\bibfnamefont
  {M.}~\bibnamefont {Roetteler}}}\ (\bibinfo  {publisher} {Springer Berlin
  Heidelberg},\ \bibinfo {address} {Berlin, Heidelberg},\ \bibinfo {year}
  {2014})\ pp.\ \bibinfo {pages} {104--120}\BibitemShut {NoStop}%
\bibitem [{\citenamefont {Wu}\ \emph {et~al.}(2014)\citenamefont {Wu},
  \citenamefont {Cai}, \citenamefont {Yang}, \citenamefont {Le}, \citenamefont
  {Bancal},\ and\ \citenamefont {Scarani}}]{Wu_2014}%
  \BibitemOpen
  \bibfield  {author} {\bibinfo {author} {\bibfnamefont {X.}~\bibnamefont
  {Wu}}, \bibinfo {author} {\bibfnamefont {Y.}~\bibnamefont {Cai}}, \bibinfo
  {author} {\bibfnamefont {T.~H.}\ \bibnamefont {Yang}}, \bibinfo {author}
  {\bibfnamefont {H.~N.}\ \bibnamefont {Le}}, \bibinfo {author} {\bibfnamefont
  {J.-D.}\ \bibnamefont {Bancal}},\ and\ \bibinfo {author} {\bibfnamefont
  {V.}~\bibnamefont {Scarani}},\ }\bibfield  {title} {\bibinfo {title} {Robust
  self-testing of the three-qubit {$W$} state},\ }\href
  {https://doi.org/10.1103/PhysRevA.90.042339} {\bibfield  {journal} {\bibinfo
  {journal} {Phys. Rev. A}\ }\textbf {\bibinfo {volume} {90}},\ \bibinfo
  {pages} {042339} (\bibinfo {year} {2014})}\BibitemShut {NoStop}%
\bibitem [{\citenamefont {Bamps}\ and\ \citenamefont {Pironio}(2015)}]{Bamps}%
  \BibitemOpen
  \bibfield  {author} {\bibinfo {author} {\bibfnamefont {C.}~\bibnamefont
  {Bamps}}\ and\ \bibinfo {author} {\bibfnamefont {S.}~\bibnamefont
  {Pironio}},\ }\bibfield  {title} {\bibinfo {title} {Sum-of-squares
  decompositions for a family of {C}lauser-{H}orne-{S}himony-{H}olt-like
  inequalities and their application to self-testing},\ }\href
  {https://doi.org/10.1103/PhysRevA.91.052111} {\bibfield  {journal} {\bibinfo
  {journal} {Phys. Rev. A}\ }\textbf {\bibinfo {volume} {91}},\ \bibinfo
  {pages} {052111} (\bibinfo {year} {2015})}\BibitemShut {NoStop}%
\bibitem [{\citenamefont {Wang}\ \emph {et~al.}(2016)\citenamefont {Wang},
  \citenamefont {Wu},\ and\ \citenamefont {Scarani}}]{All}%
  \BibitemOpen
  \bibfield  {author} {\bibinfo {author} {\bibfnamefont {Y.}~\bibnamefont
  {Wang}}, \bibinfo {author} {\bibfnamefont {X.}~\bibnamefont {Wu}},\ and\
  \bibinfo {author} {\bibfnamefont {V.}~\bibnamefont {Scarani}},\ }\bibfield
  {title} {\bibinfo {title} {All the self-testings of the singlet for two
  binary measurements},\ }\href {https://doi.org/10.1088/1367-2630/18/2/025021}
  {\bibfield  {journal} {\bibinfo  {journal} {New J. Phys.}\ }\textbf {\bibinfo
  {volume} {18}},\ \bibinfo {pages} {025021} (\bibinfo {year}
  {2016})}\BibitemShut {NoStop}%
\bibitem [{\citenamefont {{\v{S}}upi{\'{c}}}\ \emph {et~al.}(2016)\citenamefont
  {{\v{S}}upi{\'{c}}}, \citenamefont {Augusiak}, \citenamefont {Salavrakos},\
  and\ \citenamefont {Ac{\'{\i}}n}}]{chainedBell}%
  \BibitemOpen
  \bibfield  {author} {\bibinfo {author} {\bibfnamefont {I.}~\bibnamefont
  {{\v{S}}upi{\'{c}}}}, \bibinfo {author} {\bibfnamefont {R.}~\bibnamefont
  {Augusiak}}, \bibinfo {author} {\bibfnamefont {A.}~\bibnamefont
  {Salavrakos}},\ and\ \bibinfo {author} {\bibfnamefont {A.}~\bibnamefont
  {Ac{\'{\i}}n}},\ }\bibfield  {title} {\bibinfo {title} {Self-testing
  protocols based on the chained {B}ell inequalities},\ }\href
  {https://doi.org/10.1088/1367-2630/18/3/035013} {\bibfield  {journal}
  {\bibinfo  {journal} {New J. Phys.}\ }\textbf {\bibinfo {volume} {18}},\
  \bibinfo {pages} {035013} (\bibinfo {year} {2016})}\BibitemShut {NoStop}%
\bibitem [{\citenamefont {Coladangelo}\ \emph {et~al.}(2017)\citenamefont
  {Coladangelo}, \citenamefont {Goh},\ and\ \citenamefont
  {Scarani}}]{Projection}%
  \BibitemOpen
  \bibfield  {author} {\bibinfo {author} {\bibfnamefont {A.}~\bibnamefont
  {Coladangelo}}, \bibinfo {author} {\bibfnamefont {K.~T.}\ \bibnamefont
  {Goh}},\ and\ \bibinfo {author} {\bibfnamefont {V.}~\bibnamefont {Scarani}},\
  }\bibfield  {title} {\bibinfo {title} {All pure bipartite entangled states
  can be self-tested},\ }\href {https://doi.org/10.1038/ncomms15485} {\bibfield
   {journal} {\bibinfo  {journal} {Nature Communications}\ }\textbf {\bibinfo
  {volume} {8}},\ \bibinfo {pages} {15485} (\bibinfo {year}
  {2017})}\BibitemShut {NoStop}%
\bibitem [{\citenamefont {Kaniewski}\ \emph {et~al.}(2019)\citenamefont
  {Kaniewski}, \citenamefont {{\v{S}}upi{\'{c}}}, \citenamefont {Tura},
  \citenamefont {Baccari}, \citenamefont {Salavrakos},\ and\ \citenamefont
  {Augusiak}}]{Jed1}%
  \BibitemOpen
  \bibfield  {author} {\bibinfo {author} {\bibfnamefont {J.}~\bibnamefont
  {Kaniewski}}, \bibinfo {author} {\bibfnamefont {I.}~\bibnamefont
  {{\v{S}}upi{\'{c}}}}, \bibinfo {author} {\bibfnamefont {J.}~\bibnamefont
  {Tura}}, \bibinfo {author} {\bibfnamefont {F.}~\bibnamefont {Baccari}},
  \bibinfo {author} {\bibfnamefont {A.}~\bibnamefont {Salavrakos}},\ and\
  \bibinfo {author} {\bibfnamefont {R.}~\bibnamefont {Augusiak}},\ }\bibfield
  {title} {\bibinfo {title} {Maximal nonlocality from maximal entanglement and
  mutually unbiased bases, and self-testing of two-qutrit quantum systems},\
  }\href {https://doi.org/10.22331/q-2019-10-24-198} {\bibfield  {journal}
  {\bibinfo  {journal} {{Quantum}}\ }\textbf {\bibinfo {volume} {3}},\ \bibinfo
  {pages} {198} (\bibinfo {year} {2019})}\BibitemShut {NoStop}%
\bibitem [{\citenamefont {Man\v{c}inska}\ \emph {et~al.}(2021)\citenamefont
  {Man\v{c}inska}, \citenamefont {Prakash},\ and\ \citenamefont
  {Schafhauser}}]{prakash}%
  \BibitemOpen
  \bibfield  {author} {\bibinfo {author} {\bibfnamefont {L.}~\bibnamefont
  {Man\v{c}inska}}, \bibinfo {author} {\bibfnamefont {J.}~\bibnamefont
  {Prakash}},\ and\ \bibinfo {author} {\bibfnamefont {C.}~\bibnamefont
  {Schafhauser}},\ }\bibfield  {title} {\bibinfo {title} {Constant-sized robust
  self-tests for states and measurements of unbounded dimension},\ }\href
  {https://arxiv.org/abs/2103.01729} {\bibfield  {journal} {\bibinfo  {journal}
  {arXiv:2103.01729}\ } (\bibinfo {year} {2021})}\BibitemShut {NoStop}%
\bibitem [{\citenamefont {Tavakoli}\ \emph {et~al.}(2021)\citenamefont
  {Tavakoli}, \citenamefont {Farkas}, \citenamefont {Rosset}, \citenamefont
  {Bancal},\ and\ \citenamefont {Kaniewski}}]{Armin1}%
  \BibitemOpen
  \bibfield  {author} {\bibinfo {author} {\bibfnamefont {A.}~\bibnamefont
  {Tavakoli}}, \bibinfo {author} {\bibfnamefont {M.}~\bibnamefont {Farkas}},
  \bibinfo {author} {\bibfnamefont {D.}~\bibnamefont {Rosset}}, \bibinfo
  {author} {\bibfnamefont {J.-D.}\ \bibnamefont {Bancal}},\ and\ \bibinfo
  {author} {\bibfnamefont {J.}~\bibnamefont {Kaniewski}},\ }\bibfield  {title}
  {\bibinfo {title} {Mutually unbiased bases and symmetric informationally
  complete measurements in {B}ell experiments},\ }\href
  {https://doi.org/10.1126/sciadv.abc3847} {\bibfield  {journal} {\bibinfo
  {journal} {Science Advances}\ }\textbf {\bibinfo {volume} {7}},\ \bibinfo
  {pages} {eabc3847} (\bibinfo {year} {2021})}\BibitemShut {NoStop}%
\bibitem [{\citenamefont {Sarkar}\ \emph {et~al.}(2021)\citenamefont {Sarkar},
  \citenamefont {Saha}, \citenamefont {Kaniewski},\ and\ \citenamefont
  {Augusiak}}]{sarkar}%
  \BibitemOpen
  \bibfield  {author} {\bibinfo {author} {\bibfnamefont {S.}~\bibnamefont
  {Sarkar}}, \bibinfo {author} {\bibfnamefont {D.}~\bibnamefont {Saha}},
  \bibinfo {author} {\bibfnamefont {J.}~\bibnamefont {Kaniewski}},\ and\
  \bibinfo {author} {\bibfnamefont {R.}~\bibnamefont {Augusiak}},\ }\bibfield
  {title} {\bibinfo {title} {Self-testing quantum systems of arbitrary local
  dimension with minimal number of measurements},\ }\href
  {https://doi.org/10.1038/s41534-021-00490-3} {\bibfield  {journal} {\bibinfo
  {journal} {npj Quantum Information}\ }\textbf {\bibinfo {volume} {7}},\
  \bibinfo {pages} {151} (\bibinfo {year} {2021})}\BibitemShut {NoStop}%
\bibitem [{\citenamefont {Sarkar}\ and\ \citenamefont
  {Augusiak}(2022)}]{sarkaro2}%
  \BibitemOpen
  \bibfield  {author} {\bibinfo {author} {\bibfnamefont {S.}~\bibnamefont
  {Sarkar}}\ and\ \bibinfo {author} {\bibfnamefont {R.}~\bibnamefont
  {Augusiak}},\ }\bibfield  {title} {\bibinfo {title} {Self-testing of
  multipartite greenberger-horne-zeilinger states of arbitrary local dimension
  with arbitrary number of measurements per party},\ }\href
  {https://doi.org/10.1103/PhysRevA.105.032416} {\bibfield  {journal} {\bibinfo
   {journal} {Phys. Rev. A}\ }\textbf {\bibinfo {volume} {105}},\ \bibinfo
  {pages} {032416} (\bibinfo {year} {2022})}\BibitemShut {NoStop}%
\bibitem [{\citenamefont {Yang}\ and\ \citenamefont
  {Navascu\'es}(2013)}]{Yang}%
  \BibitemOpen
  \bibfield  {author} {\bibinfo {author} {\bibfnamefont {T.~H.}\ \bibnamefont
  {Yang}}\ and\ \bibinfo {author} {\bibfnamefont {M.}~\bibnamefont
  {Navascu\'es}},\ }\bibfield  {title} {\bibinfo {title} {Robust self-testing
  of unknown quantum systems into any entangled two-qubit states},\ }\href
  {https://doi.org/10.1103/PhysRevA.87.050102} {\bibfield  {journal} {\bibinfo
  {journal} {Phys. Rev. A}\ }\textbf {\bibinfo {volume} {87}},\ \bibinfo
  {pages} {050102} (\bibinfo {year} {2013})}\BibitemShut {NoStop}%
\bibitem [{\citenamefont {Woodhead}\ \emph {et~al.}(2020)\citenamefont
  {Woodhead}, \citenamefont {Kaniewski}, \citenamefont {Bourdoncle},
  \citenamefont {Salavrakos}, \citenamefont {Bowles}, \citenamefont {Ac\'in},\
  and\ \citenamefont {Augusiak}}]{random1}%
  \BibitemOpen
  \bibfield  {author} {\bibinfo {author} {\bibfnamefont {E.}~\bibnamefont
  {Woodhead}}, \bibinfo {author} {\bibfnamefont {J.}~\bibnamefont {Kaniewski}},
  \bibinfo {author} {\bibfnamefont {B.}~\bibnamefont {Bourdoncle}}, \bibinfo
  {author} {\bibfnamefont {A.}~\bibnamefont {Salavrakos}}, \bibinfo {author}
  {\bibfnamefont {J.}~\bibnamefont {Bowles}}, \bibinfo {author} {\bibfnamefont
  {A.}~\bibnamefont {Ac\'in}},\ and\ \bibinfo {author} {\bibfnamefont
  {R.}~\bibnamefont {Augusiak}},\ }\bibfield  {title} {\bibinfo {title}
  {Maximal randomness from partially entangled states},\ }\href
  {https://doi.org/https://doi.org/10.1103/PhysRevResearch.2.042028} {\bibfield
   {journal} {\bibinfo  {journal} {Phys. Rev. Research}\ }\textbf {\bibinfo
  {volume} {2}},\ \bibinfo {pages} {042028} (\bibinfo {year}
  {2020})}\BibitemShut {NoStop}%
\bibitem [{\citenamefont {Renou}\ \emph {et~al.}(2018)\citenamefont {Renou},
  \citenamefont {Kaniewski},\ and\ \citenamefont {Brunner}}]{Marco}%
  \BibitemOpen
  \bibfield  {author} {\bibinfo {author} {\bibfnamefont {M.-O.}\ \bibnamefont
  {Renou}}, \bibinfo {author} {\bibfnamefont {J.}~\bibnamefont {Kaniewski}},\
  and\ \bibinfo {author} {\bibfnamefont {N.}~\bibnamefont {Brunner}},\
  }\bibfield  {title} {\bibinfo {title} {Self-testing entangled measurements in
  quantum networks},\ }\href {https://doi.org/10.1103/PhysRevLett.121.250507}
  {\bibfield  {journal} {\bibinfo  {journal} {Phys. Rev. Lett.}\ }\textbf
  {\bibinfo {volume} {121}},\ \bibinfo {pages} {250507} (\bibinfo {year}
  {2018})}\BibitemShut {NoStop}%
\bibitem [{\citenamefont {Zhou}\ \emph {et~al.}(2022)\citenamefont {Zhou},
  \citenamefont {Xu}, \citenamefont {Zhao}, \citenamefont {Zhen}, \citenamefont
  {Li}, \citenamefont {Liu},\ and\ \citenamefont {Chen}}]{JW2}%
  \BibitemOpen
  \bibfield  {author} {\bibinfo {author} {\bibfnamefont {Q.}~\bibnamefont
  {Zhou}}, \bibinfo {author} {\bibfnamefont {X.-Y.}\ \bibnamefont {Xu}},
  \bibinfo {author} {\bibfnamefont {S.}~\bibnamefont {Zhao}}, \bibinfo {author}
  {\bibfnamefont {Y.-Z.}\ \bibnamefont {Zhen}}, \bibinfo {author}
  {\bibfnamefont {L.}~\bibnamefont {Li}}, \bibinfo {author} {\bibfnamefont
  {N.-L.}\ \bibnamefont {Liu}},\ and\ \bibinfo {author} {\bibfnamefont
  {K.}~\bibnamefont {Chen}},\ }\bibfield  {title} {\bibinfo {title} {Robust
  self-testing of multipartite {G}reenberger-{H}orne-{Z}eilinger-state
  measurements in quantum networks},\ }\href
  {https://doi.org/10.1103/PhysRevA.106.042608} {\bibfield  {journal} {\bibinfo
   {journal} {Phys. Rev. A}\ }\textbf {\bibinfo {volume} {106}},\ \bibinfo
  {pages} {042608} (\bibinfo {year} {2022})}\BibitemShut {NoStop}%
\bibitem [{\citenamefont {Šupić}\ and\ \citenamefont
  {Brunner}(2022)}]{NLWEsupic}%
  \BibitemOpen
  \bibfield  {author} {\bibinfo {author} {\bibfnamefont {I.}~\bibnamefont
  {Šupić}}\ and\ \bibinfo {author} {\bibfnamefont {N.}~\bibnamefont
  {Brunner}},\ }\bibfield  {title} {\bibinfo {title} {Self-testing nonlocality
  without entanglement},\ }\href {https://arxiv.org/abs/2203.13171} {\bibfield
  {journal} {\bibinfo  {journal} {arXiv:2203.13171}\ } (\bibinfo {year}
  {2022})}\BibitemShut {NoStop}%
\bibitem [{\citenamefont {Bennett}\ \emph
  {et~al.}(1999{\natexlab{a}})\citenamefont {Bennett}, \citenamefont
  {DiVincenzo}, \citenamefont {Fuchs}, \citenamefont {Mor}, \citenamefont
  {Rains}, \citenamefont {Shor}, \citenamefont {Smolin},\ and\ \citenamefont
  {Wootters}}]{Bennett99}%
  \BibitemOpen
  \bibfield  {author} {\bibinfo {author} {\bibfnamefont {C.~H.}\ \bibnamefont
  {Bennett}}, \bibinfo {author} {\bibfnamefont {D.~P.}\ \bibnamefont
  {DiVincenzo}}, \bibinfo {author} {\bibfnamefont {C.~A.}\ \bibnamefont
  {Fuchs}}, \bibinfo {author} {\bibfnamefont {T.}~\bibnamefont {Mor}}, \bibinfo
  {author} {\bibfnamefont {E.}~\bibnamefont {Rains}}, \bibinfo {author}
  {\bibfnamefont {P.~W.}\ \bibnamefont {Shor}}, \bibinfo {author}
  {\bibfnamefont {J.~A.}\ \bibnamefont {Smolin}},\ and\ \bibinfo {author}
  {\bibfnamefont {W.~K.}\ \bibnamefont {Wootters}},\ }\bibfield  {title}
  {\bibinfo {title} {Quantum nonlocality without entanglement},\ }\href
  {https://doi.org/10.1103/PhysRevA.59.1070} {\bibfield  {journal} {\bibinfo
  {journal} {Phys. Rev. A}\ }\textbf {\bibinfo {volume} {59}},\ \bibinfo
  {pages} {1070} (\bibinfo {year} {1999}{\natexlab{a}})}\BibitemShut {NoStop}%
\bibitem [{\citenamefont {Bennett}\ \emph
  {et~al.}(1999{\natexlab{b}})\citenamefont {Bennett}, \citenamefont
  {DiVincenzo}, \citenamefont {Mor}, \citenamefont {Shor}, \citenamefont
  {Smolin},\ and\ \citenamefont {Terhal}}]{Bennett99-1}%
  \BibitemOpen
  \bibfield  {author} {\bibinfo {author} {\bibfnamefont {C.~H.}\ \bibnamefont
  {Bennett}}, \bibinfo {author} {\bibfnamefont {D.~P.}\ \bibnamefont
  {DiVincenzo}}, \bibinfo {author} {\bibfnamefont {T.}~\bibnamefont {Mor}},
  \bibinfo {author} {\bibfnamefont {P.~W.}\ \bibnamefont {Shor}}, \bibinfo
  {author} {\bibfnamefont {J.~A.}\ \bibnamefont {Smolin}},\ and\ \bibinfo
  {author} {\bibfnamefont {B.~M.}\ \bibnamefont {Terhal}},\ }\bibfield  {title}
  {\bibinfo {title} {Unextendible product bases and bound entanglement},\
  }\href {https://doi.org/10.1103/PhysRevLett.82.5385} {\bibfield  {journal}
  {\bibinfo  {journal} {Phys. Rev. Lett.}\ }\textbf {\bibinfo {volume} {82}},\
  \bibinfo {pages} {5385} (\bibinfo {year} {1999}{\natexlab{b}})}\BibitemShut
  {NoStop}%
\bibitem [{\citenamefont {Augusiak}\ \emph {et~al.}(2011)\citenamefont
  {Augusiak}, \citenamefont {Stasi\ifmmode~\acute{n}\else \'{n}\fi{}ska},
  \citenamefont {Hadley}, \citenamefont {Korbicz}, \citenamefont {Lewenstein},\
  and\ \citenamefont {Ac\'{\i}n}}]{Augusiak_PhysRevLett.107.070401}%
  \BibitemOpen
  \bibfield  {author} {\bibinfo {author} {\bibfnamefont {R.}~\bibnamefont
  {Augusiak}}, \bibinfo {author} {\bibfnamefont {J.}~\bibnamefont
  {Stasi\ifmmode~\acute{n}\else \'{n}\fi{}ska}}, \bibinfo {author}
  {\bibfnamefont {C.}~\bibnamefont {Hadley}}, \bibinfo {author} {\bibfnamefont
  {J.~K.}\ \bibnamefont {Korbicz}}, \bibinfo {author} {\bibfnamefont
  {M.}~\bibnamefont {Lewenstein}},\ and\ \bibinfo {author} {\bibfnamefont
  {A.}~\bibnamefont {Ac\'{\i}n}},\ }\bibfield  {title} {\bibinfo {title}
  {{B}ell inequalities with no quantum violation and unextendable product
  bases},\ }\href {https://doi.org/10.1103/PhysRevLett.107.070401} {\bibfield
  {journal} {\bibinfo  {journal} {Phys. Rev. Lett.}\ }\textbf {\bibinfo
  {volume} {107}},\ \bibinfo {pages} {070401} (\bibinfo {year}
  {2011})}\BibitemShut {NoStop}%
\bibitem [{\citenamefont {Šupić}\ \emph {et~al.}(2022)\citenamefont
  {Šupić}, \citenamefont {Bowles}, \citenamefont {Renou}, \citenamefont
  {Acín},\ and\ \citenamefont {Hoban}}]{Allst}%
  \BibitemOpen
  \bibfield  {author} {\bibinfo {author} {\bibfnamefont {I.}~\bibnamefont
  {Šupić}}, \bibinfo {author} {\bibfnamefont {J.}~\bibnamefont {Bowles}},
  \bibinfo {author} {\bibfnamefont {M.-O.}\ \bibnamefont {Renou}}, \bibinfo
  {author} {\bibfnamefont {A.}~\bibnamefont {Acín}},\ and\ \bibinfo {author}
  {\bibfnamefont {M.~J.}\ \bibnamefont {Hoban}},\ }\bibfield  {title} {\bibinfo
  {title} {Quantum networks self-test all entangled states},\ }\href
  {https://arxiv.org/abs/2201.05032} {\bibfield  {journal} {\bibinfo  {journal}
  {arXiv:2201.05032}\ } (\bibinfo {year} {2022})}\BibitemShut {NoStop}%
\bibitem [{\citenamefont {Baccari}\ \emph
  {et~al.}(2020{\natexlab{a}})\citenamefont {Baccari}, \citenamefont
  {Augusiak}, \citenamefont {\ifmmode \check{S}\else
  \v{S}\fi{}upi\ifmmode~\acute{c}\else \'{c}\fi{}},\ and\ \citenamefont
  {Ac\'{\i}n}}]{subspaces1}%
  \BibitemOpen
  \bibfield  {author} {\bibinfo {author} {\bibfnamefont {F.}~\bibnamefont
  {Baccari}}, \bibinfo {author} {\bibfnamefont {R.}~\bibnamefont {Augusiak}},
  \bibinfo {author} {\bibfnamefont {I.}~\bibnamefont {\ifmmode \check{S}\else
  \v{S}\fi{}upi\ifmmode~\acute{c}\else \'{c}\fi{}}},\ and\ \bibinfo {author}
  {\bibfnamefont {A.}~\bibnamefont {Ac\'{\i}n}},\ }\bibfield  {title} {\bibinfo
  {title} {Device-independent certification of genuinely entangled subspaces},\
  }\href {https://doi.org/10.1103/PhysRevLett.125.260507} {\bibfield  {journal}
  {\bibinfo  {journal} {Phys. Rev. Lett.}\ }\textbf {\bibinfo {volume} {125}},\
  \bibinfo {pages} {260507} (\bibinfo {year} {2020}{\natexlab{a}})}\BibitemShut
  {NoStop}%
\bibitem [{\citenamefont {Fr\'erot}\ and\ \citenamefont
  {Ac\'{\i}n}(2021)}]{subspaces2}%
  \BibitemOpen
  \bibfield  {author} {\bibinfo {author} {\bibfnamefont {I.}~\bibnamefont
  {Fr\'erot}}\ and\ \bibinfo {author} {\bibfnamefont {A.}~\bibnamefont
  {Ac\'{\i}n}},\ }\bibfield  {title} {\bibinfo {title} {Coarse-grained
  self-testing},\ }\href {https://doi.org/10.1103/PhysRevLett.127.240401}
  {\bibfield  {journal} {\bibinfo  {journal} {Phys. Rev. Lett.}\ }\textbf
  {\bibinfo {volume} {127}},\ \bibinfo {pages} {240401} (\bibinfo {year}
  {2021})}\BibitemShut {NoStop}%
\bibitem [{\citenamefont {Parthasarathy}(2004)}]{CES1}%
  \BibitemOpen
  \bibfield  {author} {\bibinfo {author} {\bibfnamefont {K.~R.}\ \bibnamefont
  {Parthasarathy}},\ }\bibfield  {title} {\bibinfo {title} {On the maximal
  dimension of a completely entangled subspace for finite level quantum
  systems},\ }\href {https://doi.org/doi.org/10.1007/BF02829441} {\bibfield
  {journal} {\bibinfo  {journal} {Proc. Math. Sci.}\ }\textbf {\bibinfo
  {volume} {114}},\ \bibinfo {pages} {365–374} (\bibinfo {year}
  {2004})}\BibitemShut {NoStop}%
\bibitem [{\citenamefont {Walgate}\ and\ \citenamefont {Scott}(2008)}]{CES2}%
  \BibitemOpen
  \bibfield  {author} {\bibinfo {author} {\bibfnamefont {J.}~\bibnamefont
  {Walgate}}\ and\ \bibinfo {author} {\bibfnamefont {A.~J.}\ \bibnamefont
  {Scott}},\ }\bibfield  {title} {\bibinfo {title} {Generic local
  distinguishability and completely entangled subspaces},\ }\href
  {https://doi.org/10.1088/1751-8113/41/37/375305} {\bibfield  {journal}
  {\bibinfo  {journal} {J. Phys. A}\ }\textbf {\bibinfo {volume} {41}},\
  \bibinfo {pages} {375305} (\bibinfo {year} {2008})}\BibitemShut {NoStop}%
\bibitem [{\citenamefont {Horodecki}\ \emph {et~al.}(1998)\citenamefont
  {Horodecki}, \citenamefont {Horodecki},\ and\ \citenamefont
  {Horodecki}}]{HorodeckiBE}%
  \BibitemOpen
  \bibfield  {author} {\bibinfo {author} {\bibfnamefont {M.}~\bibnamefont
  {Horodecki}}, \bibinfo {author} {\bibfnamefont {P.}~\bibnamefont
  {Horodecki}},\ and\ \bibinfo {author} {\bibfnamefont {R.}~\bibnamefont
  {Horodecki}},\ }\bibfield  {title} {\bibinfo {title} {Mixed-state
  entanglement and distillation: Is there a ``bound'' entanglement in
  nature?},\ }\href {https://doi.org/10.1103/PhysRevLett.80.5239} {\bibfield
  {journal} {\bibinfo  {journal} {Phys. Rev. Lett.}\ }\textbf {\bibinfo
  {volume} {80}},\ \bibinfo {pages} {5239} (\bibinfo {year}
  {1998})}\BibitemShut {NoStop}%
\bibitem [{\citenamefont {Baccari}\ \emph
  {et~al.}(2020{\natexlab{b}})\citenamefont {Baccari}, \citenamefont
  {Augusiak}, \citenamefont {\ifmmode \check{S}\else
  \v{S}\fi{}upi\ifmmode~\acute{c}\else \'{c}\fi{}}, \citenamefont {Tura},\ and\
  \citenamefont {Ac\'{\i}n}}]{Flavio}%
  \BibitemOpen
  \bibfield  {author} {\bibinfo {author} {\bibfnamefont {F.}~\bibnamefont
  {Baccari}}, \bibinfo {author} {\bibfnamefont {R.}~\bibnamefont {Augusiak}},
  \bibinfo {author} {\bibfnamefont {I.}~\bibnamefont {\ifmmode \check{S}\else
  \v{S}\fi{}upi\ifmmode~\acute{c}\else \'{c}\fi{}}}, \bibinfo {author}
  {\bibfnamefont {J.}~\bibnamefont {Tura}},\ and\ \bibinfo {author}
  {\bibfnamefont {A.}~\bibnamefont {Ac\'{\i}n}},\ }\bibfield  {title} {\bibinfo
  {title} {Scalable {B}ell inequalities for qubit graph states and robust
  self-testing},\ }\href {https://doi.org/10.1103/PhysRevLett.124.020402}
  {\bibfield  {journal} {\bibinfo  {journal} {Phys. Rev. Lett.}\ }\textbf
  {\bibinfo {volume} {124}},\ \bibinfo {pages} {020402} (\bibinfo {year}
  {2020}{\natexlab{b}})}\BibitemShut {NoStop}%
\bibitem [{Sup()}]{SupMat}%
  \BibitemOpen
  \bibinfo {note} {See Supplemental Material @ for the proof of the theorems,
  which includes Ref. \cite{murphy2014c}}\BibitemShut {NoStop}%
\bibitem [{\citenamefont {Mermin}(1990)}]{Mermin}%
  \BibitemOpen
  \bibfield  {author} {\bibinfo {author} {\bibfnamefont {N.~D.}\ \bibnamefont
  {Mermin}},\ }\bibfield  {title} {\bibinfo {title} {Extreme quantum
  entanglement in a superposition of macroscopically distinct states},\ }\href
  {https://doi.org/10.1103/PhysRevLett.65.1838} {\bibfield  {journal} {\bibinfo
   {journal} {Phys. Rev. Lett.}\ }\textbf {\bibinfo {volume} {65}},\ \bibinfo
  {pages} {1838} (\bibinfo {year} {1990})}\BibitemShut {NoStop}%
\bibitem [{\citenamefont {Sarkar}\ \emph {et~al.}(2025)\citenamefont {Sarkar},
  \citenamefont {Trillo}, \citenamefont {Renou},\ and\ \citenamefont
  {Augusiak}}]{sarkar2025}%
  \BibitemOpen
  \bibfield  {author} {\bibinfo {author} {\bibfnamefont {S.}~\bibnamefont
  {Sarkar}}, \bibinfo {author} {\bibfnamefont {D.}~\bibnamefont {Trillo}},
  \bibinfo {author} {\bibfnamefont {M.~O.}\ \bibnamefont {Renou}},\ and\
  \bibinfo {author} {\bibfnamefont {R.}~\bibnamefont {Augusiak}},\ }\href
  {https://arxiv.org/abs/2503.09724} {\bibinfo {title} {Gap between quantum
  theory based on real and complex numbers is arbitrarily large}} (\bibinfo
  {year} {2025}),\ \Eprint {https://arxiv.org/abs/2503.09724} {arXiv:2503.09724
  [quant-ph]} \BibitemShut {NoStop}%
\bibitem [{\citenamefont {Sarkar}\ \emph {et~al.}(2024)\citenamefont {Sarkar},
  \citenamefont {Alexandre C.~Orthey},\ and\ \citenamefont
  {Augusiak}}]{sarkar2024universal}%
  \BibitemOpen
  \bibfield  {author} {\bibinfo {author} {\bibfnamefont {S.}~\bibnamefont
  {Sarkar}}, \bibinfo {author} {\bibfnamefont {J.}~\bibnamefont {Alexandre
  C.~Orthey}},\ and\ \bibinfo {author} {\bibfnamefont {R.}~\bibnamefont
  {Augusiak}},\ }\href {https://arxiv.org/abs/2312.04405} {\bibinfo {title} {A
  universal scheme to self-test any quantum state and extremal measurement}}
  (\bibinfo {year} {2024}),\ \Eprint {https://arxiv.org/abs/2312.04405}
  {arXiv:2312.04405 [quant-ph]} \BibitemShut {NoStop}%
\bibitem [{\citenamefont {Murphy}(2014)}]{murphy2014c}%
  \BibitemOpen
  \bibfield  {author} {\bibinfo {author} {\bibfnamefont {G.}~\bibnamefont
  {Murphy}},\ }\href {https://books.google.be/books?id=omviBQAAQBAJ} {\emph
  {\bibinfo {title} {C*-Algebras and Operator Theory}}}\ (\bibinfo  {publisher}
  {Elsevier Science},\ \bibinfo {year} {2014})\BibitemShut {NoStop}%
\end{thebibliography}
\providecommand{\noopsort}[1]{}\providecommand{\singleletter}[1]{#1}%

\end{document}